\documentclass[10pt]{article}

\usepackage{hyperref}
\usepackage{natbib}
\usepackage{fullpage}
\usepackage{graphicx}
\usepackage{graphics}
\usepackage{mdwlist}
\usepackage[top=2.5cm, bottom=2.5cm, left=2.5cm, right=2.5cm]{geometry}

\usepackage{amssymb,amsmath}
\usepackage{xspace}
\usepackage[linesnumbered, ruled]{algorithm2e} \DontPrintSemicolon
\usepackage{caption}
\usepackage[caption=false]{subfig}
\usepackage{graphicx}
\usepackage{mdwlist}
\usepackage{booktabs}
\usepackage{color}
\usepackage{xfrac}
\usepackage{enumitem}
\usepackage{times}

\begin{document}

\title{FURL: Fixed-memory and Uncertainty Reducing Local Triangle Counting for Graph Streams}

\author{ {Minsoo Jung, Sunmin Lee, Yongsub Lim, and U Kang} \\
	    %Dept. of Computer Science and Engineering\\
	    Seoul National University\\
    }
	    %{\tt minsoojung@snu.ac.kr}
	    	%
\date{}

\let\oldnl\nl% Store \nl in \oldnl
\newcommand{\nonl}{\renewcommand{\nl}{\let\nl\oldnl}}% Remove line number for one line

\newcommand{\expect}[1]{\mathbb{E}\left[#1\right]}
\newcommand{\var}[1]{\mathrm{Var}\left[#1\right]}
\newcommand{\cov}[1]{\mathrm{Cov}\left[#1\right]}

\newcommand{\etal}{et al.\xspace}
\newcommand{\eg}{e.g.,\xspace}
\newcommand{\ie}{i.e.,\xspace}

\newcommand{\natset}{\mathbb{N}}
\newcommand{\realset}{\mathbb{R}}
\newcommand{\prob}[1]{\Pr\left[ #1 \right]}

\newcommand{\algref}[1]{Algorithm~\ref{#1}}
\newcommand{\figref}[1]{\figurename~\ref{#1}}
\newcommand{\tabref}[1]{\tablename~\ref{#1}}
\newcommand{\lemref}[1]{Lemma~\ref{#1}}
\newcommand{\eqnumref}[1]{(\ref{#1})}
\newcommand{\secref}[1]{Section~\ref{#1}}
\newcommand{\obsref}[1]{Observation~\ref{#1}}
\newcommand{\lineref}[1]{Line~\ref{#1}}
\newcommand{\linesref}[1]{Lines~\ref{#1}}

\newcommand{\chng}[1]{\textcolor{black}{#1}}
\newcommand{\mchng}[1]{\textcolor{black}{#1}}

\newcommand{\hide}[1]{}
\newcommand{\reminder}[1]{\textcolor{red}{\textsf{\small #1}}}

\newtheorem{lemma}{Lemma}
\newtheorem{proof}{Proof}
\newtheorem{observation}{Observation}
\newtheorem{corollary}{Corollary}

\newcommand{\doulion}{{\sc Doulion}\xspace}

\newcommand{\kutzcov}{{\sc KP}\xspace}

\newcommand{\mre}{{\sf MRE}\xspace}
\newcommand{\pcc}{{\sf PCC}\xspace}

\newcommand{\mascot}{{\sc Mascot}\xspace}
\newcommand{\mascotrm}{MASCOT\xspace}
\newcommand{\mascotc}{{\sc Mascot-C}\xspace}
\newcommand{\mascotcrm}{MASCOT-C\xspace}

\newcommand{\triest}{{\sc Triest}\xspace}
\newcommand{\triestrm}{TRIEST\xspace}

\newcommand{\method}{{\sc Furl}\xspace}
\newcommand{\methodrm}{FURL\xspace}

\newcommand{\methodc}{{\sc Furl-S}\xspace}
\newcommand{\methodcrm}{FURL-S\xspace}
\newcommand{\methodm}{{\sc Furl-M}\xspace}
\newcommand{\methodmrm}{FURL-M\xspace}
\newcommand{\methodb}{{\sc Furl-M$_\textsc{B}$}\xspace}
\newcommand{\methodbrb}{FURL-M$_\text{B}$\xspace}
\newcommand{\methodw}{{\sc Furl-M$_\textsc{W}$}\xspace}
\newcommand{\methodwrw}{FURL-M$_\text{W}$\xspace}

\newcommand{\methodcx}{{\sc Furl-SX}\xspace}
\newcommand{\methodcxrm}{FURL-SX\xspace}
\newcommand{\methodmx}{{\sc Furl-MX}\xspace}
\newcommand{\methodmxrm}{FURL-MX\xspace}
\newcommand{\methodbx}{{\sc Furl-MX$_\textsc{B}$}\xspace}
\newcommand{\methodbxrb}{FURL-MX$_\text{B}$\xspace}
\newcommand{\methodwx}{{\sc Furl-MX$_\textsc{W}$}\xspace}
\newcommand{\methodwxrw}{FURL-MX$_\text{W}$\xspace}

\newcommand{\youtube}{\textsf{Youtube}\xspace}
\newcommand{\hudong}{\textsf{Hudong}\xspace}
\newcommand{\pokec}{\textsf{Pokec}\xspace}
\newcommand{\skitter}{\textsf{Skitter}\xspace}
\newcommand{\yahoo}{\textsf{YahooMsg}\xspace}
\newcommand{\webgraph}{\textsf{WebGraph}\xspace}
\newcommand{\twitter}{\textsf{Twitter}\xspace}
\newcommand{\facebook}{\textsf{Facebook}\xspace}
\newcommand{\actor}{\textsf{Actor}\xspace}
\newcommand{\baidu}{\textsf{Baidu}\xspace}
\newcommand{\dblp}{\textsf{DBLP-M}\xspace}

\newcommand{\chngwww}[1]{\textcolor{black}{#1}} 

\newcommand{\qed}{\noindent $\square$} 

\maketitle
\begin{abstract}
    
How can we accurately estimate local triangles for all nodes in simple and multigraph streams?
Local triangle counting in a graph stream is one of the most fundamental tasks in graph mining with important applications including anomaly detection and social network analysis.
Although there have been several local triangle counting methods in a graph stream, their estimation has a large variance which results in low accuracy, and they
do not consider multigraph streams which have duplicate edges.
In this paper, we propose \method, an accurate local triangle counting method for simple and multigraph streams.
\method improves the accuracy by reducing a variance through biased estimation and handles duplicate edges for multigraph streams. 
Also, \method handles a stream of any size by using a fixed amount of memory.
Experimental results show that \method outperforms the state-of-the-art method in accuracy and performs well in multigraph streams.
In addition, we report interesting patterns discovered from real graphs by \method, which include unusual local structures in a user communication network and a core-periphery structure in the Web.
\end{abstract}

\section{Introduction}
% rhetorical question
How can we accurately estimate local triangles for all nodes in simple and multigraph streams?
The \emph{local triangle counting} problem is to count the number of triangles containing each node in a graph, and has been extensively studied because of its wide and important applications.
For instance, it has been used for 
social role identification of a user~\cite{WelserGFS07, chou2010discovering}, content quality evaluation~\cite{BecchettiBCG10}, 
data-driven anomaly detection~\cite{BecchettiBCG10, YangWWGZD11, LimK15}, 
community detection~\cite{BerryHLP11, suri2011counting},
motif detection~\cite{MilEtAl02}, 
clustering ego-networks~\cite{epasto2015ego}, and uncovering hidden thematic layers~\cite{EckmannM02}.

% limitations of existing methods
Although a number of local triangle counting methods have been successfully applied to graph streams, there are two challenging issues.
First, their large variance causes low accuracy.
In a graph stream model where edges continuously arrive, 
only one trial is allowed and a large variance causes large difference between estimated and true local triangle counts.
Thus previous local triangle counting methods do not produce stable results and show bad worst case performance.
Second, recent real world graph streams contain duplicate edges, i.e., they are multigraph streams.
Examples include a communication network in Internet, a phone call history, SNS messages like tweets, an email network, etc.
In such environments it is crucial to carefully handle duplicate edges in local triangle counting.

% main idea + contribution
In this paper, we propose \method, 
an accurate local triangle counting method for simple and multigraph streams.
\method guarantees the fixed amount of memory usage, regardless of the size of a graph stream.
\method has two main versions: \methodcx and \methodmx for simple and multigraph streams, respectively. 
\methodcx improves the accuracy of the previous state-of-the-art by reducing a variance through biased estimation. 
\methodmx, the first algorithm for local triangle counting in multigraph streams, carefully handles multigraph streams by using reservoir sampling with random hash to sample distinct items uniformly at random, and reduces a variance of estimation. 

\begin{figure*}
	\begin{center}
		\subfloat[\youtube]
		{\includegraphics[width=0.235\textwidth,natwidth=610,natheight=642]{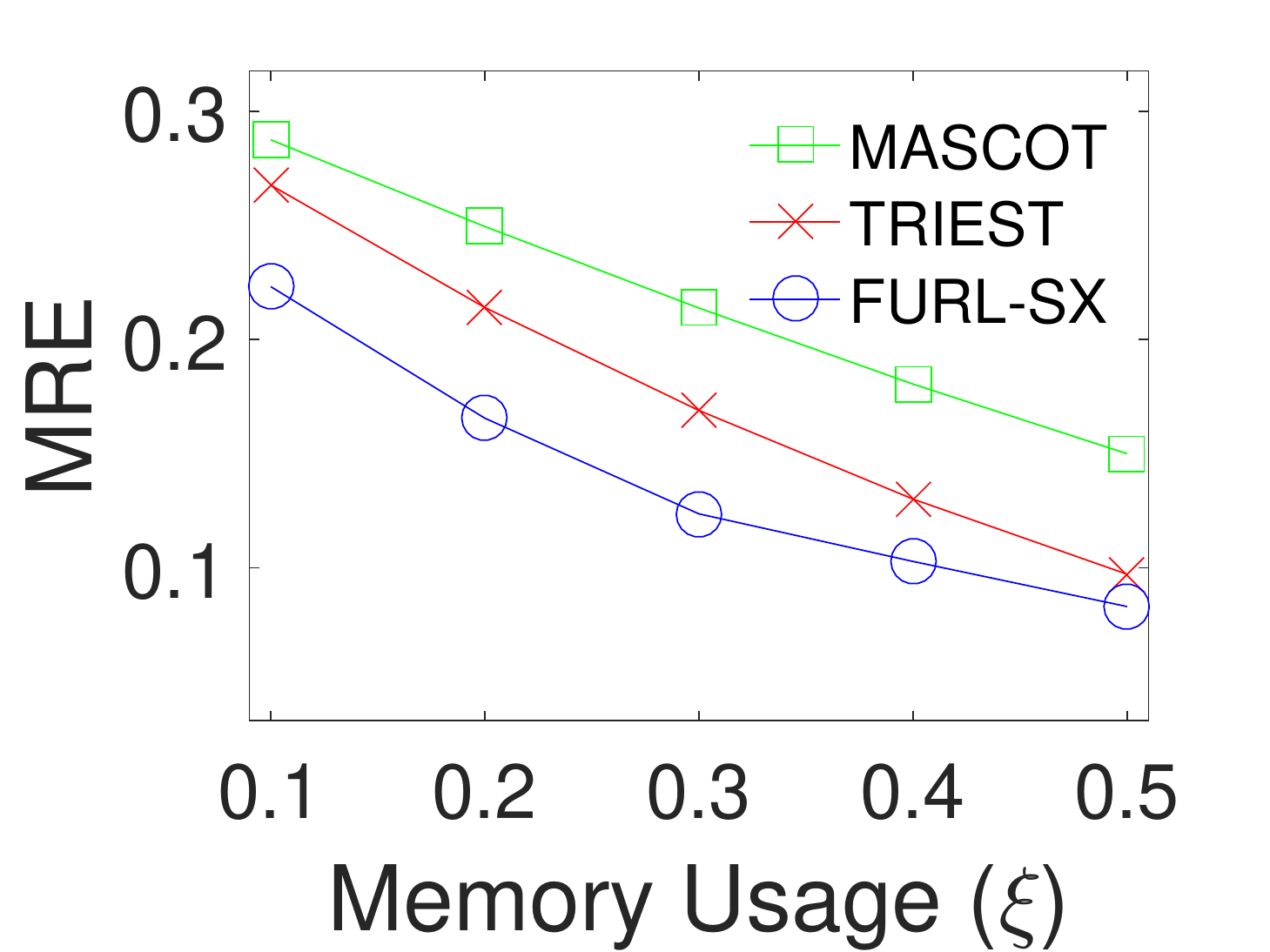} \label{fig:sim_err_youtube}} ~~~~~
		\subfloat[\pokec]
		{\includegraphics[width=0.235\textwidth,natwidth=610,natheight=642]{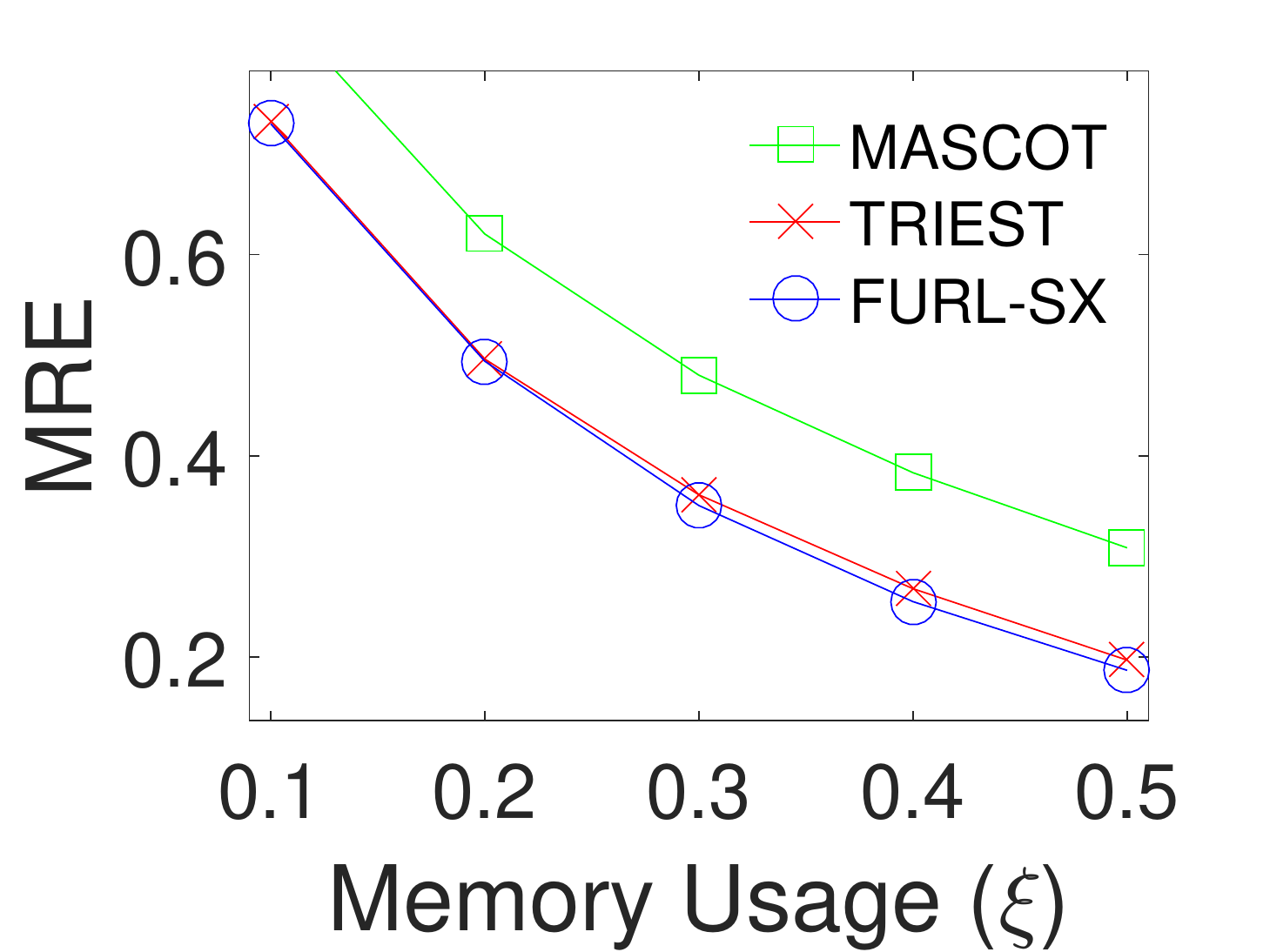} \label{fig:sim_err_pokec}} ~~~~~
		\subfloat[\skitter]
		{\includegraphics[width=0.235\textwidth,natwidth=610,natheight=642]{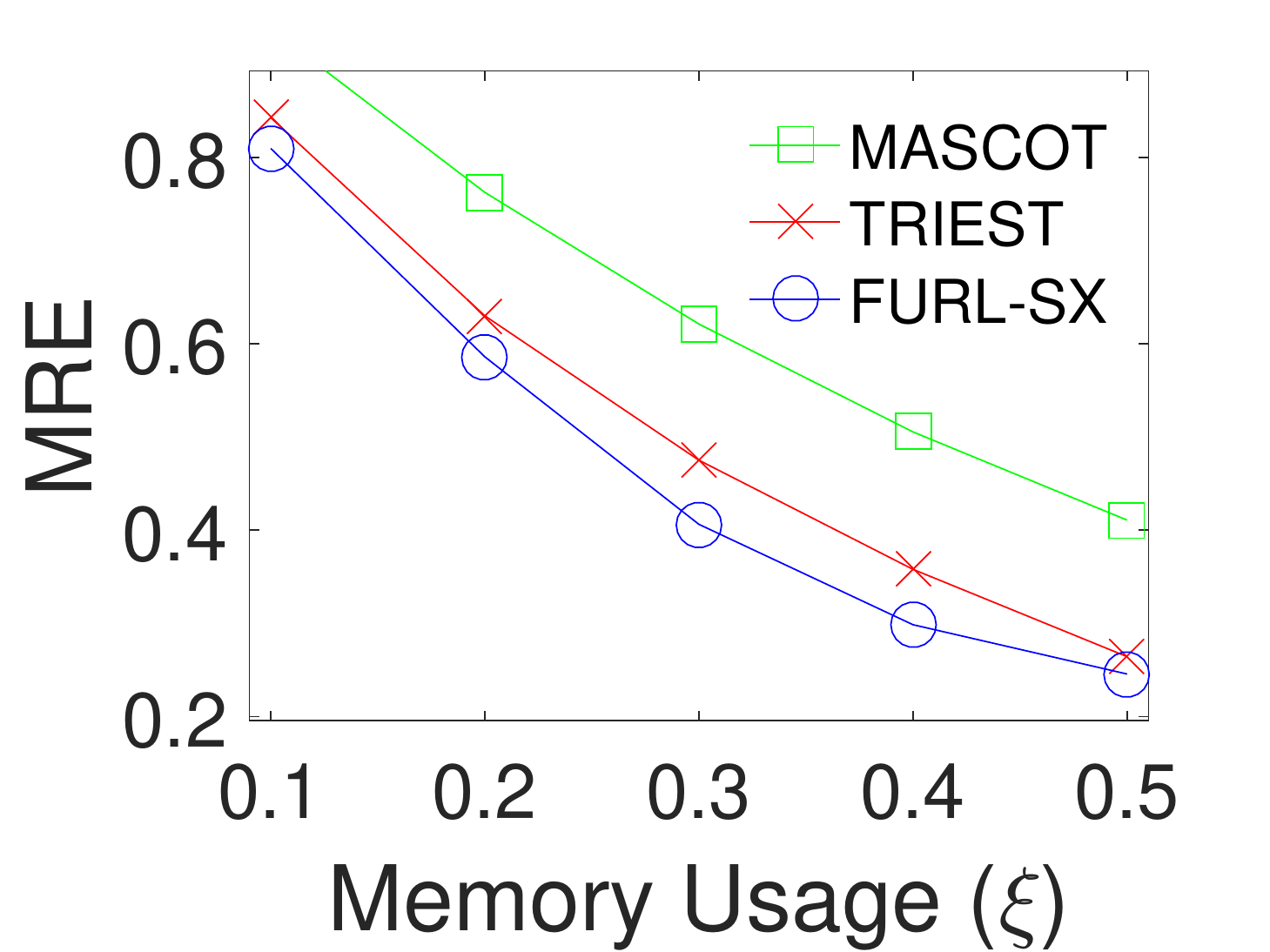} \label{fig:sim_cor_skitter}} ~~~~~
		\subfloat[\hudong]
		{\includegraphics[width=0.235\textwidth,natwidth=610,natheight=642]{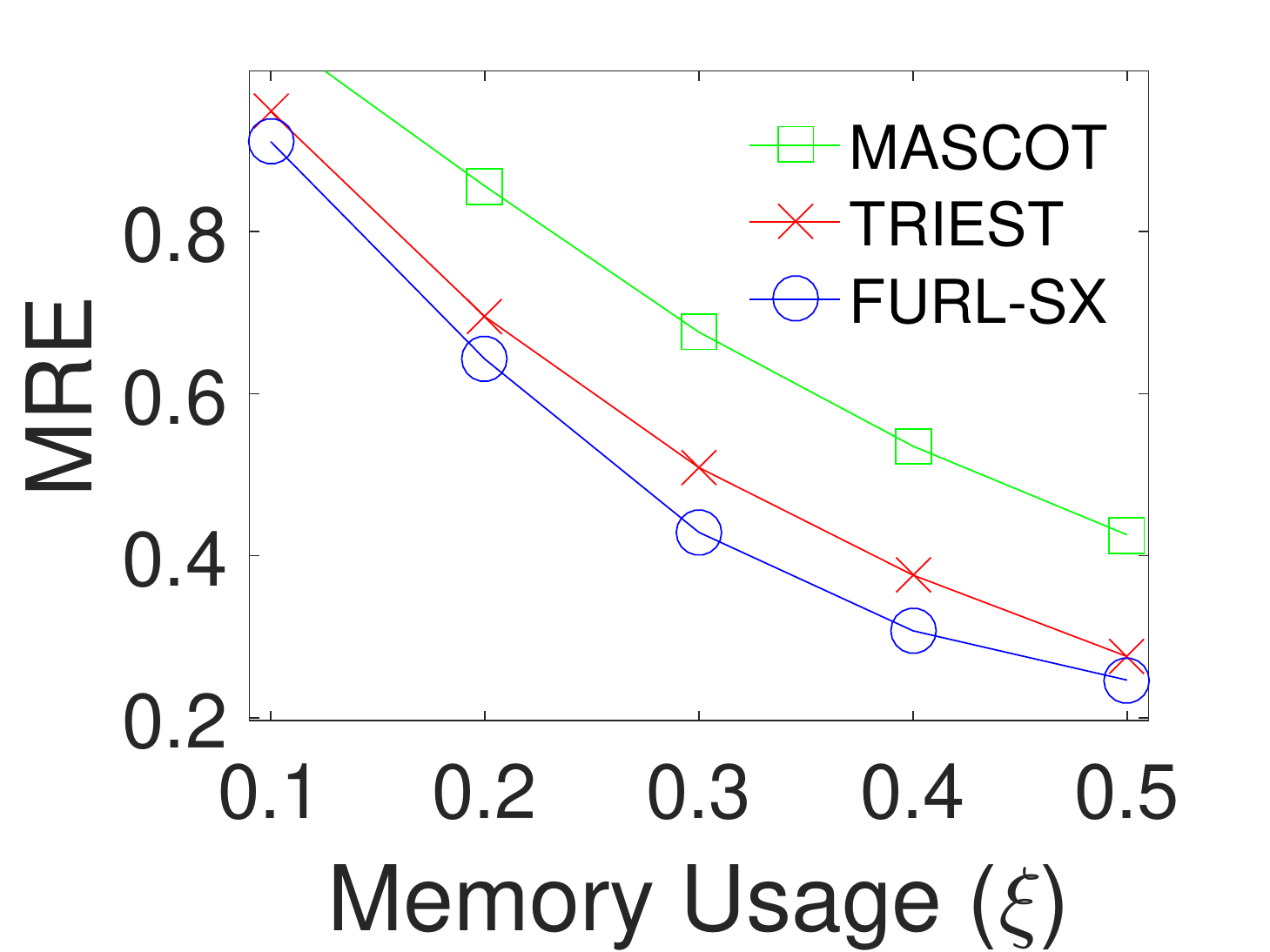} \label{fig:sim_cor_hudong}}
		 \vspace{-2mm}
		\caption{ Mean of relative error (MRE) vs.~memory usage of \methodcx and \chngwww{two competing methods \mascot and \triest} for local triangle counting in simple graph stream.
			We choose 0.7, 0.1, 0.4, and 0.4 as $\delta$ in \youtube, \pokec, \skitter, and \hudong, respectively.
			\methodcx gives the minimum error for a given memory usage. 
		}
		\label{fig:methodcx result}
	\end{center}
\vspace{-5mm}
\end{figure*}

%please fix below paragraph!
%Compared to previous methods~\cite{KutzkovP13,LimK15}, \methodcx has two main advantages: accuracy and memory-efficiency.
%First, \methodcx provides the best accuracy by decreasing the variance (i.e., increasing the probability that the estimation is close to the truth) of the estimation. Low variance estimation is especially useful in graph stream setting where each element can be seen only once.
%Second, \methodcx uses a fixed memory space, regardless of an input graph stream size. This characteristic enables to analyze a graph stream of any size, unlike previous methods~\cite{KutzkovP13,LimK15} whose memory usage is proportional to the graph size and thus eventually runs out of memory.
%Furthermore, \methodcx uses the fixed memory space fully long before the entire graph stream is received, and thus provides better accuracy than competitors under the same memory usage.

%please fix below sentence!
%Experimental results show that \methodcx outperforms the state-of-the-art methodcx as well as two basic algorithms (see Figure~\ref{fig:comp_err}).

% list of contributions
Our contributions are summarized as follows.

\begin{itemize*}
	\item \textbf{Algorithm.}
	We propose \method, an accurate local triangle counting algorithm for simple and multigraph streams. \method uses a fixed amount of memory, and thus it can handle streams of any sizes.
	\method improves the accuracy of previous simple graph stream algorithms by decreasing variance; furthermore, \method handles multigraph streams by carefully updating triangle counts.
	We give theoretical results on the accuracy bound of \method.
	\item \textbf{Performance.}
	We demonstrate that \methodcx, a version of \method for simple graph stream, outperforms the state-of-the-art method as shown in Figure~\ref{fig:methodcx result}. We also show that \methodmx, a version of \method for multigraph stream, performs well in multigraph streams.
	%
	%		Additionally we explore the effect of varying its parameters.
	%		% need to modify (experiment not yet)
	%		\methodcx and \methodmx is almost invariant on the interval of merging the past and the current estimations, and reduces the overall error when moderately weighting the past one.
	\item \textbf{Discovery.}
	We present interesting patterns discovered by applying \method to large real graph data, which include egonetworks forming a near clique/star in a user communication network, and a core-periphery structure in the Web.
\end{itemize*}

The code and datasets used in the paper are available in~\url{http://datalab.snu.ac.kr/furl}.
The rest of this paper is organized as follows.
\secref{sec:background} gives related works and preliminaries.
\secref{sec:method} describes our proposed method \method.
In~\secref{sec:experiment}, we evaluate the performance of \method and give discovery results.
We conclude in~\secref{sec:conclusion}.
\tabref{tab:symbol} shows the symbols used in this paper.

\begin{table}[t!]
	\caption{Table of symbols.}
	\vspace{-5mm}
	\label{tab:symbol}
	\begin{center}\scriptsize
		\begin{tabular}{ll} \toprule
			\textbf{Symbol} & \textbf{Description} \\ \midrule
			$S$ & Graph stream\\
			$V,E$ & Set of nodes and edges \\
			$u,v,w$ & Nodes \\
			$e$ & Edges \\
			$N_u$ & Set of neighbors of $u$ \\
			$N_{uv}$ & Set of common neighbors of $u$ and $v$ \\
			$n,m$ & Number of nodes and edges \\
			$O_e$ & Number of occurence of an edge $e$ in multigraph stream \\
			$u(T)$ & Number of unique edges in a stream at time $T$ \\
			$D$ & Buffer \\
			$|D|$ & Number of elements stored in buffer $D$ \\
			$M$ & Buffer size \\
			$h_{max}$ & Maximum hash value in a stream \\
			$D_{max}$ & Edge whose hash value is $h_{max}$ \\
			$B$ & Current bucket \\
			$b_\lambda$ & Bucket that $\lambda$ appears \\
			$J$ & Bucket size \\
		%	$r$ & Relative bucket size \\
			$\delta$ & Decaying factor for past estimations \\
			$\lambda$ & Triangle \\
			$T, t$ & Current time, and time  \\
			$T_\lambda$ & Time that a triangle $\lambda$ is counted \\
			$T_M$ & Last time when local triangle estimation equals the true triangle count\\
			$\varDelta_u$ & True local triangle count of node $u$ \\
			$c_u$ & Estimated local triangle count of node $u$ in \methodc and \methodm \\
			$\tau_u$ & Estimated local triangle count of node $u$ in \methodcx and \methodmx \\
			$q_\lambda$ & Triangle estimation weight of a triangle  $\lambda$ \\
			$X_\lambda$ & Estimated triangle counts by \methodc and \methodm \\
			$Y_\lambda$ & Estimated triangle counts by \methodcx and \methodmx \\
			
			\bottomrule
		\end{tabular}
	\end{center}
	\vspace{-0.9cm}
\end{table}

\section{Related Works}
\label{sec:background}

In this section, we present related works about triangle counting and reservoir sampling in a data stream.

\subsection{Global Triangle Counting}
There have been works on the counting of the number of triangles. Usually, existing algorithms for global triangle counting are in-memory algorithms, whose performance cannot scale up with massive graphs.

To scale up with massive graphs, various approaches have been proposed. Chu and Cheng~\cite{ChuC11} developed an I/O-efficient algorithm. Their algorithm avoids random disk access by partitioning the graph into components and processing them in parallel. Hu et al.~\cite{HuTC13} devised an I/O and CPU efficient algorithm whose I/O complexity is $O(m^2/(MB))$, where $m$, $M$ and $B$ are the number of edges, the size of internal memory, and the size of block, respectively.
Thereafter, Pagh and Silvestri~\cite{PaghS14} improved the minimum I/O complexity to $O(m^{3/2}/B\sqrt{M})$ in expectation.
Cohen~\cite{Cohen09} proposed the first MapReduce algorithm. Suri and Vassilvitskii~\cite{SuriV11} presented a graph partitioning based algorithm for the MapReduce framework.
Later, the algorithm is improved by Park and Chung~\cite{ParkC13}. Recently, there have been more parallel algorithms~\cite{ArifuzzamanKM13,KimHLPY14,ParkSKP14}.

There have been considerable studies on triangle counting in a graph stream, beginning with Bar-Yossef et al.~\cite{Bar-YossefKS02}. After that, Jowhari et al.~\cite{JowhariG05} improved the space usage. Buriol et al.~\cite{BuriolFLMS06} designed two algorithms which scale very well in a stream of edges. Each algorithm respectively takes constant and $O(\log(n(1+sn/m)))$ expected update time, where $s$ is the space requirement and $n$ is the number of nodes.
Tsourakakis et al.~\cite{TsourakakisKMF09} proposed a one pass algorithm based on edge sampling to count triangles of a graph. This sampling technique was generalized by~\cite{AhmedDNK14} for counting arbitrary subgraphs.
Jha et al.~\cite{JhaSP13} presented a space efficient algorithm to estimate the number of total triangles with a single pass to an input graph.
Kane et al.~\cite{KaneMSS12} provided a sketch based streaming algorithm for estimating the number of arbitrary subgraphs with a constant size in a large graph. Afterwards, Pavan et al.~\cite{PavanTTSW13} made the space and time complexity better.
Note that global triangle counting is different from local triangle counting, the focus of this paper, to compute the number of triangles of all individual nodes.

\subsection{Local Triangle Counting}
For local triangle counting, a simple and fast algorithm is to get $A^3$ for an adjacency matrix $A$~\cite{AlonYZ97}.
It takes only $O(m^{1.41})$ time, but its memory usage, $O(n^2)$, is quadratic.
Latapy~\cite{Latapy08} presented a fast and space-efficient algorithm, but in a non-stream setting.
To handle simple graph streams, Becchetti et al.~\cite{BecchettiBCG10} devised a semi-streaming algorithm based on min-wise independent permutations~\cite{BroderCFM98}.
Their algorithm requires $O(n)$ space in main memory and $O(\log(n))$ sequential scans over the edges of the graph.
However, it is inappropriate for handling a graph stream in real time due to its multiple scans to the graph.
Kutzkov and Pagh~\cite{KutzkovP13} proposed a randomized one pass algorithm based on node coloring~\cite{PaghT12}.
Despite the property of scanning once, it is limited in practice because it requires prior knowledge about the graph.
These problems are resolved by \mascot proposed in~\cite{LimK15} and it is based on a triangle counting method~\cite{TsourakakisKMF09} which uses an edge sampling.
\mascot takes only one parameter of edge sampling probability with one sequential scan to the graph stream but \mascot does not guarantee the memory usage; it means that eventually the out-of-memory error will happen.
De Stefani et al.~\cite{StefaniERU16} proposed \triest which computes the number of triangles in fully-dynamic streams with a fixed memory size.
Their large variance, however, causes low accuracy since accuracy heavily depends on a variance in a stream setting.
Our proposed algorithm \methodcx improves the accuracy of \triest by reducing the variance through biased estimation.
For multigraph streams, there are no existing works. 

\subsection{Reservoir Sampling}
Reservoir sampling is a technique to sample a given number of elements from a stream uniformly at random~\cite{Vitter85}.
Let $S$ be a stream of items, and $D$ be an array of size $M$.
For each item $e$ arriving at time $T\geq 1$, the sampling procedure is as follows.
If $T\leq M$, $e$ is unconditionally stored in $D_T$, the $T$th item of $D$.
\chngwww{If $T>M$}, first pick a random integer $i$ from $\{1,\ldots,T\}$.
If $i\leq M$, the existing item at $D_i$ is dropped and $e$ is stored at $D_i$; otherwise, $e$ is discarded.
As a result, the sampling probability for every observed item becomes $\min(M/T, 1)$ at time $T$.
This reservoir sampling is used in our algorithms to keep a constant number of edges from a graph stream regardless of its length, as presented in Section~\ref{sec:method}.
\vspace{-3.8mm}
\paragraph{Reservoir Sampling With Random Hash}
Reservoir sampling with random hash \chngwww{(random sort~\cite{sunter1977list})} is a modification of reservoir sampling to sample distinct items uniformly at random in stream environment.
The main idea is to assign each item a random hash value and to keep $M$ items with minimum hash values.
The arriving item $e$ is sampled if and only if $h(e)<h_{max}=\max_{f \in D}h(f)$.
Note that items in the buffer have minimal hash values among all items in a stream.
The sampling technique samples distinct items uniformly at random, because it is equivalent to picking $M$ distinct items with minimum values in the whole stream.
The reservoir sampling with random hash is used in \methodm to process multigraph stream\chngwww{s}.
\vspace{-1mm}

\hide{
	\subsubsection{Edge Sampling based Local Triangle Counting}\label{sec:prelim}
	Edge sampling has been adopted by various graph mining algorithms.
	Among them, \doulion provides a basic edge sampling technique for global triangle counting~\cite{TsourakakisKMF09}. Given a graph $G=(V,E)$ and $0<p\leq 1$, it samples every edge with probability $p$. Let $\triangle$ be the number of the original graph, and $\tau$ be that of the sampled graph; then unbiased estimation for $\triangle$ is calculated as $\tau p^{-3}$. This is because each triangle is sampled with probability $p^3$.
	
	Recently, this idea has been extended to local triangle counting in graph streams by Lim and Kang~\cite{LimK15}.
	Their primal algorithm \mascotc can be understood as an online \doulion.
	Whenever a new edge $e=(u,v)$ arrives, \mascotc samples $e$ with a predefined probability $p$. If the sampling succeeds, they update estimation as $\tau_w = \tau_w + p^{-3}|N_{uv}|$ for $w\in \{u,v\}$, and $\tau_w = \tau_w + p^{-3}$ for $w\in N_{uv}$ where $N_{uv} = N_u \cap N_v$ and $N_u$ is the set of neighbors of $u$. That is, for a sampled triangle created by edge sampling, it updates estimation as like \doulion.
	Their main algorithm \mascot raises accuracy of the estimation via decoupling sampling and counting triangles. Precisely, whenever a new edge $e=(u,v)$ is sampled, \mascot first updates estimation as if $e$ is sampled, and then determines whether keeping $e$ or not by sampling.
	Note that in \mascot the probability $p^3$ of sampling a triangle is different from the probability $p^2$ of counting it---the last edge of a triangle is unconditionally considered.
	This technique changes the update as follows: $\tau_w = \tau_w + p^{-2}|N_{uv}|$ for $w\in \{u,v\}$, and $\tau_w = \tau_w + p^{-2}$ for $w\in N_{uv}$. \mascot achieves higher accuracy with the same memory spaces compared with \mascotc.
	This idea enables to achieve higher accuracy than only using the reservoir sampling.
}

\section{Proposed Method}
\label{sec:method}

How can we accurately estimate local triangles in simple and multigraph streams with a fixed memory space?
In this section, we propose \method,
a suite of accurate and memory efficient local triangle counting methods for graph streams to answer the question.
We first present two basic algorithms \methodc (Section~\ref{sec:method:simple}) and \methodm (Section~\ref{sec:method:multi}) based on the reservoir sampling: \methodc
is an algorithm for simple graph streams, and
\methodm is an algorithm for multigraph streams.
Our main proposed algorithms \methodcx and \methodmx (Section~\ref{sec:method:x}) further improve the accuracy of \methodc and \methodm by additionally assembling intermediate estimation results.
Note that in all the proposed methods, we store edges only in a buffer $D$ which has a fixed capacity $M$.

\subsection{FURL for Simple Graph}
\label{sec:method:simple}

We present \methodc (Algorithm~\ref{alg:methodc}) which estimates local triangles in simple graph \chngwww{streams} without duplicate edges.
The algorithm computes triangle estimations with \chngwww{a} fixed memory using reservoir sampling, and decouples two events of sampling and counting a triangle for efficiency using the idea of MASCOT~\cite{LimK15}. A recent method called TRIEST~\cite{StefaniERU16}, which also uses the same idea of combining MASCOT and the reservoir sampling, has been proposed independently.
Below, we describe the procedure of the method.

Let $T$ be the current time, $e=(u,v)$ be a new edge arriving from a stream at time $T$, and $c_u$ be local triangles estimation of node $u$.
First, local triangle estimation is updated if the new edge $e$ forms a triangle with two edges in the sample graph (lines~\ref{line:updates_start}--\ref{line:updates_end}). If a triangle $\lambda$ is counted by \methodc it is when its last edge arrives.
Before the buffer $D$ overflows, \methodc has the exact triangle counts since the buffer stores all edges which \chngwww{have arrived} so far.
\chngwww{The triangle estimation is increased by factor $1$ when we have the exact triangle counts, and it is increased by factor $q_T=\frac{(T-1)(T-2)}{M(M-1)}$ otherwise.}
$ExactCnt$ in \algref{alg:methodc} indicates whether $c_u$ equals the true local triangle count of node $u$ or not, and is initialized to $true$.
Let $N_x$ and $N_{xy}$ be a set of neighbors of $x$, and a set of common neighbors of $x$ and $y$ in the sample graph, respectively.
Then, for each $e=(u,v)$ from the stream, $c_u$ and $c_v$ increase by $|N_{uv}| \cdot \theta$,  and for every $w\in N_{uv}$, $c_w$ increases by $\theta$ where $\theta$ is a triangle estimation increase factor.
The second task is to sample $e$ with probability $\min(M/T, 1)$ using reservoir sampling.
In the sampling procedure, $ExactCnt$ turns to $false$ when the buffer overflows (line~\ref{line:exactcnt_off_s}),
\chngwww{because $D$ does not store all arriving edges anymore.}

\begin{algorithm}[t!]
	\SetKwFunction{UpdateS}{$UpdateTriangles\text{-S}$}
	\SetKwFunction{Increase}{$IncreaseEstimation$}
	\SetKwFunction{SampleS}{$SampleNewEdge\text{-S}$}
	\SetKwFunction{ReplaceS}{$ReplaceEdge\text{-S}$}
	\SetKwFunction{DiscoverNode}{$DiscoverNode$}
	
	\SetKwProg{myproc}{Subroutine}{}{end}
	
	\KwIn{Graph stream $S$, and max. number $M$ of edges to store}
	\KwOut{Local triangle estimation $c$}

	\medskip
	
	$G \leftarrow (V,D)$ with $V=D=\emptyset$.\;
	$D$ is initialized by $M$ dummy edges.\;
	The estimated count $c$ for nodes is initialized to $\emptyset$.\;
	The boolean flag $ExactCnt$ is initialized to $true$.\;
	\ForEach(~~\tcp*[h]{at time $T$}){edge $e=(u,v)$ from $S$}{
		$\DiscoverNode(u)$.\;
		$\DiscoverNode(v)$.\;
		$\UpdateS(e)$.\;
		$\SampleS(e)$
}

\medskip

\myproc{$\DiscoverNode{u}$}{
	\lIf{$u\notin V$}{$V \leftarrow V\cup\{u\}$ and $c_u=0$.}
}

\medskip
\myproc{$\UpdateS{e}$}{ \label{line:updates_start}
	\lIf{$ExactCnt=true$}{ $\Increase(e,1)$.
	}\Else{
	$q_T \leftarrow \frac{(T-1)(T-2)}{M(M-1)}$.\;
	$\Increase(e,q_T)$.\;			
	}
} \label{line:updates_end}

\medskip
\myproc{\Increase{$e=(u,v)$, $\theta$}}{
	$N_{uv} \leftarrow N_u \cap N_v.$\;
	\lForEach{$w\in N_{uv}$}{
		$c_w \leftarrow c_w + \theta$.
	}
	$c_u \leftarrow c_u + \left(\theta\times |N_{uv}| \right)$. \;
	$c_v \leftarrow c_v + \left(\theta\times |N_{uv}| \right)$. \;
}

\medskip
\myproc{$\SampleS{e}$}{
	\lIf{$|D|<M$}{
		Append $e$ to D.
	}\Else{
	\lIf{$ExactCnt=true$}{
		$ExactCnt \leftarrow false$. \label{line:exactcnt_off_s}
	}
	$\ReplaceS(e)$.\;	
	}
}

\medskip
\myproc{$\ReplaceS{e}$}{
	$i \leftarrow uniform(1,T)$.\;
	\lIf{$i \leq M$}{
		$D_i \leftarrow e$
	}
}

\caption{\methodc for simple graph stream}
\label{alg:methodc}	
\end{algorithm}

\methodc provides unbiased estimation as shown in the following Lemma~\ref{lem:methodc_expect}.
Note that we do not consider when $ExactCnt$ is $true$ since \methodc provides the exact counting.
\begin{lemma}\label{lem:methodc_expect}
	Let $\varDelta_u$ be the true local triangle count for a node $u$, and $c_u$ be the estimation given by \methodc. For every node $u$,
	$
	\expect{c_u} = \varDelta_u.
	$
\end{lemma}
\begin{proof}
	\chng{
		%Let $X_{\lambda}$ be a 0-1 random variable indicating whether a triangle $\lambda=(e,f,g)$ is counted or not, and $T_e$ be the time that $e$ arrives.
		Let $T_\lambda$ be the time $\lambda$ is formed, $S_\lambda^{-}$ be a set of triangles that \mchng{are} formed when we have the exact triangle counts, i.e. $T_\lambda\leq M+1$ , and $S_\lambda^{+}$ be a set of triangles that  \mchng{are} formed when we do not have the exact triangle counts, i.e. $T_\lambda> M+1$. \mchng{Note that every triangle is included in} either $S_\lambda^{-}$ or $S_\lambda^{+}$.  We define $X_\lambda^{-}$ and $X_\lambda^{+}$ as follows.
		$$
		X_{\lambda}^{-} =
		\begin{cases}
		1 &  \lambda \ is \ counted\\
		0 & otherwise.
		\end{cases}
		$$	
		$$
		X_{\lambda}^{+} =
		\begin{cases}
		q_{T_\lambda}= \frac{(T_\lambda-1)(T_\lambda-2)}{M(M-1)} & \lambda \ is \ counted\\
		0 & otherwise.
		\end{cases}
		$$		
%		$$
%		X_{\lambda}^{-} =
%		\begin{cases}
%		1 & \lambda \ is \ counted \ at \ T_\lambda \leq M+1\\
%		q_{T_\lambda} & \lambda \ is \ counted \ at \ T_\lambda > M+1\\
%		0 & otherwise.
%		\end{cases}
%		$$
%		%Here, $e,f$ and $g$ are edges and thus a triangle is represented as a triple of edges.
		%Below, we assume that $e,f$ and $g$ arrive in the order.
		%Note that if any triangle is counted by \methodc, the time for counting is when its last edge arrives.
		%Thus, at time $T$ with the buffer size $M$,
	}
		
		\chng{
		Then, $$
		\expect{c_u} = \sum\limits_{\lambda \in \Lambda_u \cap S_\lambda^{-}} \expect{X_{\lambda}^{-}} + \sum\limits_{\lambda \in \Lambda_u \cap S_\lambda^{+}} \expect{X_{\lambda}^{+}}
		$$		
		where $\Lambda_u$ is the set of triangles containing $u$.
		For a triangle $\lambda\in \Lambda_u \cap S_\lambda^{-}$,
		$$
			\expect{X_{\lambda}^{-}}=1 \times\prob{\lambda \ is \ counted} = 1.
		$$
		$\prob{\lambda \ is \ counted} = 1$ since all edges are unconditionally sampled in the buffer.
		For a triangle $\lambda \in \Lambda_u \cap S_\lambda^{+}$,
		$$
		\expect{X_{\lambda}^{+}}=q_{T_\lambda} \times\prob{\lambda \ is \ counted} = q_{T_\lambda} \times \frac{M(M-1)}{(T_\lambda-1)(T_\lambda-2)} = 1.
		$$
		%because $\prob{\lambda \ is \ counted} = \frac{M(M-1)}{(T_\lambda-1)(T_\lambda-2)} = q_{T_\lambda}^{-1}$.
		%Note that $\prob{\lambda \ is \ counted} = \min \frac{M(M-1)}{(T_\lambda-1)(T_\lambda-2)} = q_{T_\lambda}^{-1}$.
		%	
		%In \methodc, any triangle $\lambda=(e,f,g)$ is counted if and only if $e$ and $f$ are samples at $T_g-1$. Thus, $\prob{X_{\lambda} = 1} = q_{T_\lambda}$.
		\mchng{Thus,}
		$$
		\expect{c_u} = \sum\limits_{\lambda \in \Lambda_u \cap S_\lambda^{-}} \expect{X_{\lambda}^{-}} + \sum\limits_{\lambda \in \Lambda_u \cap S_\lambda^{+}} \expect{X_{\lambda}^{+}}= \sum\limits_{\lambda \in \Lambda_u \cap S_\lambda^{-}}1 + \sum\limits_{\lambda \in \Lambda_u \cap S_\lambda^{+}}1= |\Lambda_u| = \varDelta_u.
		$$
		}
	\qed
\end{proof}

\subsection{FURL for Multigraph}
\label{sec:method:multi}

We present two fixed-memory local triangle counting algorithms in \chngwww{multigraph streams containing duplicate edges.}
\chngwww{There are two ways in handling the duplication: binary counting and weighted counting.}
\chngwww{The} binary counting considers only the existence of an edge, leaving out the number of occurrences of an edge.
In contrast, the weighted counting takes the number of occurrences of an edge into account. For instance, given a triangle $\lambda = (e_a,e_b,e_c)$ and a multigraph stream $(e_a,e_a,e_a,e_b,e_b,e_c)$, binary counting counts $\lambda$ as 1 while weighted counting counts $\lambda$ as $3 \times 2 \times 1 =6$. \chngwww{Below, we describe \methodb for binary counting, and \methodw for weighted counting in multigraph streams.}

\paragraph{Binary Counting (\methodb)}
For binary triangle counting, we sample distinct edges with the equal probability regardless of the degree of duplications: i.e., we make sure $M$ edges in buffer $D$ are chosen uniformly at random from the set of distinct edges observed so far.
Reservoir sampling with random hash is used to deal with duplicate edges.
\chngwww{In the reservoir sampling with random hash,} each arriving edge $e$ is assigned a hash \chngwww{value} $h(e)$, and it is sampled if $h(e)$ belongs to minimum $M$ hash values.
\chngwww{
Note that duplicate edges are treated as one edge since they have a same hash value. }
Without loss of generality, we assume the codomain of the hash function is (0,1).

\begin{algorithm}[t!]
	\SetKwProg{myproc}{Subroutine}{}{end}
	\SetKwFunction{UpdateMB}{$UpdateTriangles\text{-M}_B$}
	\SetKwFunction{SampleMB}{$SampleNewEdge\text{-M}_B$}
	\SetKwFunction{ReplaceM}{$ReplaceEdge\text{-M}$}
	\KwIn{Graph stream $S$, maximum number $M$ of edges stored, and a random hash function $h$ }
	\KwOut{Local triangle estimation $c$}
	\medskip
	\chngwww{Initialize variables as in lines 1-4 of \algref{alg:methodc}.\;}
	\ForEach(~~\tcp*[h]{at time $T$}){edge $e=(u,v)$ from $S$}{
		$\DiscoverNode(u)$.\;
		$\DiscoverNode(v)$.\;
		\If{$e \notin D$}{ \label{line:not_in_buffer}
			$\SampleMB(e)$.\;
			$\UpdateMB(e)$.\;
		}
	}
	
	\myproc{$\SampleMB{e}$}{ \label{line:sampleb_start}
		\lIf{$|D|<M$}{
			Append $e$ to D.
		}\Else{
		\lIf{$ExactCnt=true$}{
			$ExactCnt \leftarrow false$.
		}
		$\ReplaceM(e)$.\;	
	}
} \label{line:sampleb_end}

\medskip
\myproc{$\ReplaceM {e} $}{
	\If{$h(e) < h_{max}$}{
		$D_{max} \leftarrow e' ~such ~that ~ h(e')=h_{max}$.\;
		Remove $D_{max}$ from $D$.\;
		Append $e$ to $D$.\;	
	}
}

\medskip
\myproc{$\UpdateMB{e}$}{ \label{line:updateb_start}
	\lIf{$ExactCnt=true$}{ $\Increase(e,1)$.
	}\ElseIf{$e$ is sampled}{
	$q_T \leftarrow \frac{M-3}{M} \frac{1}{h_{max}^3} $.\;
	$\Increase(e,q_T)$.\;			
}		
} \label{line:updateb_end}

\caption{\methodbrb for multigraph stream (binary)}
\label{alg:methodb}

\end{algorithm}

Different from \methodc, \methodb (Algorithm~\ref{alg:methodb}) goes through sampling process first and then updates triangle estimation. \methodb cannot decouple the event of triangle counting and sampling the edge because it causes the same triangle to be counted multiple times by duplicate edges. Note that \methodb discards $e$ if $e$ is already sampled (line~\ref{line:not_in_buffer}).
The sampling process (lines~\ref{line:sampleb_start}--\ref{line:sampleb_end}) is similar to that of \methodc except that \methodb uses a hash value of an edge to replace another edge.
After the sampling process, triangle estimation is updated (lines~\ref{line:updateb_start}--\ref{line:updateb_end}). Before the buffer overflows, the triangle estimation is updated by a factor 1. After the buffer is full, triangle estimation is updated by a factor $q_T= \frac{M-3}{M} \cdot \frac{1}{(h_{max}^3)}$ only when the edge $e$ is sampled where $h_{max}=max_{f\in D} h(f)$ is the maximum hash value in the buffer. Note that if any triangle is counted by \methodb it is counted at the first time all three edges of a triangle arrive.

\chngwww{\methodb provides unbiased estimation as shown in the following lemma.
%Note that we do not consider when $ExactCnt$ is $true$ since \methodb provides the exact counting.
}
\begin{lemma}\label{lem:methodb_expect}
	Let $\varDelta_u$ be the true local triangle count for a node $u$, and $c_u$ be the estimation given by \methodb. For every node $u$,
	$
	\expect{c_u} = \varDelta_u.
	$
\end{lemma}
\begin{proof}
	%The estimation triangle count for a node $u$, $c_u$, is a sum of triangle estimation for the set of triangles containing $u$. Before the buffer overflows, all distinct edges are unconditionally sampled and thus all the triangles are also counted unconditionally with weight 1. It is evident that triangles formed before the buffer overflows gives unbiased estimation. The rest of the proof shows the unbiasedness of triangle estimation for the triangles formed after the buffer overflows.
	
	\chng{
	Let $T_\lambda$ be the time $\lambda$ is formed and $T_M$ be the last time we have the exact triangle count, that is the last time all edges are unconditionally sampled. Let $S_\lambda^{-}$ be a set of triangles that are formed when we have the exact triangle counts, i.e. $T_\lambda\leq T_M$ , and $S_\lambda^{+}$ be a set of triangles that are formed when we do not have the exact triangle counts, i.e. $T_\lambda> T_M$. Note that every triangle is included in  either $S_\lambda^{-}$ or $S_\lambda^{+}$.
We define $X_\lambda^{-}$ and $X_\lambda^{+}$ as follows.
	$$
	X_{\lambda}^{-} =
	\begin{cases}
	1 &  \lambda \ is \ counted\\
	0 & otherwise.
	\end{cases}
	$$	
	$$
	X_{\lambda}^{+} =
	\begin{cases}
	q_{T_\lambda}= \frac{M-3}{M} \cdot \frac{1}{h_{max}^3} & \lambda \ is \ counted\\
	0 & otherwise.
	\end{cases}
	$$		
%	Let $T_\lambda$ be the time a triangle is counted and $T_M$ be the last time we have the exact triangle count, that is the last time all edges are unconditionally sampled. Let $X_\lambda$ be a random variable that is
%	$$
%	X_{\lambda} =
%	\begin{cases}
%	1 & \lambda ~is ~counted ~at ~T_\lambda \leq T_M \\
%	q_{T_\lambda}= \frac{M-3}{M} \cdot \frac{1}{(h_{max}^3)}& \lambda ~is ~counted ~at ~T_\lambda > T_M \\
%	0 & \lambda ~ is ~not~ counted.
%	\end{cases}
%	$$
}
	%$X_{\lambda}$ be a random variable that has the value $q_{T_\lambda}= \frac{M-3}{M} \cdot \frac{1}{(h_{max}^3)}$ if a triangle $\lambda$ is sampled or 0 otherwise.
	
	\chng{
	Then,
	$$
	\expect{c_u} =  \sum\limits_{\lambda \in \Lambda_u \cap S_\lambda^{-}} \expect{X_{\lambda}^{-}} + \sum\limits_{\lambda \in \Lambda_u \cap S_\lambda^{+}} \expect{X_{\lambda}^{+}}
	$$
	where $\Lambda_u$ is the set of triangles containing node $u$.
	For a triangle $\lambda\in \Lambda_u \cap S_\lambda^{-}$, $$
	\expect{X_{\lambda}^{-}}=1 \times\prob{\lambda \ is \ counted} = 1. $$  $\prob{\lambda \ is \ counted}=1$ since all edges are unconditionally sampled at $T\leq T_M$. %have $\prob{\lambda ~is ~counted}=1$.
	}
	%Every triangle is unconditionally counted because every edge in the stream so far is unquestioningly sampled.
	%It is evident that triangles formed when distinct edges are unconditionally sampled, i.e, $T_\lambda \leq T_M$ gives unbiased estimation. Distinct triangles are counted unconditionally since all distinct edges so far in the stream are unquestioningly sampled in the buffer.
	
	\chng{
	Now we show $\expect{X_{\lambda}^{+}}=1$ for each triangle $\lambda\in \Lambda_u \cap S_\lambda^{+}$.
	}Let $u(T_\lambda)$ be the number of unique edges at time $T_\lambda$. Leaving out all the repeated edges, we get a refined stream that can be viewed as a sequence of independent random variables $h_i(1 \leq i \leq u(T_\lambda))$ that has uniform distribution in range (0,1). Then $h_{max}$ can be viewed as the $M$th smallest value among $u(T_\lambda)$ number of random variables in range $(0,1)$. By order statistics, $h_{max} \sim Beta(M,u(T_\lambda)+1-M)$ and its probability density function is given as $f(x)=\frac{\gamma(u(T_\lambda)+1)}{\gamma(M) \cdot \gamma(u(T_\lambda)+1-M)} \cdot x^{M-1} \cdot (1-x)^{u(T_\lambda)-M}$.
	To get $\expect{X_{\lambda}^{+}}$ we compute the expectation of $q_{T_\lambda}$ over $h_{max}$, and multiply it with the probability of $\lambda$ being sampled.
	Note that the probability of an event that all three edges of a triangle $\lambda$ being sampled is $\frac{M(M-1)(M-2)}{u(T_\lambda)(u(T_\lambda)-1)(u(T_\lambda)-2)}$, and it is independent of an event $h_{max}=x$.

	\begin{align*}
	\chng{\expect{X_{\lambda}^{+} }}& = \prob{\lambda ~is ~counted}\cdot \int_{0}^{1} q_{T_\lambda} \cdot f(x) dx\\
	&=\frac{M(M-1)(M-2)}{u(T_\lambda)\cdot(u(T_\lambda)-1)\cdot(u(T_\lambda)-2)} \cdot \int_{0}^{1}\frac{M-3}{M} \frac{1}{x^3} \cdot f(x) dx\\
	&=\int_{0}^{1}  \frac{\gamma(u(T_\lambda)-2)}{\gamma(M-3) \cdot \gamma(u(T_\lambda)+1-M)} \cdot x^{M-4} \cdot (1-x)^{u(T_\lambda)-M}dx\\
	&=1
	\end{align*}
	
	\chng{
		Therefore, $$
		\expect{c_u} = \sum\limits_{\lambda \in \Lambda_u \cap S_\lambda^{-}} \expect{X_{\lambda}^{-}} + \sum\limits_{\lambda \in \Lambda_u \cap S_\lambda^{+}} \expect{X_{\lambda}^{+}}= \sum\limits_{\lambda \in \Lambda_u \cap S_\lambda^{-}}1 + \sum\limits_{\lambda \in \Lambda_u \cap S_\lambda^{+}}1= |\Lambda_u| = \varDelta_u.
		$$
	}
\qed
\end{proof}

\paragraph{Weighted Counting (\methodw)}
In weighted triangle counting, weights of edges are multiplied in computing the number of triangles.
We propose \methodw (Algorithm~\ref{alg:methodw}) for weighted triangle counting in multigraph streams.
\chngwww{ \methodw uses reservoir sampling with random hash to sample $M$ distinct edges uniformly at random. Along with an edge $e$, the occurrence number $O_e$ is kept in the buffer to reflect the edge's weight.
}

The algorithm first updates local triangle estimation and then goes through the sampling procedure (lines~\ref{line:updatew_start}--\ref{line:updatew_end}). If a triangle is counted by \methodw it is counted at the time last edge of a triangle arrives.
The triangle estimation is increased by factor 1 when we have the exact count of triangles, and by $q_T=\frac{M-2}{M} \cdot \frac{1}{(h_{max}^2)}$ otherwise.
The procedure of \chngwww{incrementing triangle estimation (lines~\ref{line:increase_start}--\ref{line:increase_end}) is a bit different from those of \methodc and \methodb because weighted triangle counting considers the edges' weights.} %treats triangles that share the same edges as different triangles.
When an edge $e=(u,v)$ arrives, for each $w\in N_{uv}$, $(O_{(u,w)} \cdot O_{(v,w)})$ number of triangles are created. Therefore triangle estimations $c_u,c_v,$ and $c_w$ are incremented by multiplying a given factor with $O_{(u,w)}$ and $O_{(v,w)}$.

\begin{algorithm}[t!]
	\SetKwProg{myproc}{Subroutine}{}{end}
	\SetKwFunction{SampleMW}{$SampleNewEdge\text{-M}_W$}
	\SetKwFunction{UpdateMW}{$UpdateTriangles\text{-M}_W$}
	\SetKwFunction{IncreaseMW}{$IncreaseEstimation\text{-M}_W$}
	
	\KwIn{Graph stream $S$, maximum number $M$ of edges stored, and a random hash function $h$ }
	\KwOut{Local triangle estimation $c$}
	
	\medskip
	\chngwww{Initialize variables as in lines 1-4 of \algref{alg:methodc}.\;}
	
	\ForEach(~~\tcp*[h]{at time $T$}){edge $e=(u,v)$ from $S$}{
		$\DiscoverNode(u)$.\;
		$\DiscoverNode(v)$.\;
		%		$p = \min((M/T)^3,1)$.\;
		$\UpdateMW(e)$.\;
		$\SampleMW(e)$.\;
	}
	
	\medskip
	\myproc{$\UpdateMW(e)$}{ \label{line:updatew_start}
		\lIf{$ExactCnt$}{
			$\IncreaseMW(e,1)$.
		}\Else{
		$q_T \leftarrow \frac{M-2}{M} \frac{1}{h_{max}^2} $.\;
		$\IncreaseMW(e,q_T)$.\;			
	}	
} \label{line:updatew_end}
\medskip
\myproc{\IncreaseMW{$e=(u,v)$, $\theta$ } }{ \label{line:increase_start}
	$N_{uv} \leftarrow N_u \cap N_v.$\;
	\ForEach{$w\in N_{uv}$}{
		$c_w \leftarrow c_w+\theta \cdot O_{(u,w)} \cdot O_{(v,w)}$.\;
		$c_u \leftarrow c_u+\theta \cdot O_{(u,w)} \cdot O_{(v,w)}$.\;
		$c_v \leftarrow c_v+\theta \cdot O_{(u,w)} \cdot O_{(v,w)}$.\;
	}
}\label{line:increase_end}

\medskip
\myproc{$\SampleMW{e}$}{ \label{line:samplew_start}
	\lIf{$e \in D$}{
		$O_e \leftarrow O_e+1$.
	}\lElse{
	The same as $\SampleMB$.
}	
} \label{line:samplew_end}

\caption{\methodwrw for multigraph stream (weighted)}
\label{alg:methodw} 	
\end{algorithm}

\chngwww{The sampling procedure is based on reservoir sampling with random hash (lines~\ref{line:samplew_start}--\ref{line:samplew_end}). It} is similar to \methodb except for how it handles duplicate edges. When the edge that is already in the buffer arrives again, \methodw unconditionally samples the edge. Sampling the edge here corresponds to increasing $O_e$ by 1.
The algorithm gives unbiased estimation as shown in Lemma~\ref{lemma:furlmw}. % Note that we do not consider when $ExactCnt$ is $true$ since \methodw provides the exact counting.

\begin{lemma}
\label{lemma:furlmw}
	Let $\varDelta_u$ be the true local triangle count for a node $u$, and $c_u$ be the estimation given by \methodw. For every node $u$,
	$
	\expect{c_u} = \varDelta_u.
	$
\end{lemma}
\begin{proof}
\chng{
	Let $T_\lambda$ be the time $\lambda$ is formed and $T_M$ be the last time we have the exact triangle count, that is the last time all edges are unconditionally sampled. Let $S_\lambda^{-}$ be a set of triangles that are formed when we have the exact triangle counts, i.e. $T_\lambda\leq T_M$ , and $S_\lambda^{+}$ be a set of triangles that are formed when we do not have the exact triangle counts, i.e. $T_\lambda> T_M$. Note that every triangle is included in either $S_\lambda^{-}$ or $S_\lambda^{+}$. We define $X_\lambda^{-}$ and $X_\lambda^{+}$ as follows.
	$$
	X_{\lambda}^{-} =
	\begin{cases}
	1 &  \lambda \ is \ counted\\
	0 & otherwise.
	\end{cases}
	$$	
	$$
	X_{\lambda}^{+} =
	\begin{cases}
	q_{T_\lambda}= \frac{M-2}{M} \cdot \frac{1}{(h_{max}^2)} & \lambda \ is \ counted\\
	0 & otherwise.
	\end{cases}
	$$	
}
%$X_{\lambda}$ be a random variable that has the value $q_{T_\lambda}= \frac{M-3}{M} \cdot \frac{1}{(h_{max}^3)}$ if a triangle $\lambda$ is sampled or 0 otherwise.

\chng{
Then,
$$\expect{c_u} =  \sum\limits_{\lambda \in \Lambda_u \cap S_\lambda^{-}} \expect{X_{\lambda}^{-}} + \sum\limits_{\lambda \in \Lambda_u \cap S_\lambda^{+}} \expect{X_{\lambda}^{+}}$$
where $\Lambda_u$ is the set of triangles containing node $u$.
For a triangle $\lambda\in \Lambda_u \cap S_\lambda^{-}$, $$
\expect{X_{\lambda}^{-}}=1 \times\prob{\lambda \ is \ counted}=1 $$ because $\prob{\lambda \ is \ counted}=1$ since all edges are unconditionally sampled at $T\leq T_M$.
}

\chng{
Now we show $\expect{X_{\lambda}^{+}}=1$ for each triangle $\lambda\in \Lambda_u \cap S_\lambda^{+}$.
} %when the buffer overflowed that edges are stochastically sampled.
%The estimation triangle count for a node $u$, $c_u$, is a sum of triangle estimation for the set of triangles containing $u$. Before the buffer overflows, all the edges in the stream are unconditionally sampled and thus all the triangles are also counted unconditionally with weight 1. It is evident that triangles formed before the buffer overflows gives unbiased estimation. The rest of the proof shows triangle estimation for the triangles formed after the buffer overflows is unbiased.
%Let $T_\lambda$ be the first time all three edges of a triangle arrive, and  $X_{\lambda}$ be a random variable that has the value $q_{(T_\lambda-1)}= \frac{M-2}{M} \cdot \frac{1}{(h_{max}^2)}$ if a triangle $\lambda$ is counted or 0 otherwise. Then $\expect{c_u} =  \sum_{\lambda\in \Lambda_u} \expect{X_{\lambda} }$ where $\Lambda_u$ is the set of triangles containing node $u$.
%We now show	$\expect{X_{\lambda} }=1$.
Let $u(T_\lambda)$ be the number of unique edges at time $T_\lambda$. Leaving out all the repeated edges, we get a refined stream that can be viewed as a sequence of independent random variables $h_i(1 \leq i \leq u(T_\lambda-1))$ that has uniform distribution in range (0,1).
Note that in \methodw triangle is updated before sampling the edge thus when updating a triangle at $T_\lambda$, the edges in the buffer is a result of one time stamp ago, i.e, $T_\lambda-1$.
Then $h_{max}$ can be viewed as the $M$th smallest value among $u(T_\lambda-1)$ number of random variables in range $(0,1)$.
By order statistics, $h_{max} \sim Beta(M,u(T_\lambda-1)+1-M)$ and its probability density function is given as $f(x)=\frac{\gamma(u(T_\lambda-1)+1)}{\gamma(M) \cdot \gamma(u(T_\lambda-1)+1-M)} \cdot x^{M-1} \cdot (1-x)^{u(T_\lambda-1)-M}$.
To get \chng{$\expect{X_{\lambda}^{+}}$}, we compute the expectation of $q_{T_\lambda-1}$ over $h_{max}$, and multiply it with the probability of $\lambda$ being counted.
\chng{
The probability of an event of a triangle $\lambda$ being counted is $\frac{M(M-1)}{u(T_\lambda-1)(u(T_\lambda-1)-1)}$ because \mchng{$\lambda$} is counted if \mchng{an arriving} edge forms $\lambda$ with the other two edges in the buffer. Note that $\prob{\lambda ~is ~counted}$ is independent of an event $h_{max}=x$.
}
%	Let $f(x)$ denote the probability density function of $m_{c-1}(M)$. We focus on $m_{c-1}(M)$, not $m_c(M)$ because the algorithm for weighted counting updates triangle estimation first, and then goes through the sampling procedure. Also note that triangle $\lambda=(e_a,e_b,e_c)$ with duplicate edges are treated as $O_{e_a}\cdot O_{e_b} \cdot O_{e_c}$ number of different triangles. Then $\expect{\tau_u} =  \sum_{\lambda \in \Lambda_u} \expect{X_{\lambda} }$.
	\begin{align*}
	\chng{\expect{X_{\lambda}^{+}}} & = \prob{\lambda ~is ~counted}\cdot \int_{0}^{1} q_{(T_\lambda-1)} \cdot f(x) dx \\
	&=\frac{M(M-1)}{u(T_\lambda-1)\cdot(u(T_\lambda-1)-1)} \cdot \int_{0}^{1}\frac{M-2}{M} \frac{1}{x^2} \cdot f(x) dx\\
	&=\int_{0}^{1}  \frac{\gamma(u(T_\lambda-1)-1)}{\gamma(M-2) \cdot \gamma(u(T_\lambda-1)+1-M)} \cdot x^{M-3} \cdot (1-x)^{u(T_\lambda-1)-M}dx\\
	&=1
%	\expect{X_\lambda }&= \int_{0}^{1}  \frac{M-2}{M} \frac{1}{x^2} \prob{e_a, e_b \in D \mid m_{c-1}(M)=x} \cdot f(x) dx	\nonumber\\
%	&=\frac{M-2}{M}\cdot \frac{M(M-1)}{(c-1)(c-2)} \int_{0}^{1} \frac{1}{x^2} \cdot f(x) dx\\
%	&=\frac{\gamma(c-2)}{\gamma(M-2) \cdot \gamma(c-M)} \int_{0}^{1} x^{M-3} \cdot (1-x)^{c-1-M}dx\\
%	&=1
	\end{align*}
	
	\chng{
		Therefore, $$
		\expect{c_u} = \sum\limits_{\lambda \in \Lambda_u \cap S_\lambda^{-}} \expect{X_{\lambda}^{-}} + \sum\limits_{\lambda \in \Lambda_u \cap S_\lambda^{+}} \expect{X_{\lambda}^{+}}= \sum\limits_{\lambda \in \Lambda_u \cap S_\lambda^{-}}1 + \sum\limits_{\lambda \in \Lambda_u \cap S_\lambda^{+}}1= |\Lambda_u| = \varDelta_u.
		$$
		}
\qed
\end{proof}

\subsection{Improving Accuracy of FURL}
\label{sec:method:x}

The unbiased estimation of \methodc and \methodm guarantees the accuracy if the estimation is obtained by averaging results from multiple independent trials.
In real situations of processing graph streams, however, it requires very high costs to keep multiple independent sample graphs or scan a graph multiple times.
It means that when only one trial is allowed, the accuracy of an estimation greatly depends on the variance as well as the unbiasedness.
In this section, we propose \methodcx and \methodmx which have a lower variance than \methodc and \methodm respectively with a slightly biased estimation.
As a result, the overall estimation error of \methodcx and \methodmx become smaller than those of \methodc and \methodm.

\subsubsection{\methodcxrm}
The main idea of \methodcx (\algref{alg:methodcx}) is to combine past estimations with the current one.
The procedure of \methodcx is very similar to that of \methodc, but its estimation is calculated by a weighted average of estimations obtained at certain times.

Let $c^{(t)}$ and $\tau^{(t)}$ be the estimations of \methodc and \methodcx at time $t$, respectively. Let $T_M$ be the last time when $c^{(t)}$ equals the true triangle count.
$T_M$ is the time when $|D|=M$ and thus $ExactCnt$ becomes $false$, because we cannot store all the edges after this moment (line 13).
We consider $[1,T_M]$ to be the bucket $0$ and divide the interval $[T_M+1, T]$ into buckets of size $J>0$.
Then, $\tau^{(t)} = c^{(t)}$ during the bucket $0$ and $\tau^{(t)} = \delta \tau^{(T_M+(i-1)J)} + (1-\delta)c^{(t)}$ during the bucket $i > 0$.
Note that \methodcx updates the estimation by $\tau^{(T_M+iJ)} = \delta \tau^{(T_M+(i-1)J)} + (1-\delta)c^{(T)}$ at the boundary of each bucket $i>0$, \ie $T=T_M+iJ$.

\SetKwFunction{Query}{$Query$}
\SetKwFunction{UpdateEstimation}{$UpdateEst$}

In \algref{alg:methodcx},
\methodcx internally uses \methodc and accepts two more parameters: the bucket size $J$ and a decaying factor $0\leq \delta< 1$ for the past estimation.
Note that \methodcx with $\delta=0$ is the same as \methodc.

\begin{algorithm}[t!]	
	\SetKwFunction{WeightAvg}{$WeightedAverage$}
	\SetKwFunction{SampleSX}{$SampleNewEdge\text{-SX}$}
	\SetKwProg{myproc}{Subroutine}{}{end}
	
	\KwIn{Graph stream $S$, maximum number $M$ of edges stored, interval $J$ for updating estimation, and decaying factor $\delta$ for past estimation}
	\KwOut{Local triangle estimation $\tau= \Query(\delta)$}
	
	\medskip
	
	Initialize variables as in lines 1-4 of \algref{alg:methodc}.\;
	\ForEach(~~\tcp*[h]{at time $T$}){edge $e=(u,v)$ from $S$}{
		$\DiscoverNode(u)$.\;
		$\DiscoverNode(v)$.\;
		$\UpdateS(e)$. 	//	$c$ is computed here.\;
		$\SampleSX(e)$.\;
		$\WeightAvg(\delta)$.\;
	}
	
		\medskip	
		\myproc{\SampleSX{$e$}}{
			\lIf{$|D|<M$}{
				Append $e$ to D.
			}\Else{
			\If{$ExactCnt=true$}{
				$T_M \leftarrow T$.\;
				$ExactCnt \leftarrow false$. \label{line:exactcnt_off}
			}
			$\ReplaceS(e)$.\;	
		}
	}
	
	\smallskip
	\myproc{$\WeightAvg{$\delta$}$}{
		\If{$T = T_M$}{
			$\hat{\tau} \leftarrow c$.
		}\ElseIf{$T > T_M$ and $(T-T_M)\bmod{J} = 0$}{
		$\hat{\tau} \leftarrow \delta\hat{\tau} + (1-\delta)c$. \label{line:our_est_update}
		}
	}

	\smallskip
	\myproc{$\Query{$\delta$}$}{
		\lIf{$T\leq T_M$}{
			\Return $c$.
		}\lElseIf{$(T-T_M)\bmod{J} = 0$}{
		\Return $\hat{\tau}$.
	}\lElse {
	\Return $\delta\hat{\tau} + (1-\delta)c$.
		}
	}

\caption{\methodcx for simple graph stream}
\label{alg:methodcx}
\end{algorithm}

Below,
we \chngwww{show} the expectation and the variance of each triangle appearing in the stream \chngwww{(Lemmas~\ref{lem:our_expect} and~\ref{lem:our_var})}.
Note that for triangles in the $0$th bucket, \ie~for the time $T\leq M+1$, \methodcx provides unbiased estimation with variance $0$ which means the exact counting. Thus, we consider a triangle $\lambda$ with $b_\lambda>0$ below. % below.
\begin{lemma}\label{lem:our_expect}
	Let \mchng{$b_\lambda$ be the bucket where $\lambda$ is formed} and $Y_\lambda$ be the estimated count of a triangle $\lambda$ \mchng{with $b_\lambda>0$} by \methodcx.
	For every triangle $\lambda$,
	$$
	\expect{Y_\lambda} = 1 - \phi(b_\lambda),
	$$
	where $\phi(i) = \delta^{B-i+1}$ and $B$ is the current bucket.
\end{lemma}
\begin{proof}
	Let $W(i)=(1-\delta)\delta^{B-i}$ for $i\geq 1$ and $q_{T_\lambda} =  \frac{(T-1)(T-2)}{M(M-1)}$. By definition, for $\lambda$ appearing in bucket $b_\lambda$, $Y_\lambda$ becomes $q_{T_\lambda} W(b_\lambda) + q_{T_\lambda} W(b_\lambda+1) + \cdots + q_{T_\lambda} W(B)$ if $\lambda$ is counted, otherwise $Y_\lambda=0$.
	Thus, we obtain
	$$
		\expect{Y_\lambda} = \prob{\lambda \ is \ counted} \cdot \left(q_{T_\lambda}\sum_{b_\lambda\leq j\leq B} W(j)\right) = 1 - \phi(b_\lambda).
	$$
	which proves the lemma.
	\qed
\end{proof}

\begin{lemma}\label{lem:our_var}
	Let \mchng{$b_\lambda$ be the bucket where $\lambda$ is formed} and $Y_\lambda$ be the estimated count of a triangle $\lambda$ \mchng{with $b_\lambda>0$} by \methodcx.
	For every triangle $\lambda$,
	$$
	\var{Y_\lambda} = \left(1-\phi(b_\lambda)\right) ^2
	\left( q_{T_\lambda} - 1 \right),
	$$
	where $T_\lambda$ is the time that the last edge of $\lambda$ arrives,
	$q_{T_\lambda} = \frac{(T_\lambda-1)(T_\lambda-2)}{M(M-1)}$,
	$\phi(i) = \delta^{B-i+1}$, and $B$ is the current bucket.
\end{lemma}
\begin{proof}
	Following the definition $\var{Y_\lambda} = \expect{Y_\lambda^2} - \expect{Y_\lambda}^2$ with \lemref{lem:our_expect}, the proof is done.
	\qed
\end{proof}

Note that the expectation and the variance of \methodcx are smaller than those of \methodc that has $\expect{X_\lambda}=1$ and $\var{X_\lambda}=q_{T_\lambda}-1$, respectively, where $X_\lambda$ is the estimated count of $\lambda$ by \methodc.

Now, we provide Lemmas~\ref{lem:our_accuracy} and \ref{lem:our_concent_prob}, which are the main results of this section; they state that the estimation of \methodcx is much concentrated around the \chngwww{true value}, \chngwww{\ie more close to the ground truth,} compared with \methodc.

\begin{lemma} \label{lem:our_accuracy}
	Consider any triangle $\lambda$ that is counted at time $T_\lambda > M + 1$.
	If $T_\lambda \geq \sqrt{2}M + 1$,
	the interval $\expect{Y_\lambda} \pm \var{Y_\lambda}$ is strictly included in the interval $\expect{X_\lambda} \pm \var{X_\lambda}$.
\end{lemma}
\begin{proof}
	We first show that $\expect{Y_\lambda} - \var{Y_\lambda} > \expect{X_\lambda} - \var{X_\lambda}$.
	
	Let $\psi_\lambda=1-\delta^{B-b_\lambda+1}$; we will find the condition satisfying
	$
	\psi_\lambda - \psi_\lambda^2\left(q_{T_\lambda}-1\right) - 2 + q_{T_\lambda} > 0,
	$
	which is developed as follows.
	\begin{align}
		  \psi_\lambda &> \frac{\left(2-q_{T_\lambda}\right)}{\left(q_{T_\lambda}-1\right)}  \label{eq:our_acc_main_ineq}
	\end{align}
	
	Below, we will show the condition under which \eqref{eq:our_acc_main_ineq} holds.
	
	Let $T_\lambda - 1 = \alpha M$.
	By definition,
	\begin{align*}
		q_{T_\lambda} &= \frac{(T_\lambda-1)(T_\lambda-2)}{M(M-1)}
		> \frac{(T_\lambda-1)^2}{M^2} = \alpha^2.
	\end{align*}
	
	Then,
	\begin{align} \label{eq:our_acc_q}
		\frac{2-q_{T_\lambda}}{q_{T_\lambda}-1} &<
			\frac{2-\alpha^2}{\alpha^2-1} = \frac{1}{\alpha^2-1} - 1.
	\end{align}
	
	Now we examine the left term of \eqref{eq:our_acc_main_ineq}. Since $1\leq b_\lambda\leq B$, the lower bound of $\psi_\lambda$ becomes
	\begin{align} \label{eq:our_acc_p_lb}
		\psi_\lambda \geq 1 - \delta.
	\end{align}
	
	Putting \eqref{eq:our_acc_q} and~\eqnumref{eq:our_acc_p_lb} into \eqref{eq:our_acc_main_ineq}, we obtain a sufficient condition for \eqnumref{eq:our_acc_main_ineq} as follows:
	\begin{align}
%		1 - \delta &> \frac{1}{2r + r^2} - 1 \\
		\alpha^2 &\geq \frac{3-\delta}{2-\delta} \label{eq:our_acc_final_ineq}
	\end{align}
	
	For $\alpha \geq \sqrt{2}$,
	\eqref{eq:our_acc_final_ineq} always holds since the upper bound of the right term becomes $2$ due to $0\leq \delta<1$.
	Thus, we finally obtain the condition under which \eqref{eq:our_acc_main_ineq} holds as follows:
	$$
	T_\lambda \geq \sqrt{2}M + 1.
	$$
	
	Due to the underestimation of \methodcx, $\expect{Y_\lambda} + \var{Y_\lambda} < \expect{X_\lambda} + \var{X_\lambda}$ always holds, which completes the proof.
%	this condition also satisfies $\expect{Y_\lambda} + \var{Y_\lambda} < \expect{Z_\lambda} + \var{Z_\lambda}$, which completes the proof.
\qed
\end{proof}

%\begin{corollary}
%	The interval by $\expect{Y_\lambda} \pm \var{Y_\lambda}$ is strictly included in that by $\expect{Z_\lambda} \pm \var{Z_\lambda}$ for every triangle $\lambda$ such that $\frac{T_\lambda-M-1}{M} >\sqrt{2}-1$.
%\end{corollary}
%\begin{proof}
%	Since \methodcx underestimates, the proof is immediate from \lemref{lem:our_accuracy}.
%\end{proof}

Note that \lemref{lem:our_accuracy} holds for all triangles except for the first few ones.
%Precisely, those exceptions only include triangles formed by the edges at time after $M$ and approximately before $1.4M$.
%For example, if we set $\alpha=1$, \ie~$J=M$,  those exceptions only include triangles formed by the edges at time after $M$ and before $1.4M$.
Fortunately, for those exceptive triangles, the estimation error is insignificant by construction, compared with \methodc.
%For instance, if we set $\alpha=1$, those triangles belong to bucket 1, and thus
%$\expect{Y_\lambda} = (1-\delta^B)\expect{Z_\lambda}$ and $\var{Y_\lambda} = (1-\delta^B)^2  \var{Z_\lambda}$. Since $\delta^B$ quickly goes to $0$ for increasing $B$, $\expect{Y_\lambda}\approx\expect{Z_\lambda}$ and $\var{Y_\lambda}\approx\var{Z_\lambda}$.

Below, we state that \methodcx guarantees a larger lower bound of the probability that its estimation is around the true value than \methodc does.

\begin{lemma}\label{lem:our_concent_prob}
	Given a triangle $\lambda$ with $T_\lambda>M+1$,
	let $\mu_Y = \expect{Y_\lambda}, \sigma_Y^2 = \var{Y_\lambda}$, and $\epsilon = \mu_X - \mu_Y+\sigma_Y^2$ where $\mu_X = \expect{X_\lambda} = 1$ is the unbiased estimate of the triangle count for $\lambda$.
	Then,
	$\prob{|Y_\lambda - \mu_X| < \epsilon}$ has a larger lower bound based on Chebyshev inequality than $\prob{|X_\lambda - \mu_X| < \epsilon}$ does
	if $(2-\delta)/(1-\delta)<q_{T_\lambda}$.
\end{lemma}
\begin{proof}
	Using the Chebyshev inequality, we obtain
	$
	\prob{|Y_\lambda-\mu_Y| < \sigma_Y^2} \geq 1-{1}/{\sigma_Y^2}
	$, and thus $\prob{|Y_\lambda-\mu_X| < \epsilon} \geq 1-{1}/{\sigma_Y^2}$ since $\mu_X=1$ and $\mu_Y\leq \mu_X$. Similarly, the following also holds:
	$
	\prob{|X_\lambda-\mu_X| < \epsilon} \geq 1-{\sigma_X^2}/{\epsilon^2}.
	$
	
	Applying the desired lower bound condition $1 - {1}/{\sigma_Y^2} > 1 - {\sigma_X^2}/{\epsilon^2}$, we get 	 \begin{align}\label{eq:concent_prob1}
		1-\phi(b_\lambda) > {1}/{\left(q_{T_\lambda}-1\right)},
	\end{align}
	where $\phi(i) = \delta^{B-i+1}$.
	%	which is equivalent to ${1}/{\sigma_Y^2} < {\sigma_Z^2}/{\epsilon^2}$.
	%We find the condition under which \eqref{eq:concent_prob1} holds.
	%	\ie~\eqref{eq:our_concent_prob} holds.
	Since $1-\phi(b_\lambda) \geq 1-\delta$ regardless of $b_\lambda$,  \eqref{eq:concent_prob1} has a sufficient condition of $1-\delta > 1/(q_{T_\lambda}-1)$. Rearranging terms, we obtain $(2-\delta)/(1-\delta)<q_{T_\lambda}$.
	\qed
\end{proof}

%%For a small $\delta\leq 1/2$, only triangles appearing between $M$ and $\sqrt{3}M$ violate the condition of \lemref{lem:our_concent_prob}: for a triangle $\lambda$ with $T_\lambda > \sqrt{3}M$, $q_{T_\lambda} <1/3 \leq (1-\delta)/(2-\delta)$.
%Note that an interval for violating triangles gets larger for a larger $\delta$.
%This is because increasing $\delta$ makes $(1-\delta)/(2-\delta)$ smaller, and thus much smaller $q_\lambda$ is required, \ie~violating triangles contain $\lambda$ which appears much later.
\chngwww{Consequently, \methodcx provides lower error compared to \methodc, which is also shown in the experimental results at \secref{sec:experiment}.
}

\subsubsection{\methodmxrm}
\methodmx improves \methodm by combining past estimations with the current one.
%
%Different from \methodcx, \methodmx considers the bucket $0$ to be $[1,T_M]$ instead of $[1,M]$.
%
%$T_M$ is the last time where $u(T) = M$, \ie $u(T) = M$ and $u_{T+1} = M + 1$.
%
We present two versions of \methodmx for binary and weighted triangle counting since \methodm considers both of them.

%\subsubsection{Binary Counting}
\paragraph{Binary Counting (\methodbx)}
We propose \methodbx which \chngwww{improves} the binary local triangle counting algorithm \methodb.
\chngwww{
\methodbx first goes through updating triangle estimation and sampling process as in \methodb, and then goes through $WeightedAverage$ as shown in Algorithm~\ref{alg:methodbx}.
%Note that $SampleNewEdge\text{-MX}_B$ is different from $SampleNewEdge\text{-M}_B$ in only \lineref{line:exactcnt_off_binary}.
$T_M$, the last time when $c^{(t)}$ equals the true triangle count,
is the time when $|D|=M$ and thus $ExactCnt$ becomes $false$.
The right-hand side term of line 12 is $(T-1)$ because \methodbx samples edge first and then updates triangle estimation.
}
%The right hand side is $T-1$, not $T$ because \methodb first samples edge and then updates triangle estimation and thus \methodb has the exact triangle count only when $ExactCnt$ is true. When $ExactCnt$ turns false in line ~\ref{line:exactcnt_off}, $T_M$ should be determined as $T-1$.
%The right hand side is $T-1$, not $T$ because the moment $ExactCnt$ turns false, triangles cannot be updated exactly since \methodm(binary) first samples edge and then updates triangle estimation.

%$T_M$, the last time when $c^{(t)}$ equals the true triangle count, is a little different with that of \methodcx because \methodmx(binary) first samples the edge and then update triangle. Note that with regard to the pseudocode in Algorithm ~\algref{alg:methodc}, $T_M \leftarrow T-1$ should be inserted in line ~\ref{Tm}.
%
%\methodmx is same with \methodm except the subroutines $\Process$ and has two additional subroutines $\WeightAvg$ and $\Query$ which are similar to those of \methodcx.

\begin{algorithm}[t!]
	\SetKwProg{myproc}{Subroutine}{}{end}
	\SetKwFunction{SampleMXB}{$SampleNewEdge\text{-MX}_B$}
	\KwIn{Graph stream $S$, maximum number $M$ of edges stored, interval $J$ for updating estimation, and decaying factor $\delta$ for past estimation}
	\KwOut{Local triangle estimation $\tau = \Query(\delta)$}
	
	\medskip
	
	%\nonl The same as \methodm (binary) except for the following.  \;
	Initialize variables as in lines 1-4 of \algref{alg:methodc}.\;
	\ForEach(~~\tcp*[h]{at time $T$}){edge $e=(u,v)$ from $S$}{
		$\DiscoverNode(u)$.\;
		$\DiscoverNode(v)$.\;
		\If{$e \notin D$}{ \label{line:not_in_buffer2}
		$\SampleMXB(e)$.\;
		$\UpdateMB(e)$. 	//	$c$ is computed here.\;
		}
		$\WeightAvg(\delta)$.\;
	}
	
%\begin{comment}	
%	\smallskip
%	\myproc{$\WeightAvg{$\delta$}$}{
%		\If{$T = T_M$}{
%			$\hat{\tau} \leftarrow c$.
%		}\ElseIf{$T > T_M$ and $(T-T_M)\bmod{J} = 0$}{
%		$\hat{\tau} \leftarrow \delta\hat{\tau} + (1-\delta)c$.
%		%			 \UpdateEstimation{$\delta$}$.\;
%	}
%}
%
%\smallskip
%\myproc{$\Query{$\delta$}$}{
%	\lIf{$T\leq T_M$}{
%		\Return $c$.
%	}\lElseIf{$(T-T_M)\bmod{J} = 0$}{
%	\Return $\hat{\tau}$.
%}\lElse {
%\Return $\delta\hat{\tau} + (1-\delta)c$.
%%	\Return $\UpdateEstimation{$\delta$}$.
%}
%}
%
%
%\end{comment}

\myproc{$\SampleMXB{e}$}{
	\lIf{$|D|<M$}{
		Append $e$ to D.
	}\Else{
		\If{$ExactCnt=true$}{
			$T_M \leftarrow T-1$.\; \label{line:exactcnt_off_binary}
			$ExactCnt \leftarrow false$.
		}
	$\ReplaceM(e)$.\;	
	}
}

\caption{\methodbx for multigraph stream (binary)}
\label{alg:methodbx}
\end{algorithm}

We analyze the expectation and the variance of each triangle appearing in the stream and show that \methodbx gives more accurate results closer to the ground truth than \methodb does.  % below.
Note that for triangles in the $0$th bucket, \ie~for the time $u(T) \leq M$, \methodbx provides the exact counting. Thus, we consider a triangle $\lambda$ with $b_\lambda>0$ below.
Let $Y_\lambda$ and $X_\lambda$ be the estimated counts of a triangle $\lambda$ by \methodbx and \methodb, respectively.
\begin{lemma} \label{lem:our_expect2}
	%	Let $Y_\lambda$ be the amount of contribution of a triangle $\lambda$ to the triangle counts of its three nodes by \methodcx.
	Let \mchng{$b_\lambda$ be the bucket where $\lambda$ is formed} and $Y_\lambda$ be the estimated count of a triangle $\lambda$ \mchng{with $b_\lambda>0$} by \methodbx.
	For every triangle $\lambda$,
	$$
	\expect{Y_\lambda} = 1 - \phi(b_\lambda),
	$$
	where $\phi(i) = \delta^{B-i+1}$ and $B$ is the current bucket.
\end{lemma}
\begin{proof}
	The lemma is proved in the same way as in Lemma~\ref{lem:our_expect}.
	\qed
\end{proof}

\begin{lemma} \label{lem:our_var2}
	Let \mchng{$b_\lambda$ be the bucket where $\lambda$ is formed} and $Y_\lambda$ be the estimated count of a triangle $\lambda$ \mchng{with $b_\lambda>0$} by \methodbx.
	For every triangle $\lambda$,
	$$
	\var{Y_\lambda} = \left(1-\phi(b_\lambda)\right) ^2
	\left( k_{T_\lambda} - 1 \right),
	$$
	where $T_\lambda$ is the first time all three edges of $\lambda$ arrive,
	$k_{T_\lambda} = \frac{(M-3)(u(T_\lambda)-3)(u(T_\lambda)-4)(u(T_\lambda)-5)}{M(M-4)(M-5)(M-6)}$,
	$\phi(i) = \delta^{B-i+1}$, and $B$ is the current bucket.
\end{lemma}
\begin{proof}
	The lemma is proved in the same way as in Lemma~\ref{lem:our_var}.
	\qed
\end{proof}
%\begin{proof}
%	Following the definition $\var{Y_\lambda} = \expect{Y_\lambda^2} - \expect{Y_\lambda}^2$ with \lemref{lem:our_expect2}, the proof is done.
%\end{proof}
%Note that the expectation and the variance of \methodmx are smaller than those of \methodm that has $\expect{Z_\lambda}=1$ and $\var{Z_\lambda}=q_\lambda^{-1}-1$, respectively, where $Z_\lambda$ is the estimated count of $\lambda$ by \methodc.

%The estimation of \methodmx is more concentrated around the ground truth than that of \methodm.
%
%Note that we consider a triangle $\lambda$ with $b_\lambda>0$ below since \methodmx provides unbiased estimation with variance $0$ for the time $T\leq T_M$.

\begin{lemma} \label{lem:our_accuracy2}
	Consider any triangle $\lambda$ that is counted at time $T_\lambda > T_M$.
	If $u(T_\lambda) \geq \sqrt[3]{2}M + 3$,
	the interval by $\expect{Y_\lambda} \pm \var{Y_\lambda}$ is strictly included in that by $\expect{X_\lambda} \pm \var{X_\lambda}$.
\end{lemma}
\begin{proof}
	We first show that $\expect{Y_\lambda} - \var{Y_\lambda} > \expect{X_\lambda} - \var{X_\lambda}$.
	
	Let $\psi_\lambda=1-\delta^{B-b_\lambda+1}$; we will find the condition satisfying
	$
	\psi_\lambda - \psi_\lambda^2\left(k_{T_\lambda}-1\right) - 2 + k_{T_\lambda} > 0,
	$
	which is developed as follows.
	\begin{align}
	\psi_\lambda &> \frac{\left(2-k_{T_\lambda}\right)}{\left(k_{T_\lambda}-1\right)}  \label{eq:our_acc_main_ineq2}
	\end{align}
	
	Below, we will show the condition under which \eqref{eq:our_acc_main_ineq2} holds.
	
	Let $u(T_\lambda) - 3 = \alpha M$.
	By definition,
	\begin{align*}
	k_{T_\lambda} &= \frac{(M-3)(u(T_\lambda)-3)(u(T_\lambda)-4)(u(T_\lambda)-5)}{M(M-4)(M-5)(M-6)} \\
	&> \frac{(M-3)(u(T_\lambda)-3)^3}{M(M-4)^3} > \frac{(u(T_\lambda)-3)^3}{M(M-4)^2} \\
	&> \frac{(u(T_\lambda)-3)^3}{M^3} = \alpha^3.
	\end{align*}
	
	Then,
	\begin{align} \label{eq:our_acc_q2}
	\frac{2-k_{T_\lambda}}{k_{T_\lambda}-1} &<
	\frac{2-\alpha^3}{\alpha^3-1} = \frac{1}{\alpha^3-1} - 1.
	\end{align}
	
	Now we examine the left term of \eqref{eq:our_acc_main_ineq2}. Since $1\leq b_\lambda\leq B$, the lower bound of $\psi_\lambda$ becomes
	\begin{align} \label{eq:our_acc_p_lb2}
	\psi_\lambda \geq 1 - \delta.
	\end{align}
	
	Putting \eqref{eq:our_acc_q2} and~\eqnumref{eq:our_acc_p_lb2} into \eqref{eq:our_acc_main_ineq2}, we obtain a sufficient condition for \eqnumref{eq:our_acc_main_ineq2} as follows:
	\begin{align}
	%		1 - \delta &> \frac{1}{2r + r^2} - 1 \\
	\alpha^3 &\geq \frac{3-\delta}{2-\delta} \label{eq:our_acc_final_ineq2}
	\end{align}
	
	For $\alpha \geq \sqrt[3]{2}$,
	\eqref{eq:our_acc_final_ineq2} always holds since the upper bound of the right term becomes $2$ due to $0\leq \delta<1$.
	Thus, we finally obtain the condition under which \eqref{eq:our_acc_main_ineq2} holds as follows:
	$$
	u(T_\lambda) \geq \sqrt[3]{2}M + 3.
	$$
	
	Due to the underestimation of \methodbx, $\expect{Y_\lambda} + \var{Y_\lambda} < \expect{X_\lambda} + \var{X_\lambda}$ always holds, which completes the proof.
	%	this condition also satisfies $\expect{Y_\lambda} + \var{Y_\lambda} < \expect{Z_\lambda} + \var{Z_\lambda}$, which completes the proof.
	\qed
\end{proof}
%Below, we state that \methodmx guarantees a larger lower bound of the probability that its estimation is around the true value than \methodm does.

\begin{lemma}\label{lem:our_concent_prob2}
	Given a triangle $\lambda$ with $T_\lambda>T_M$,
	let $\mu_Y = \expect{Y_\lambda}, \sigma_Y^2 = \var{Y_\lambda}$, and $\epsilon = \mu_X - \mu_Y+\sigma_Y^2$ where $\mu_X = \expect{X_\lambda} = 1$ is the unbiased estimate of the triangle count for $\lambda$.
	Then,
	$\prob{|Y_\lambda - \mu_X| < \epsilon}$ has a larger lower bound based on Chebyshev inequality than $\prob{|X_\lambda - \mu_X| < \epsilon}$ does
	if $(2-\delta)/(1-\delta)<k_{T_\lambda}$.
\end{lemma}
\begin{proof}
	The lemma is proved in the same way as in Lemma~\ref{lem:our_concent_prob}.
	\qed
\end{proof}

\paragraph{Weighted Counting (\methodwx)}
We propose \methodwx which \chngwww{improves} the weighted local triangle counting algorithm \methodw. %For each arriving edge, triangle estimation is computed by the same process of \methodw.
%
%\methodwx first performs the steps of \methodw and then goes through $WeightedAverage$ as in Algorithm~\ref{alg:methodcx}.
%Insert $T_M \leftarrow T$ in ~\algref{alg:methodc} line ~\ref{line:exactcnt_off} to determine $T_M$, the last moment we have true triangle count.
%
\chngwww{
\methodwx(\algref{alg:methodwx}) first goes through updating triangle estimation and sampling process as in \methodw, and then combines past estimations with the current one.
%Note that $SampleNewEdge\text{-MX}_W$ is different from $SampleNewEdge\text{-M}_W$ in only \lineref{line:exactcnt_off_binary}.
$T_M$, the last time when $c^{(t)}$ equals the true triangle count,
is the time when $|D|=M$ and $ExactCnt=true$, as in \methodcx.
}

\methodwx gives more accurate results closer to the ground truth than \methodw does. The detailed analysis of the expectation and the variance of \methodwx is very similar to that of \methodbx.
Note that for triangles in the $0$th bucket, \ie~for the time $u(T) \leq M+1$,\methodwx provides the exact counting. Thus, we consider a triangle $\lambda$ with $b_\lambda>0$ below.
Let $Y_\lambda$ and $X_\lambda$ be the estimated counts of a triangle $\lambda$ by \methodwx and \methodw, respectively.
\begin{lemma} \label{lem:our_expect3}
	%	Let $Y_\lambda$ be the amount of contribution of a triangle $\lambda$ to the triangle counts of its three nodes by \methodcx.
	Let \mchng{$b_\lambda$ be the bucket where $\lambda$ is formed} and $Y_\lambda$ be the estimated count of a triangle $\lambda$ \mchng{with $b_\lambda>0$} by \methodwx.
	For every triangle $\lambda$,
	$$
	\expect{Y_\lambda} = 1 - \phi(b_\lambda),
	$$
	where $\phi(i) = \delta^{B-i+1}$ and $B$ is the current bucket.
\end{lemma}
\begin{proof}
	The lemma is proved in the same way as in Lemma~\ref{lem:our_expect}.
	\qed
\end{proof}

\begin{lemma} \label{lem:our_var3}
	Let \mchng{$b_\lambda$ be the bucket where $\lambda$ is formed} and $Y_\lambda$ be the estimated count of a triangle $\lambda$ \mchng{with $b_\lambda>0$} by \methodwx.
	For every triangle $\lambda$,
	$$
	\var{Y_\lambda} = \left(1-\phi(b_\lambda)\right) ^2
	\left( l_{T_\lambda} - 1 \right),
	$$
	where $T_\lambda$ is the first time all three edges of $\lambda$ arrive,
	$l_{T_\lambda} = \frac{(M-2)(u(T_\lambda-1)-2)(u(T_\lambda-1)-3)}{M(M-3)(M-4)}$,
	$\phi(i) = \delta^{B-i+1}$, and
	$B$ is the current bucket.
\end{lemma}
\begin{proof}
	The lemma is proved in the same way as in Lemma~\ref{lem:our_var}.
	\qed
\end{proof}
%\begin{proof}
%	Following the definition $\var{Y_\lambda} = \expect{Y_\lambda^2} - \expect{Y_\lambda}^2$ with \lemref{lem:our_expect2}, the proof is done.
%\end{proof}
%Note that the expectation and the variance of \methodmx are smaller than those of \methodm that has $\expect{Z_\lambda}=1$ and $\var{Z_\lambda}=q_\lambda^{-1}-1$, respectively, where $Z_\lambda$ is the estimated count of $\lambda$ by \methodc.

%The estimation of \methodmx is more concentrated around the ground truth than that of \methodm.
%
%Note that we consider a triangle $\lambda$ with $b_\lambda>0$ below since \methodmx provides unbiased estimation with variance $0$ for the time $T\leq T_M$.

\begin{lemma} \label{lem:our_accuracy3}
	Consider any triangle $\lambda$ that is counted at time $T_\lambda > T_M$.
	If $u(T_\lambda-1) \geq \sqrt{2}M + 2$,
	the interval by $\expect{Y_\lambda} \pm \var{Y_\lambda}$ is strictly included in that by $\expect{X_\lambda} \pm \var{X_\lambda}$.
\end{lemma}
\begin{proof}
	We first show that $\expect{Y_\lambda} - \var{Y_\lambda} > \expect{X_\lambda} - \var{X_\lambda}$.
	
	Let $\psi_\lambda=1-\delta^{B-b_\lambda+1}$; we will find the condition satisfying
	$
	\psi_\lambda - \psi_\lambda^2\left(l_{T_\lambda}-1\right) - 2 + l_{T_\lambda} > 0,
	$
	which is developed as follows.
	\begin{align}
	\psi_\lambda &> \frac{\left(2-l_{T_\lambda}\right)}{\left(l_{T_\lambda}-1\right)}  \label{eq:our_acc_main_ineq3}
	\end{align}
	
	Below, we will show the condition under which \eqref{eq:our_acc_main_ineq3} holds.
	
	Let $u(T_\lambda-1) - 2 = \alpha M$.
	By definition,
	\begin{align*}
	l_{T_\lambda} &= \frac{(M-2)(u(T_\lambda-1)-2)(u(T_\lambda-1)-3)}{M(M-3)(M-4)} \\
	&> \frac{(M-2)(u(T_\lambda-1)-2)^2}{M(M-3)^2} > \frac{(u(T_\lambda-1)-2)^2}{M(M-3)} \\
	&> \frac{(u(T_\lambda-1)-2)^2}{M^2} = \alpha^2.
	\end{align*}
	
	Then,
	\begin{align} \label{eq:our_acc_q3}
	\frac{2-l_{T_\lambda}}{l_{T_\lambda}-1} &<
	\frac{2-\alpha^2}{\alpha^2-1} = \frac{1}{\alpha^2-1} - 1.
	\end{align}
	
	Now we examine the left term of \eqref{eq:our_acc_main_ineq3}. Since $1\leq b_\lambda\leq B$, the lower bound of $\psi_\lambda$ becomes
	\begin{align} \label{eq:our_acc_p_lb3}
	\psi_\lambda \geq 1 - \delta.
	\end{align}
	
	Putting \eqref{eq:our_acc_q3} and~\eqnumref{eq:our_acc_p_lb3} into \eqref{eq:our_acc_main_ineq3}, we obtain a sufficient condition for \eqnumref{eq:our_acc_main_ineq3} as follows:
	\begin{align}
	%		1 - \delta &> \frac{1}{2r + r^2} - 1 \\
	\alpha^2 &\geq \frac{3-\delta}{2-\delta} \label{eq:our_acc_final_ineq3}
	\end{align}
	
	For $\alpha \geq \sqrt{2}$,
	\eqref{eq:our_acc_final_ineq3} always holds since the upper bound of the right term becomes $2$ due to $0\leq \delta<1$.
	Thus, we finally obtain the condition under which \eqref{eq:our_acc_main_ineq3} holds as follows:
	$$
	u(T_\lambda-1) \geq \sqrt{2}M + 2.
	$$
	
	Due to the underestimation of \methodwx, $\expect{Y_\lambda} + \var{Y_\lambda} < \expect{X_\lambda} + \var{X_\lambda}$ always holds, which completes the proof.
	%	this condition also satisfies $\expect{Y_\lambda} + \var{Y_\lambda} < \expect{Z_\lambda} + \var{Z_\lambda}$, which completes the proof.
	\qed
\end{proof}
%Below, we state that \methodmx guarantees a larger lower bound of the probability that its estimation is around the true value than \methodm does.

\begin{lemma}\label{lem:our_concent_prob3}
	Given a triangle $\lambda$ with $T_\lambda>T_M$,
	let $\mu_Y = \expect{Y_\lambda}, \sigma_Y^2 = \var{Y_\lambda}$, and $\epsilon = \mu_X - \mu_Y+\sigma_Y^2$ where $\mu_X = \expect{X_\lambda} = 1$ is the unbiased estimate of the triangle count for $\lambda$.
	Then,
	$\prob{|Y_\lambda - \mu_X| < \epsilon}$ has a larger lower bound based on Chebyshev inequality than $\prob{|X_\lambda - \mu_X| < \epsilon}$ does
	if $(2-\delta)/(1-\delta)<l_{T_\lambda}$.
\end{lemma}
\begin{proof}
	The lemma is proved in the same way as in Lemma~\ref{lem:our_concent_prob}.
	\qed
\end{proof}

%the following lemmas show that \methodwx gives more accurate results closer to the ground truth than \methodw does.

%Note that $Y_\lambda$ and $X_\lambda$ are the estimated counts of a triangle $\lambda$ by \methodwx and \methodw, respectively.

%Note that subroutines $\WeightAvg$ and $\Query$ are same with those in \algref{alg:methodmx}.

\begin{algorithm}[t!]	
	\SetKwProg{myproc}{Subroutine}{}{end}
	\SetKwFunction{SampleMXW}{$SampleNewEdge\text{-MX}_W$}
	\KwIn{Graph stream $S$, maximum number $M$ of edges stored, interval $J$ for updating estimation, and decaying factor $\delta$ for past estimation}
	\KwOut{Local triangle estimation $\tau = \Query(\delta)$}
	
	\medskip
	
	%\nonl The same as \methodm (binary) except for the following.  \;
	Initialize variables as in lines 1-4 of \algref{alg:methodc}.\;
	\ForEach(~~\tcp*[h]{at time $T$}){edge $e=(u,v)$ from $S$}{
		$\DiscoverNode(u)$.\;
		$\DiscoverNode(v)$.\;
		$\UpdateMW(e)$. 	//	$c$ is computed here.\;
		$\SampleMXW(e)$.\;
		$\WeightAvg(\delta)$.\;
	}
	
%	\begin{comment}	
%	\smallskip
%	\myproc{$\WeightAvg{$\delta$}$}{
%	\If{$T = T_M$}{
%	$\hat{\tau} \leftarrow c$.
%	}\ElseIf{$T > T_M$ and $(T-T_M)\bmod{J} = 0$}{
%	$\hat{\tau} \leftarrow \delta\hat{\tau} + (1-\delta)c$.
%	%			 \UpdateEstimation{$\delta$}$.\;
%	}
%	}
%	
%	\smallskip
%	\myproc{$\Query{$\delta$}$}{
%	\lIf{$T\leq T_M$}{
%	\Return $c$.
%	}\lElseIf{$(T-T_M)\bmod{J} = 0$}{
%	\Return $\hat{\tau}$.
%	}\lElse {
%	\Return $\delta\hat{\tau} + (1-\delta)c$.
%	%	\Return $\UpdateEstimation{$\delta$}$.
%	}
%	}
%	
%	\end{comment}
	
	\myproc{$\SampleMXW{e}$}{
		\lIf{$e \in D$}{
			$O_e \leftarrow O_e+1$.
		}\lElse{	
			\lIf{$|D|<M$}{
				Append $e$ to D.
			}\Else{
			\If{$ExactCnt=true$}{
				$T_M \leftarrow T$.\; %\label{line:exactcnt_off_binary}
				$ExactCnt \leftarrow false$.
			}
			$\ReplaceM(e)$.\;	
			}
		}
	}
\caption{\methodwx for multigraph stream (weighted)}
\label{alg:methodwx}
\end{algorithm}

\hide{
	
	\begin{corollary}
		For every triangle $\lambda$,
		$$
		\expect{Y_\lambda} + \var{Y_\lambda} < \expect{X_\lambda} + \var{X_\lambda}.
		$$
	\end{corollary}
	\begin{proof}
		\methodcx provides underestimation, \ie~$\expect{Y_\lambda}\leq \expect{X_\lambda}$, which completes the proof with \lemref{lem:our_accuracy}.
	\end{proof}
	
	\begin{align*}
		\frac{\expect{X_\lambda} - \expect{Y_\lambda}}{\var{X_\lambda}} &= \frac{1-p}{q_\lambda^{-2}-1} \\
		\frac{\var{Y_\lambda}}{\var{X_\lambda}} &= p^2.
	\end{align*}
	Below, we verify the condition under which the following inequality holds:
	\begin{align*}
		\frac{\expect{X_\lambda} - \expect{Y_\lambda} + \var{Y_\lambda}}{\var{X_\lambda}} < 1.
	\end{align*}
}

\hide{

	\begin{lemma}
		For every triangle $\lambda$,
		$$
		\expect{Y_\lambda} - \var{Y_\lambda} > \expect{X_\lambda} - \var{X_\lambda}.
		$$
	\end{lemma}
	\begin{proof}
		We first remind the expectation and the variance of $X_\lambda$ as follows.
		\begin{align*}
			\expect{X_\lambda} &= 1. \\
			\var{X_\lambda} &= q_\lambda^{-2} - 1.
		\end{align*}
		Let $p=(1-\phi(i)/2^B)$; then we obtain the following equations:
		\begin{align*}
			\expect{X_\lambda} - \expect{Y_\lambda} &= q_\lambda^{-2} -  q_\lambda^{-2}p = \frac{q_\lambda^{-2}-q_\lambda^{-2}p}{q_\lambda^{-2}-1} \var{X_\lambda} \\
			\var{Y_\lambda} &= p^2 \var{X_\lambda}.
		\end{align*}
		Below, we verify the condition under which the following inequality holds:
		\begin{align*}
			\frac{\expect{X_\lambda} - \expect{Y_\lambda} + \var{Y_\lambda}}{\var{X_\lambda}} < 1.
		\end{align*}
		
		We obtain the series of the following inequalities.
		\begin{align*}
			\frac{q_\lambda^{-2}- q_\lambda^{-2}p}{q_\lambda^{-2}-1} - p^2 &< 1  \\
			q_\lambda^{-2} - q_\lambda^{-2}p + q_\lambda^{-2}p^2 - p^2 &< q_\lambda^{-2}-1\\
			q_\lambda^{-2}p - q_\lambda^{-2}p^2 + p^2 &> 1 \\
			p\left( 1-q_\lambda^{-2}\right) + q_\lambda^{-2} &> p^{-1} \\
			p\left( 1-q_\lambda^{-2}\right) &> p^{-1} - q_\lambda^{-2}\\
			p &< \frac{p^{-1} - q_\lambda^{-2}}{1-q_\lambda^{-2}}.
		\end{align*}
		As shown in the last inequality, the condition is $p^{-1}-q_\lambda^{-2}<0$. That is,
		$$
		p > q_\lambda^2.
		$$
		We set $J = \alpha M$; then
		\begin{align*}
			1-\frac{\phi(i)}{2^B} > \frac{M^2}{(i\alpha M - 1)^2} > \frac{M^2}{(i\alpha M)^2} = \frac{1}{i^2\alpha^2}
		\end{align*}
		For $i\geq 1$, the following holds:
		\begin{align*}
			\frac{2^{B-i+1}}{i^2(2^{B-i+1} - 1)} = \frac{2^B}{i^2(2^B - 2^{i-1})} < \alpha^2.
		\end{align*}
		Setting $\alpha=1$, the inequality always holds for $i\geq 2$. As a result, $M\leq 2|E|$ satisfies the lemma with $\alpha = 1$.
	\end{proof}

}

\hide{
	
	\subsection{\methodcxrm}
	Considering triangles certainly sampled in the past, \methodc reduces variance compared with \methodc.
	But, \methodc still misses some information, \ie~estimation itself in the past.
	In this section, we propose the main algorithm \methodcx by considering edge samples at a certain time in the past to raise accuracy.
	
	Before describing \methodcx in detail, we first consider the following case.
	Let $D_1$ and $D_2$ be two arrays each of which has of size $M$.
	Running \methodc with $D_1$ and $D_2$ independently and simultaneously for the same stream, we obtain two estimation $\tau_1$ and $\tau_2$. Then, $\tau=(\tau_1 + \tau_2)/2$ is also unbiased estimation, and moreover $\var{\tau} < \var{\tau_1} = \var{\tau_2}$.
	Through this independent sampling, the variance gets lower as more arrays are used. Precisely, the variance becomes $1/k$ if $k$ arrays of $D_1,\ldots,D_k$ are used.
	
	Maintaining $k$ independent samples, however, results in increasing memory requirements $k$ times. Instead, our \methodcx focuses on dependent samples obtained by the following observation.
	\begin{observation}\label{obs:main_obs}
		Consider that we apply the reservoir sampling to a random sequence of edges $S=(e_1,\ldots,e_T)$ with an array $D$ of size $M$. Let $T$ be the current time.
		Every edge in $D$ at time $t\leq T$ is a sample from $S$ with probability $t/T \times \min(M/t,1)= t/T$ as well as a sample from $\{e_1,\ldots,e_t\}$ with probability $\min(M/t,1)$.
	\end{observation}
	Let $\tau_u^{(t)}$ be the estimation of \methodc at time $t$.
	\obsref{obs:main_obs} states that the following is also an unbiased estimation of $u$ at time $T$:
	$$
	\tau_u^{(T)} = \alpha\times c_{u}^{(t)} (t/T)^{-3} + (1-\alpha)\times \tau_u^{(T)},
	$$
	where $0\leq \alpha\leq 1$. This is because every edge until $t$ is a sample from all edges until $T$ with probability $t/T$. We will discuss the problem  the dependency of samples at the end of the section.
	
	\begin{algorithm}	
		\SetKwProg{myproc}{Subroutine}{}{end}
		\SetKwFunction{UpdateEstimation}{$UpdateEst$}
		\SetKwFunction{Query}{$Query$}
		
		\KwIn{Graph stream $S$, maximum number $M$ of edges stored, and interval $J$ for updating estimation}
		\KwOut{Local triangle estimation $\tau = \Query(\alpha)$}
		
		\medskip
		
		\nonl The same as \methodc except for the following.  \;
		
		\myproc{$\Process{e, i}$}{
			$p \leftarrow \min(M/(T-1),1)$.\;
			Run $\Update(e,+,p^{-2})$ of \methodc.\;
			\uIf{$T = M$}{
				$\hat{\tau} \leftarrow c$.\;
			}\ElseIf{$T > M$ and $T\bmod{J} = 0$}{
			$\hat{\tau} \leftarrow \UpdateEstimation{$\alpha,\min(J,T-M)$}$.\;
		}
	}
	
	\smallskip
	\myproc{$\UpdateEstimation{$\alpha$, L}$}{
		$q \leftarrow \max(T - L, M) / T$. \;
		\Return $\alpha q^{-3} \hat{\tau}  + (1-\alpha)c$.\;
	}
	\smallskip
	\myproc{$\Query{$\alpha$}$}{
		\lIf{$T\leq M$}{
			\Return $c$.
		}\lElseIf{$T\bmod{J} = 0$}{
		\Return $\hat{\tau}$.
	}\lElse {
	\Return $\UpdateEstimation{$\alpha, T\bmod{J}$}$.
}
}

\caption{\methodcx for simple graph stream}
\label{alg:methodv}
\end{algorithm}

Based on this idea, \algref{alg:methodcx} describes \methodcx---except for $ProcessNewEdge$, the procedure is the same as \methodc.
\methodcx does not need to additional estimation for $T\leq M$ since $c$ is the exact count for $G(T)$ where $G(T)$ is the graph consisting of edges appearing until time $T$. The idea using dependent samples is applied per $J$ iterations only if $T>M$. The following lemma states that \methodcx gives unbiased estimation.
\begin{lemma}
	Let $\triangle_u$ be the true local triangle count for a node $u$, and $\tau_u$ be estimation given by \methodcx for a randomly ordered edge sequence. At any time $T$, for every node $u$,
	$$
	\expect{\tau_u} = \triangle_u.
	$$
\end{lemma}
\begin{proof}
	If $T\leq M$, $\tau = c$ is the exact count, and by \algref{alg:methodcx}, no additional estimation is applied.
	
	When \methodcx updates $\hat{\tau}$, let $\hat{\tau}'$ is $\hat{\tau}$ after the update.
	At time $T = M$, $\hat{\tau}' = c$. By  \lemref{lem:methodc_exp}, the following holds:
	$$
	\expect{\hat{\tau}'_u} = \sum_{\lambda\in \Theta_u^{(M)}} \expect{X_\lambda^{(M)}},
	$$
	where $\Theta_u^{(M)}$ is the set of triangles of $u$ formed by $E_M=\{e_1,\ldots,e_M\}$, and $X_\lambda^{(M)}$ indicates whether $\lambda$ is counted or not until $M$.
	
	Consider $\UpdateEstimation$ at
	$T>M$ with $t=\max(T-L,M)$. Let us assume that $\expect{\hat{\tau}_u} = \sum_{\lambda\in \Theta_u^{(t)}} \expect{X_\lambda^{(t)}}$.
	We first rewrite the expectation of $\tau_u$ as follows.
	$$
	\expect{\hat{\tau}_u} = \sum_{\lambda\in \Theta_u^{(T)}} \expect{X_\lambda^{(t)}} Y_\lambda^{(T)},
	$$
	where $Y_\lambda^{(t)}$ indicates whether all edges of $\lambda$ are in $E_t$. Since the edge sequence is randomly ordered, we know that every edge in $E_t$ is a sample from $E_T$ with probability $t/T$: in other words $\expect{Y_\lambda^{(T)}} = (t/T)^3$. Moreover since $X_\lambda^{(t)}$ and $Y_\lambda^{(T)}$ are independent, we obtain
	\begin{align*}
		\expect{q^{-3}\hat{\tau}_u} &=
		\left(\frac{T}{t}\right)^3 \sum_{\lambda\in \Theta_u^{(T)}} \expect{Y_\lambda^{(t)}}\expect{X_\lambda^{(T)}} \\
		&=
		\left(\frac{T}{t}\right)^3 \sum_{\lambda\in \Theta_u^{(T)}} \left(\frac{t}{T}\right)^3 \expect{X_\lambda^{(T)}} \\
		&= \sum_{\lambda\in \Theta_u^{(T)}} \expect{X_\lambda^{(T)}} = \varDelta_u^{(T)},
	\end{align*}
	Hence, for $0\leq \alpha\leq 1$, the following holds.
	$$
	\expect{\hat{\tau}'_u} = \expect{\alpha q^{-3}\tau_u + (1-\alpha)c} = \sum_{\lambda\in \Theta_u^{(T)}} \expect{X_\lambda^{(T)}},
	$$
	which completes the proof.
\end{proof}

The problem is that $c^{(t)}$ and $c^{(T)}$ are not obtained from independent samples. In other words, some selection of $\alpha$ result in increasing variance, which means that \methodcx becomes worse than \methodc.

The following lemma states the condition under which \methodcx has a variance lower than \methodc.

\textbf{Variance analysis.} Let $Y$ be an estimation at time $t_1=t_2-J$ and $Z$ be an estimation at time $t_2$. Let $0\leq \alpha\leq 1$.
\begin{itemize*}
	\item Fact 1: $\varSigma = \var{Y} = \var{Z}$.
	\item Fact 2: $0\leq \alpha\leq 1$.
\end{itemize*}
Let $\beta = 1-\alpha$.
\begin{align*}
	\var{\alpha Y + \beta Z} &= \alpha^2 \var{Y} + \beta^2\var{Z} + 2\alpha\beta\cov{Y,Z}\\
	&= \alpha^2 \varSigma + (1-2\alpha+\alpha^2)\varSigma + 2\alpha\beta\cov{Y,Z}\\
	&= 2\alpha^2\varSigma - 2\alpha\varSigma + \varSigma + 2\alpha\beta\cov{Y,Z}
\end{align*}
The covariance is calculated as follows.
\begin{align*}
	\cov{Y,Z} &= \expect{YZ} - \expect{Y}\expect{Z} \\
	Y &= \sum_{\lambda=(x,y,z)\in \Theta_u}  \frac{X_{\lambda}}{q_\lambda^2}  \\
	Z &= \sum_{\lambda=(x,y,z)\in \Theta_u}  \frac{X_{\lambda}}{q_\lambda^2}  \\
	\expect{YZ} &= \sum_{\substack{
			\lambda=(x_1,y_1,z_1)\in \Theta_u\\
			\gamma=(x_2,y_2,z_2)\in \Theta_u}}
	\frac{\expect{X_\lambda X_\gamma}}{q_\lambda^2 q_\gamma^2}  \\
	\expect{YZ}-\expect{Y}\expect{Z} &= \sum_{\substack{
			\lambda=(x_1,y_1,z_1)\in \Theta_u\\
			\gamma=(x_2,y_2,z_2)\in \Theta_u}}
	\left(\frac{\expect{X_\lambda X_\gamma}}{q_\lambda^2 q_\gamma^2} - 1\right) \\
	& \because \expect{Y}\expect{Z} = \varDelta^2 \text{ and } |\Theta_u| = \varDelta.
\end{align*}

\textbf{\boldmath If $\lambda \cap \gamma = \emptyset$:}
$$
\expect{X_\lambda X_\gamma} = q_\lambda^2 q_\gamma^2.
$$

\textbf{\boldmath If $|\lambda \cap \gamma| = 1$:}
Let $q_\lambda \geq q_\gamma$.
$$
\expect{X_\lambda X_\gamma} \leq q_\lambda q_\gamma^2.
$$

\textbf{\boldmath If $\lambda = \gamma$:}
$$
\expect{X_\lambda X_\gamma} \leq q_\lambda q_\gamma.
$$

\begin{gather*}
	\expect{YZ} = \sum_{\substack{
			\lambda=(x_1,y_1,z_1)\in \Theta_u\\
			\gamma=(x_2,y_2,z_2)\in \Theta_u}} \frac{1}{q_\lambda q_\gamma} \\
	\cov{Y,Z} = \varDelta^2 \left(\frac{1}{q_\pi^2} - 1\right) \\
	\var{\alpha Y + \beta Z} = 2\alpha^2 \varSigma - 2\alpha\varSigma + \varSigma + 2\alpha(1-\alpha)\left(\frac{1}{q_\pi^2}-1\right)\varDelta^2,
\end{gather*}
where $q_\pi$ is the latest triangle of $u$.

\textbf{Error.}
Consider a edge sequence randomly shuffled. Let $\lambda$ be a triangle at time $t$, $r_\lambda^{(i)}$ be a weight of $\lambda$ when it appears at the $i$th bucket. Then, we obtain
$$
r_{(i)} = \frac{1}{p(i)^2}\left(1 - \frac{2^i-1}{2^b-1}\right)
$$
Let $\tau$ is the weight of a triangle given by the algorithm; then its expectation becomes
\begin{align*}
	\expect{\tau} &= \frac{1}{b}\sum_{i=1}^b \frac{1}{p(i)^2}\left(1 - \frac{2^i-1}{2^b-1}\right) \\
	&=  \frac{1}{b} \left(b - \sum_{i=1}^b \frac{1}{p(i)^2}\frac{2^i-1}{2^b-1}\right).
\end{align*}
Note that $q_{(i-1)*J}\leq p(i) \leq q_{iJ}$. Consider upper bound first.
\begin{align*}
	\expect{\tau} &\leq \frac{1}{b} \left(b - \sum_{i=1}^b \frac{1}{q_{iJ}^2}\frac{1}{2^{b-i}}\right). \\
	&\leq \left(1 - \frac{J^2}{bM^2}\sum_{i=1}^b \frac{i^2}{2^{b-i}}\right). \\
	&= \left(1 - \frac{J^2}{bM^22^b}\sum_{i=1}^b i^22^i\right). \\
\end{align*}
Also the lower bound is
\begin{align*}
	\expect{\tau} &\geq \frac{1}{b} \left(b - \sum_{i=1}^b \frac{1}{q_{(i-1)J}^2}\frac{1}{2^{b-i}}\right). \\
	&\geq \left(1 - \frac{J^2}{bM^2}\sum_{i=1}^b \frac{(i-1)^2}{2^{b-i}}\right). \\
	&= \left(1 - \frac{J^2}{bM^22^b}\sum_{i=1}^b (i-1)^22^i\right). \\
	&= \left(1 - \frac{2 J^2}{bM^22^{b}}\sum_{i=1}^b (i-1)^22^{i-1}\right). \\
\end{align*}

First the last summation is solved by the following equation.
\begin{align*}
	\sum_{i=1}^n &= k^2 z^k \\
	&= z\frac{1+z  - (n+1)^2 z^n + (2n^2+2n-1)z^{n+1} - n^2z^{n+2}}{(1-z)^3}
\end{align*}
With $z=2$, we obtain the following.
\begin{align*}
	\sum_{i=1}^b &= k^2 2^k \\
	&= -2\left(3  - (b+1)^2 2^n + (2b^2+2b-1)2^{b+1} - b^2 2^{b+2}\right) \\
	&= 2\left(b^2 2^{b+2} - (2b^2+2b-1)2^{b+1} + (b+1)^2 2^b - 3\right) \\
	&= 2^{b+1}\left(2^{2} b^2 - 2(2b^2+2b-1) + (b+1)^2 \right) - 6\\
	&= 2^{b+1}\left(2^{2} b^2 - 2(2b^2+2b) + (b+1)^2 \right) + 2^{b+2}-6\\
	&= 2^{b+1} (b-1)^2 + 2^{b+2}-6.
\end{align*}

The following holds:
\begin{align*}
	\frac{2^{b+1}(b-1)^2 + 2^{b+2} - 6}{b2^b} &= \frac{2(b-1)^2 + 4}{b} - \frac{6}{2b^2} \\
	&= \frac{2b^2 - 4b + 6}{b} - \frac{3}{b^2} \\
	&= 2b - 4 + \frac{6}{b} - \frac{3}{b^2} \\
	&\geq 2b - 4 + \frac{6b - 3}{b^2} \\
	&\geq 2b - 1 + \frac{-3b^2 + 6b - 3}{b^2} \\
	&\geq 2b - 1 - \frac{3(b-1)^2}{b^2} \\
	&\geq 2b - 4 ~~ \left(\because \frac{(b-1)^2}{b^2}< 1\right) \\
	&\geq b-2
\end{align*}
Also
\begin{align*}
	\frac{2^{b+1}(b-1)^2 + 2^{b+2} - 6}{b2^b} &= \frac{2(b-1)^2 + 4}{b} - \frac{6}{2b^2} \\
	&= \frac{2b^2 - 4b + 6}{b} - \frac{3}{b^2} \\
	&= 2b - 4 + \frac{6}{b} - \frac{3}{b^2} \\
	&= 2b - \frac{4b^2 - 6b + 3}{b^2}\\
	&= 2b - \frac{b^2 + 3b^2 - 6b + 3}{b^2} \\
	&= 2b - \frac{b^2 + 3(b-1)^2}{b^2} \\
	&\leq 2b \\
\end{align*}

Then, the following holds:
\begin{align*}
	\expect{\tau} &\leq
\end{align*}
}

\hide{
	
	We construct the following three sets:
	\begin{align*}
		\varLambda^{(i)} &= \{(\lambda,\gamma)\in \Theta_u^2 : |\lambda \cap \gamma| = i \}.
	\end{align*}
	Note that $0\leq i\leq 2$.
	We rewrite $\expect{YZ}$ as follows.
	\begin{align*}
		\expect{YZ} &=
		\sum_{(\lambda,\gamma)\in \varLambda^{(0)}}
		\frac{(T_{z_1}-1)^2(T_{z_2}-1)^2}{M^4} \expect{X_\lambda X_\gamma} \\
		&+	\sum_{(\lambda,\gamma)\in \varLambda^{(1)}}
		\frac{(T_{z_1}-1)^2(T_{z_2}-1)^2}{M^4} \expect{X_\lambda X_\gamma} \\
		&+	\sum_{(\lambda,\gamma)\in \varLambda^{(2)}}
		\frac{(T_{z_1}-1)^2(T_{z_2}-1)^2}{M^4} \expect{X_\lambda X_\gamma} \\
	\end{align*}

	Since $X_\lambda,X_\gamma\in \{0,1\}$, $\expect{X_\lambda X_\gamma} = \Pr[X_\lambda]\Pr[X_\gamma|X_\lambda]$.
	
	\textbf{\boldmath Case $(\lambda,\gamma)\in \varLambda^{(0)}$}: In this case, both $X_\lambda$ and $X_\gamma$ are independent so that $\expect{X_\lambda X_\gamma} = \expect{X_\lambda}\expect{X_\gamma}$.
	
	\textbf{\boldmath Case $(\lambda,\gamma)\in \varLambda^{(1)}$}:
	\begin{align*}
		\Pr[X_\lambda=1] &= \frac{M^2}{(T_{z_1}-1)^2} \\
		\Pr[X_\gamma=1~|~X_\lambda=1] &= \frac{M}{(T_{z_2}-1)} \frac{T_{z_1}-1}{T_{z_2}-1}
	\end{align*}
	Finally, we obtain
	$$
	\expect{X_\lambda X_\gamma} = \frac{M^3}{(T_{z_1}-1)(T_{z_2}-1)^2} = \sqrt{\expect{X_\lambda}}\expect{X_\gamma}.
	$$
	
	\textbf{\boldmath Case $(\lambda,\gamma)\in \varLambda^{(2)}$}:
	\begin{align*}
		\Pr[X_\lambda=1] &= \frac{M^2}{(T_{z_1}-1)^2} \\
		\Pr[X_\gamma=1~|~X_\lambda=1] &= 1
	\end{align*}
	Finally we obtain
	$$
	\expect{X_\lambda X_\gamma} = \frac{M^2}{(T_{z_1}-1)^2} = \expect{X_\lambda}.
	$$
	
	Now we compute $\expect{YZ}$ as follows.
	\begin{align*}
		\expect{YZ} &=
		\sum_{(\lambda,\gamma)\in \varLambda^{(0)}}
		\expect{X_\lambda} \expect{X_\gamma} \\
		&+	\sum_{(\lambda,\gamma)\in \varLambda^{(1)}}
		\sqrt{\expect{X_\lambda}} \expect{X_\gamma} \\
		&+	\sum_{(\lambda,\gamma)\in \varLambda^{(2)}}
		\expect{X_\lambda} \\
	\end{align*}
	Let $k_u$ be the number of pairs of triangles of $u$ sharing one edge.
	Then,
	\begin{align*}
		\cov{Y,Z} &=
		\sum_{(\lambda,\gamma)\in \varLambda^{(0)}}
		\left( \expect{X_\lambda} \expect{X_\gamma} - \expect{X_\lambda} \expect{X_\gamma} \right)\\
		&+	\sum_{(\lambda,\gamma)\in \varLambda^{(1)}}
		\left( \sqrt{\expect{X_\lambda}} \expect{X_\gamma} - \expect{X_\lambda} \expect{X_\gamma} \right) \\
		&+	\sum_{(\lambda,\gamma)\in \varLambda^{(2)}}
		\left( \expect{X_\lambda} - \expect{X_\lambda} \expect{X_\gamma} \right) \\
	\end{align*}
}

\section{Experiments}
\label{sec:experiment}

In this section, we present experimental results focusing on the following questions:
\begin{itemize*}
	\item[Q1] How accurate is \methodcx for local triangle counting in simple graph stream? (\secref{sec:exp_sim})
	\item[Q2] How accurate are \methodm and \methodmx for local triangle counting in multigraph stream? (\secref{sec:exp_mul})
	%	\item[Q2] How does the performance of \methodcx and \methodmx change with varying parameters? (\secref{sec:exp_eval})
	\item[Q3] What are the patterns discovered from real world graphs by \method? (\secref{sec:exp_anom})
\end{itemize*}

\subsection{Experimental Settings}
\label{sec:exp-setting}

\textbf{Dataset.}
The real world graph datasets used in our experiments are listed in \tabref{tab:datasets}.
For simple graphs,
we remove self-loops, edge direction, and duplicate edges.
For multigraphs, we remove self-loops and edge direction.
We use a random order of edges for graphs without timestamps.
%We also shuffle each graph's edges randomly to make more realistic graph streams.

\noindent \textbf{Parameters.}
We used $J=M$ and $\delta \in \{0.1,0.4,0.7\}$ in \methodcx and \methodmx. In Sections~\ref{sec:exp_sim} and \ref{sec:exp_mul}, we present results with $\delta$ that shows the best performance in each data.

\begin{figure*} [t!]
	\begin{center}
		\subfloat[\facebook]
		 {\includegraphics[width=0.235\textwidth,natwidth=610,natheight=642]{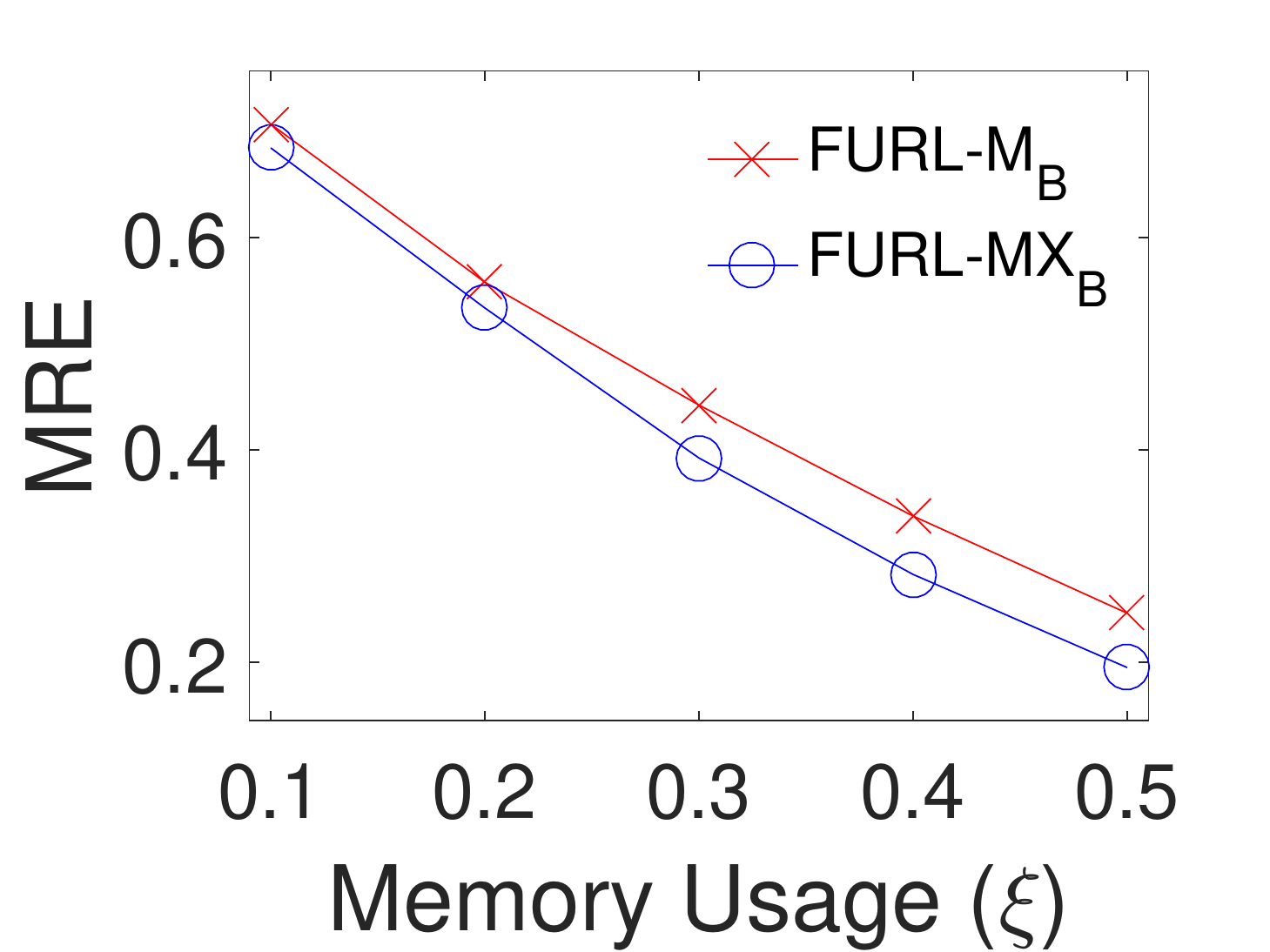}\label{fig:mul_err_facebook_b}} ~~~~~
		\subfloat[\actor]
		 {\includegraphics[width=0.235\textwidth,natwidth=610,natheight=642]{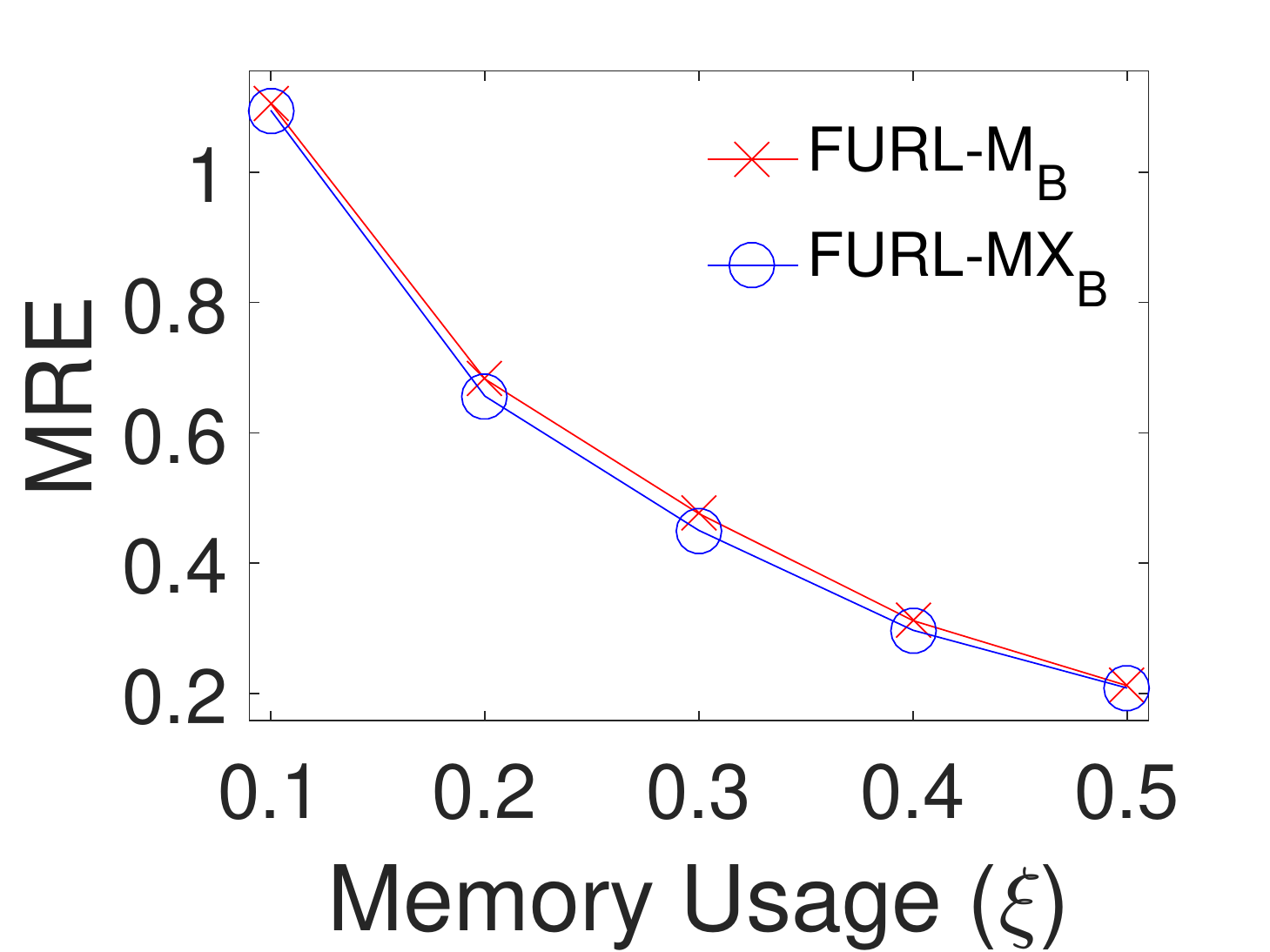}\label{fig:mul_err_actor_b}} ~~~~~
		\subfloat[\baidu]
		 {\includegraphics[width=0.235\textwidth,natwidth=610,natheight=642]{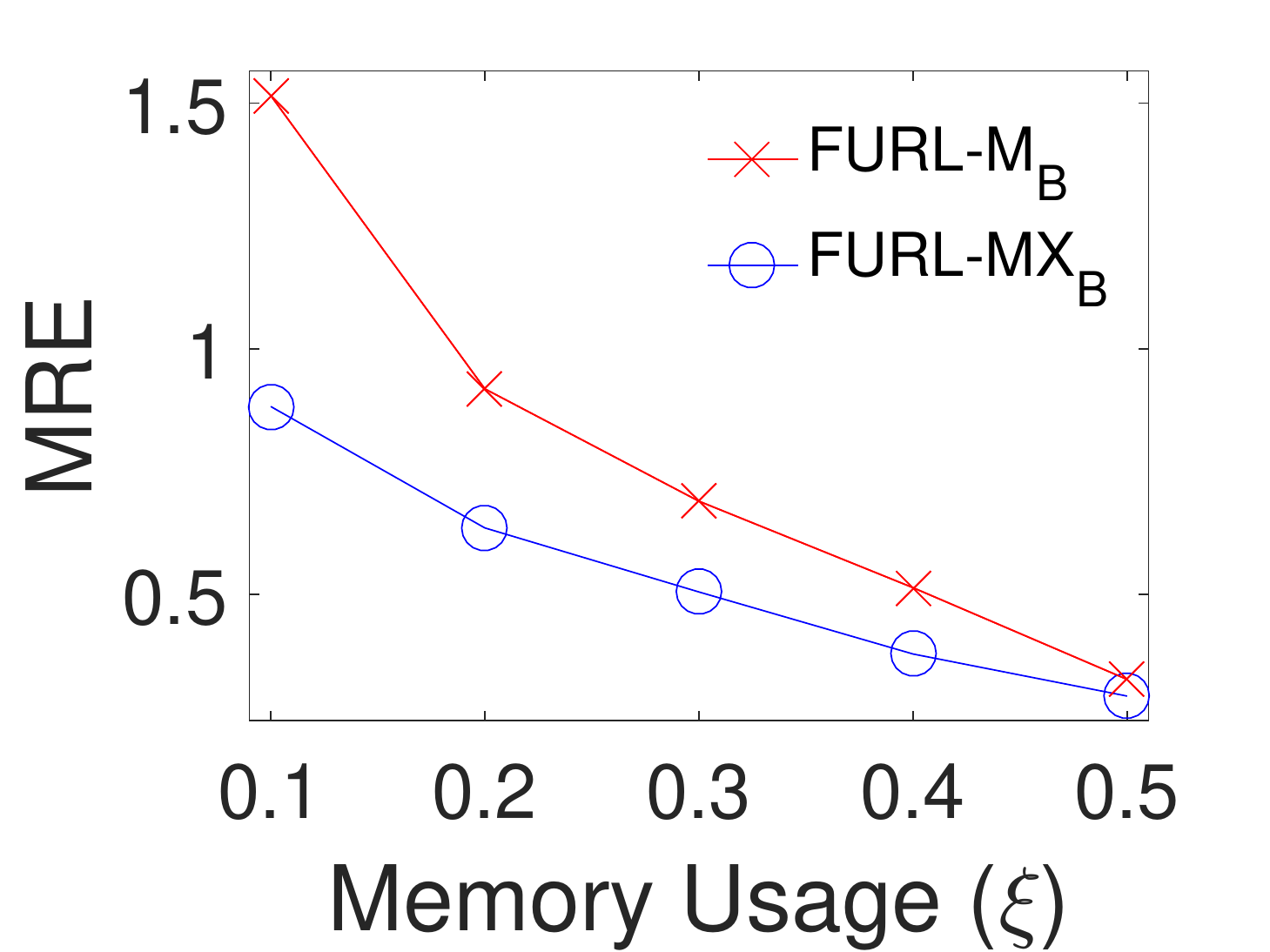}\label{fig:mul_err_baidu_b}} ~~~~~
		\subfloat[\dblp]
		 {\includegraphics[width=0.235\textwidth,natwidth=610,natheight=642]{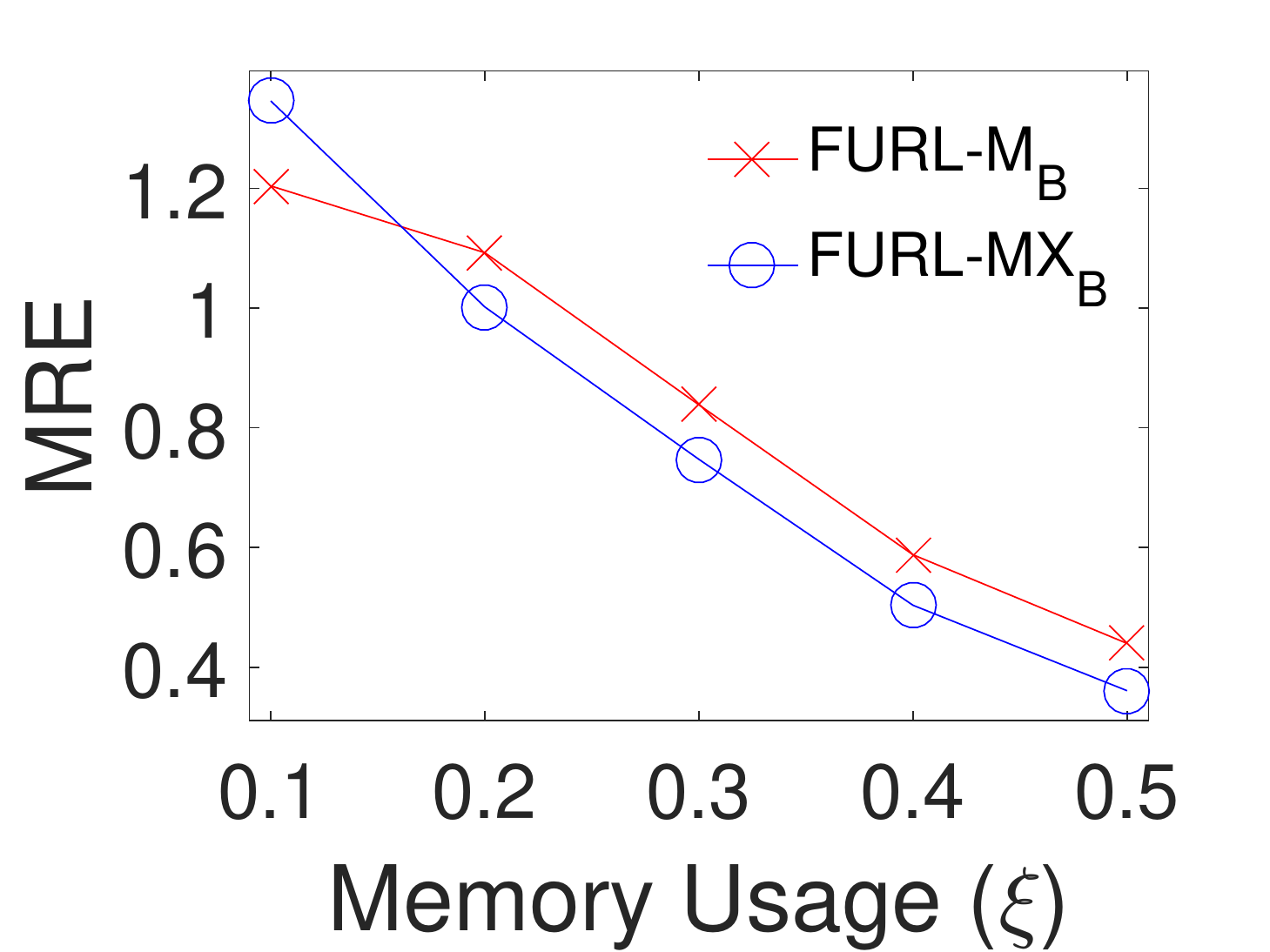}\label{fig:mul_err_dblp_b}}
		 \vspace{-2mm}
		\caption{ Mean of relative error vs.~memory usage for \methodbx and \methodb. Results with $\delta=$ 0.7, 0.4, 0.7, and 0.7 in \facebook, \actor, \baidu, and \dblp,  respectively.
\methodbx gives the minimum error for a given memory usage in almost all cases. %outperforms \methodb except for \dblp when $\xi=0.1$.
		}
		\label{fig:mul_b}
	\end{center}
	\vspace{-5mm}
\end{figure*}

\begin{figure*} [t!]
	\begin{center}
		\subfloat[\facebook]
		 {\includegraphics[width=0.235\textwidth,natwidth=610,natheight=642]{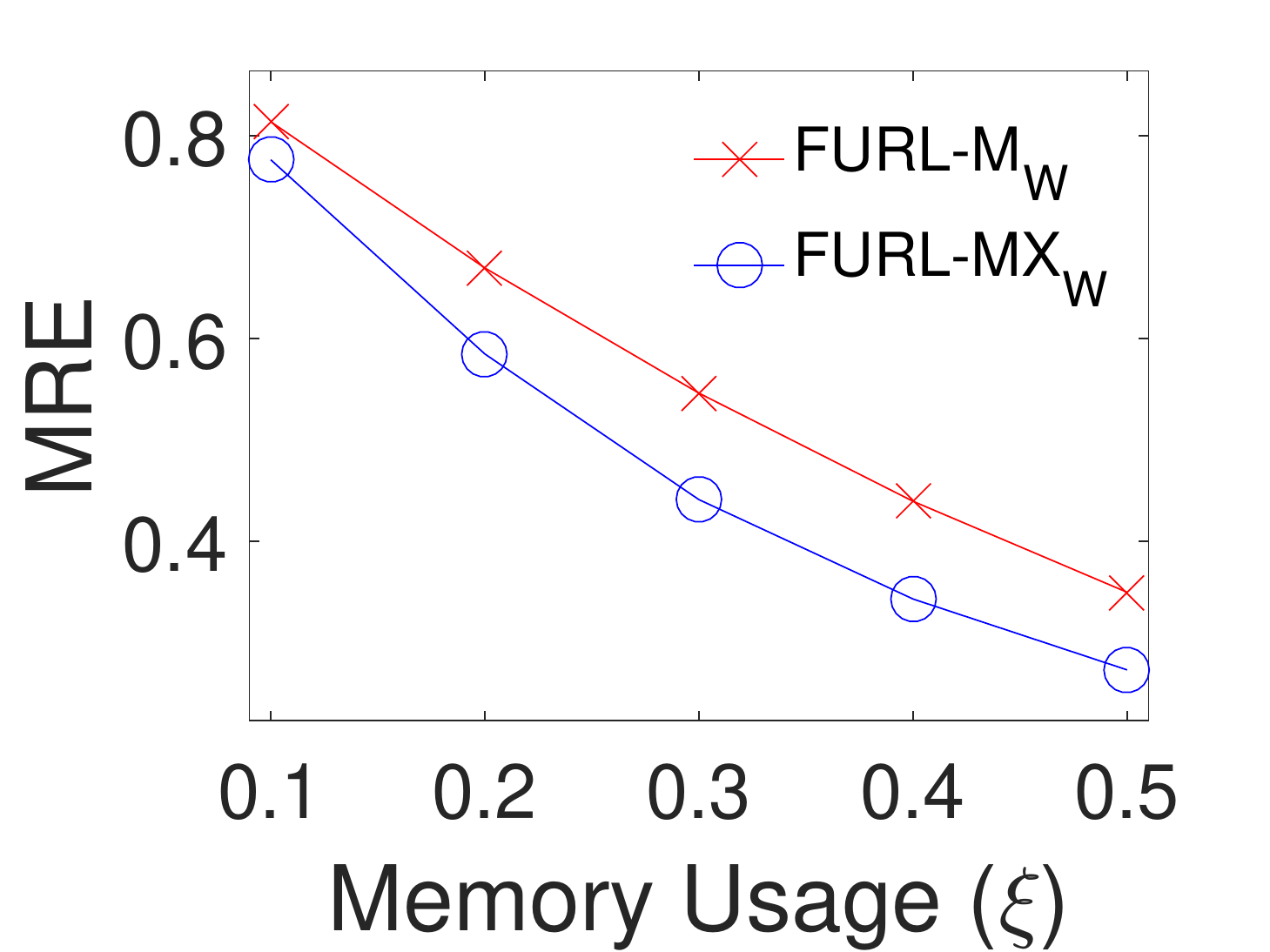}\label{fig:mul_err_facebook_w}} ~~~~~
		\subfloat[\actor]
		 {\includegraphics[width=0.235\textwidth,natwidth=610,natheight=642]{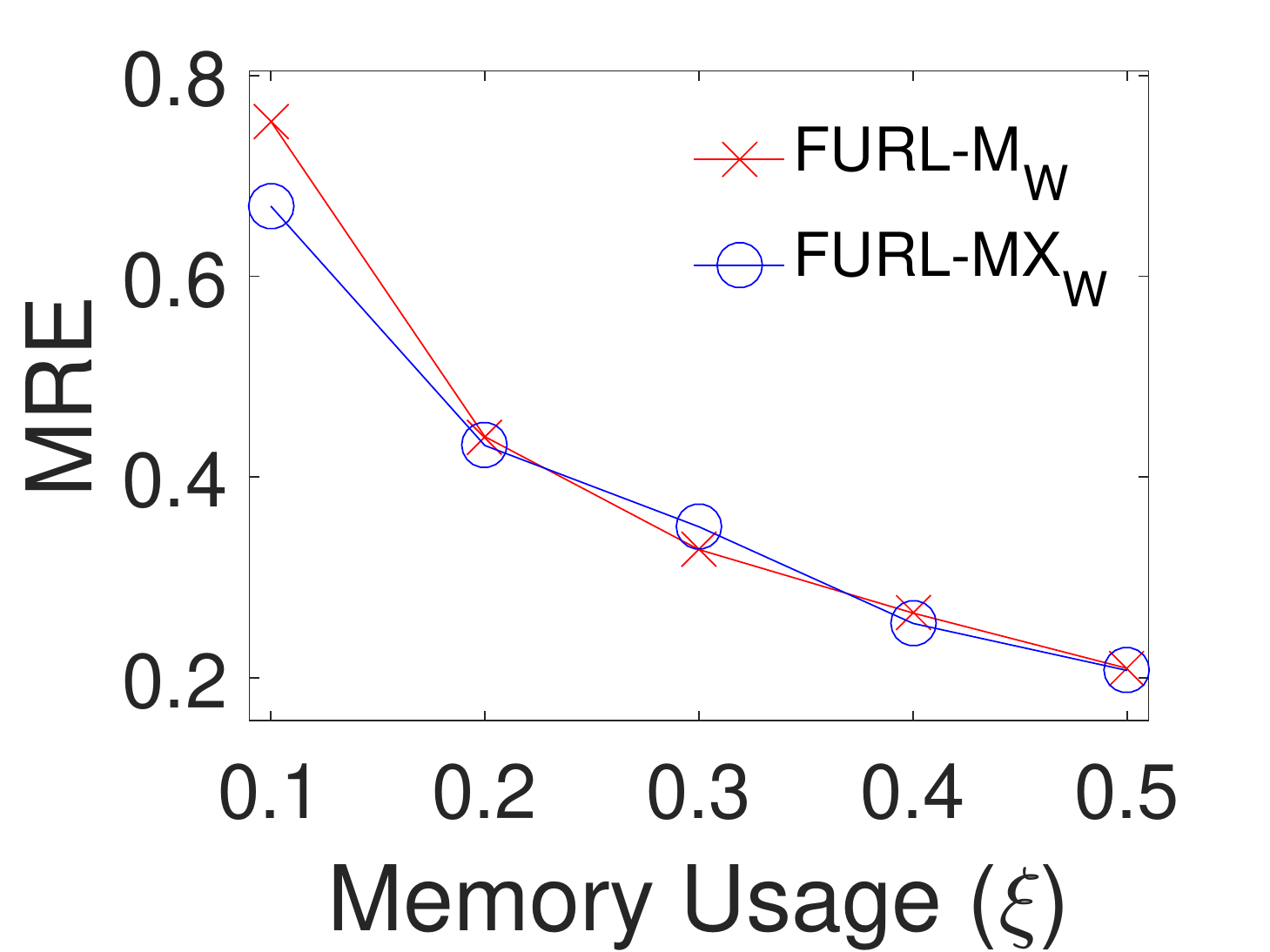}\label{fig:mul_err_actor_w}} ~~~~~
		\subfloat[\baidu]
		 {\includegraphics[width=0.235\textwidth,natwidth=610,natheight=642]{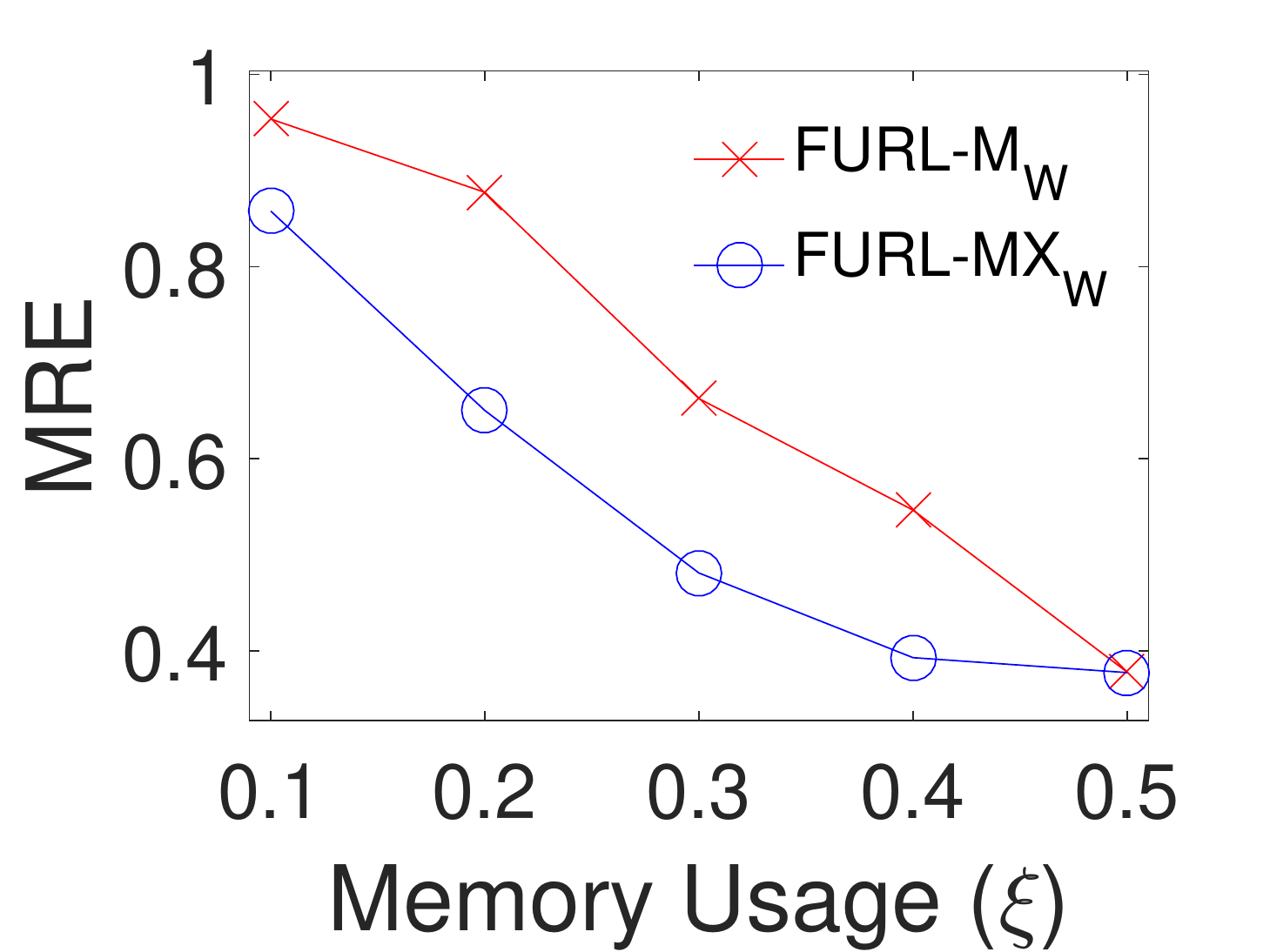}\label{fig:mul_err_baidu_w}} ~~~~~
		\subfloat[\dblp]
		 {\includegraphics[width=0.235\textwidth,natwidth=610,natheight=642]{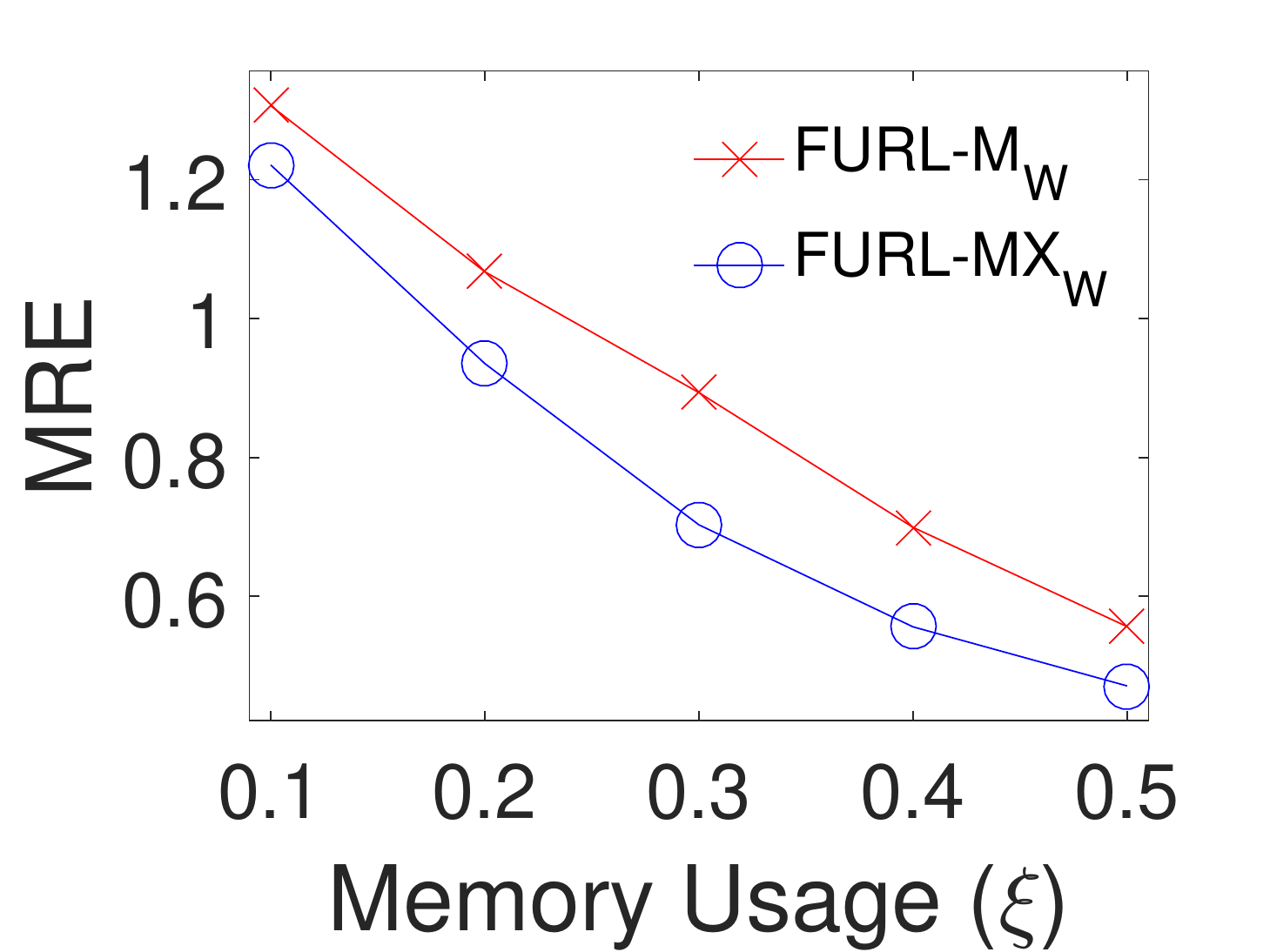}\label{fig:mul_err_dblp_w}}
		 \vspace{-2mm}
		\caption{Mean of relative error vs.~memory usage for \methodwx and \methodw. Results with $\delta=$ 0.7, 0.1, 0.7, and 0.7 in \facebook, \actor, \baidu, and \dblp, respectively.
\methodwx gives the minimum error for a given memory usage in almost all cases.
%outperforms \methodw except for \actor when $\xi=0.3$.
		}
		\label{fig:mul_w}
	\end{center}
	\vspace{-5mm}
\end{figure*}

\begin{table}[t!]
	\caption{Datasets used in our experiments.}
	\vspace{-5mm}
	\label{tab:datasets}
	\begin{center}\scriptsize
		\begin{tabular}{lrrp{5cm}}
			\toprule
			\textbf{Name} & \textbf{Node} & \textbf{Edge} & \textbf{Description}  \\ \midrule
			\multicolumn{4}{c}{Simple Graphs}\\
			\midrule
			\textsf{YahooMsg}\textsuperscript{2} & 100,000 & 587,963 & User communication in Yahoo! messenger\\
			\textsf{Youtube}\textsuperscript{1} & 1,138,499 & 2,990,287 & Social network of Youtube users  \\
			\textsf{Pokec}\textsuperscript{1} & 1,632,803 & 22,301,964 & Friendship network from Pokec \\
			\textsf{Skitter}\textsuperscript{1} & 1,696,415 & 11,095,298& Autonomous system on the Internet \\
			\textsf{Hudong}\textsuperscript{1} & 2,452,715 & 18,690,759 & ``related to" links in  encyclopedia Hudong\\
			\textsf{WebGraph}\textsuperscript{3} & 42,889,765 & 582,567,291 & Hyperlinks in Web\\
			\midrule
			\multicolumn{4}{c}{Multigraphs}\\
			\midrule
			\textsf{Facebook}\textsuperscript{1} & 46,952 & 855,542 & Wall posts on other user's wall on Facebook\\
			\textsf{Actor}\textsuperscript{1} & 382,219 & 33,115,812 & Actor collaboration in movies\\
			\textsf{Baidu}\textsuperscript{1} & 415,641 & 3,284,387 &  Baidu Encyclopedia "related to" links\\
			\textsf{DBLP-M}\textsuperscript{1} & 1,314,050 & 18,986,618 & Co-author network in DBLP\\
			%		\textsf{Twitter}\textsuperscript{4} & 41,652,230 & 1,202,513,046 & Following graph in Twitter \\
			\bottomrule
		\end{tabular}
	\end{center}\vspace{-2.5mm}
	{\tiny\textsuperscript{1}\url{http://konect.uni-koblenz.de/networks/}} \\
	 {\tiny\textsuperscript{2}\url{http://webscope.sandbox.yahoo.com/catalog.php?datatype=g}} \\
	{\tiny\textsuperscript{3}\url{http://webdatacommons.org/hyperlinkgraph/#toc2}} \\
	\vspace{-0.9cm}
\end{table}

\noindent \textbf{Evaluation Metric.}
The performances of algorithms are evaluated in the following two aspects.
\begin{itemize*}
%	\begin{comment}
%	\item \textbf{Pearson Correlation Coefficient (\pcc)}: This measures the degree of a linear relationship between an estimation and the ground truth.
%	
%	$$
%	\rho(\tau, \varDelta) = \frac{\sigma_{\tau \varDelta}}{\sigma_{\tau}\sigma_{\varDelta}},
%	$$
%	%where $\sigma_x$ is the standard deviation of a vector $x$, and $\sigma_{xy}$ is the covariance of two vectors $x$ and $y$.
%	
%	\end{comment}
	
	\item \textbf{Mean of Relative Error (\mre)}: This measures how an estimation is close to the ground truth.
	$$
	\varepsilon(\tau,\varDelta) = \frac{1}{\left|V\right|}\sum_{u\in V} \frac{\left|\tau_u-\varDelta_u\right|}{{\varDelta_u+1}},
	$$
	where $V$ is a set of nodes appearing in a graph stream. We add $1$ to the denominator for the case of $\varDelta_u=0$.
	\item \textbf{\boldmath Proportion $\xi$ of Sampled Edges}: This is the dominant factor of required memory spaces.
\end{itemize*}

%For every algorithm used,
For all the algorithms, \mre is computed by the average of results obtained by $10$ independent runnings since they are randomized algorithms.
The memory usage $\xi$ is determined as follows: for \chngwww{\methodc, \methodcx} and \triest, $\xi = M/m$ where $m$ is the number of edges in a graph;
\chngwww{for \mascot, $\xi=p$ where $p$ is a given edge sampling rate;}
for \methodm and \methodmx, $\xi = M/u$ where $u$ is the number of unique edges in a graph.

%\begin{figure*}
%	\begin{center}
%		\subfloat[\youtube]
%		{\includegraphics[width=0.21\textwidth]{FIG/youtube_links_Corr.pdf}} ~~~~~
%		\subfloat[\pokec]
%		 {\includegraphics[width=0.21\textwidth]{FIG/soc_pokec_relationships_Corr.pdf}} ~~~~~
%		\subfloat[\skitter]
%		 {\includegraphics[width=0.21\textwidth]{FIG/as_skitter_Corr.pdf}} ~~~~~
%		\subfloat[\hudong]
%		 {\includegraphics[width=0.21\textwidth]{FIG/zhishi_hudong_relatedpages_Corr.pdf}}
%		\caption{Pearson correlation coefficient vs.~memory usage for \methodcx, \mascot, and \triest.
%			%
%			We set $r=1$ and $\delta = 0.5$.
%			%
%			\methodcx and \triest are comparable, and consistently outperform \mascot.
%			%			For \methodcx, we use $J=M$ and $\delta = 0.5$.
%			%Note that \methodcx and \methodc show high PCCs for every graph, and outperform the others consistently.
%			}
%		\label{fig:sim_cor}
%	\end{center}
%	\vspace{-4mm}
%\end{figure*}

\subsection{Accuracy for Simple Graph Stream} \label{sec:exp_sim}
\chngwww{
For simple graph stream,
	we compare \methodc with two competing methods: \mascot~\cite{LimK15} and \triest~\cite{StefaniERU16}.
}
Note that \triest is the same as \methodc introduced in Section~\ref{sec:method:simple}.

\figref{fig:methodcx result} shows comparison between \methodcx, \triest and \mascot in \mre over the memory usage $\xi$.
%We use $\alpha=1$ and $\delta=0.5$ for \methodcx.
We choose $\delta=0.7, 0.1, 0.4$, and $0.4$ in \youtube, \pokec, \skitter, and \hudong, respectively.
%Note that \methodcx outperforms \triest.
As shown in the figure, \methodcx gives the minimum error for a given memory usage compared to \triest and \mascot since \methodcx gives more concentrated estimation around the ground truth.

%\methodcx outperforms \mascot as well as \triest, except only for \triest with $\xi=0.05$ on \youtube.
%\mascot is more inaccurate even than \methodc.
%In contrast to \figref{fig:comp_cor}, the accuracy of \methodcx and \methodc is clearly distinguished. Except for \youtube with $\xi=0.05$, \methodcx always outperforms \methodc.
%Although \mascot is more accurate than \methodc, it is degraded than \methodcx and \methodc.

%\figref{fig:sim_cor_pokec} and~\ref{fig:sim_cor_skitter} show comparison of \methodcx with the 2 competing methods in \pcc over the memory usage $\xi$.
%%We use $\alpha=1$ and $\delta=0.5$ for \methodcx.
%As shown in the figure, \methodcx and \triest are comparable, and consistently outperform \mascot.
%%Note that the performance gap between \methodcx and \methodc is insignificant.

\subsection{Accuracy for Multigraph Stream} \label{sec:exp_mul}
For multigraph stream,
we compare only our proposed methods \methodmx and \methodm since there is no previous work.

\figref{fig:mul_b} shows the accuracy of \methodbx and \methodb for binary counting
in \mre over the memory usage $\xi$. We present results with $\delta= 0.7, 0.4, 0.7$, and $0.7$ for \facebook, \actor, \baidu, and \dblp, respectively.
Note that \methodbx gives the minimum error for a given memory usage in almost all cases since \methodbx gives concentrated results.
%both algorithms give small errors, and
%\methodbx outperforms \methodb except for \dblp when $\xi=0.1$
%since \methodbx gives concentrated results.
%~\ref{fig:mul_err_facebook_b},~\ref{fig:mul_err_actor_b}, ~\ref{fig:mul_err_baidu_b}, ~\ref{fig:mul_err_dblp_b} respectively.

\figref{fig:mul_w} shows comparison of \methodwx and \methodw for weighted counting
in \mre over the memory usage $\xi$. We present results with $\delta=  0.7, 0.1, 0.7$, and $0.7$ for \facebook, \actor, \baidu, and \dblp, respectively. %~\ref{fig:mul_err_facebook_w},~\ref{fig:mul_err_actor_w}, ~\ref{fig:mul_err_baidu_w}, ~\ref{fig:mul_err_dblp_w} respectively.
As in the case of binary counting,
%both give small errors, and
\methodwx outperforms \methodw,
giving the minimum error for a given memory usage in almost all cases
since \methodwx gives concentrated results.

\subsection{Anomaly Detection} \label{sec:exp_anom}
%In this section, we investigate anomalous patterns discovered by \method in real world graphs \yahoo and \webgraph listed in \tabref{tab:datasets}.
%We use \methodcx since they are simple graphs.

In this section, we show anomalous patterns and nodes on real world graphs using \method. The number of local triangles is an important index that indicates characteristics of nodes and cohesiveness of groups which neighbors of nodes form.
We use two datasets \yahoo and \webgraph listed in \tabref{tab:datasets}.
\yahoo is a user communication network in Yahoo! messenger whose edge means a user sends a message \chng{to} another user. \webgraph is a hyperlink network of webpages. Its node and edge correspond to a web page and a hyperlink, respectively. 
Since they are simple graphs, we use \methodcx with the following parameters: $\delta=0.4$, $\xi=0.01$ and $J=M$.
%
%We set $\delta=0.4$, $\xi=0.01$ and $J=M$.

\begin{figure}[t!]
	\begin{center}
		\subfloat[\yahoo]
		 {\includegraphics[width=0.5\columnwidth,natwidth=610,natheight=642]{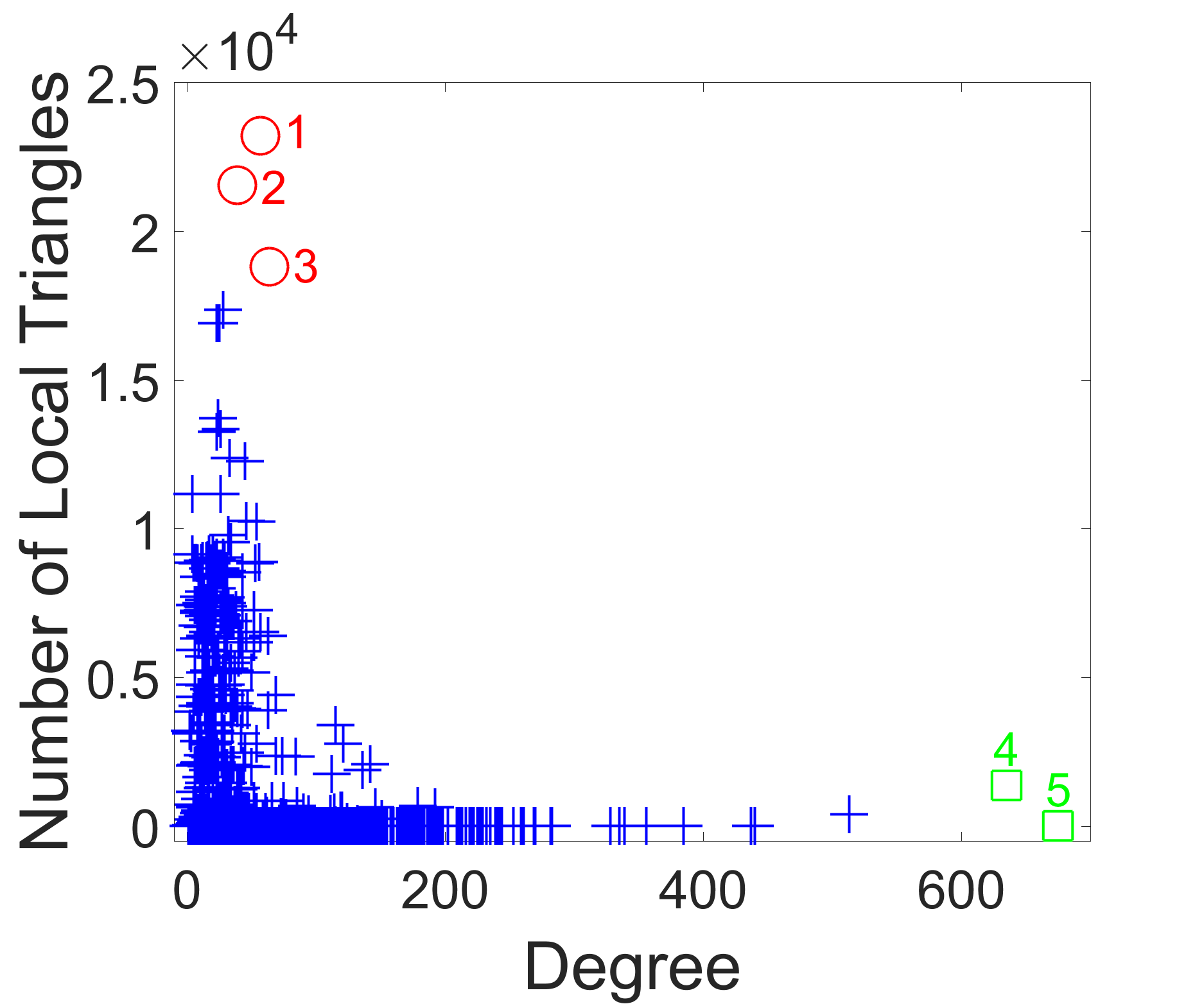}\label{fig:yahoomsg_a}} ~~~~~
		\subfloat[Spyplots for egonets]{
			 \includegraphics[width=0.43\columnwidth,natwidth=610,natheight=642]{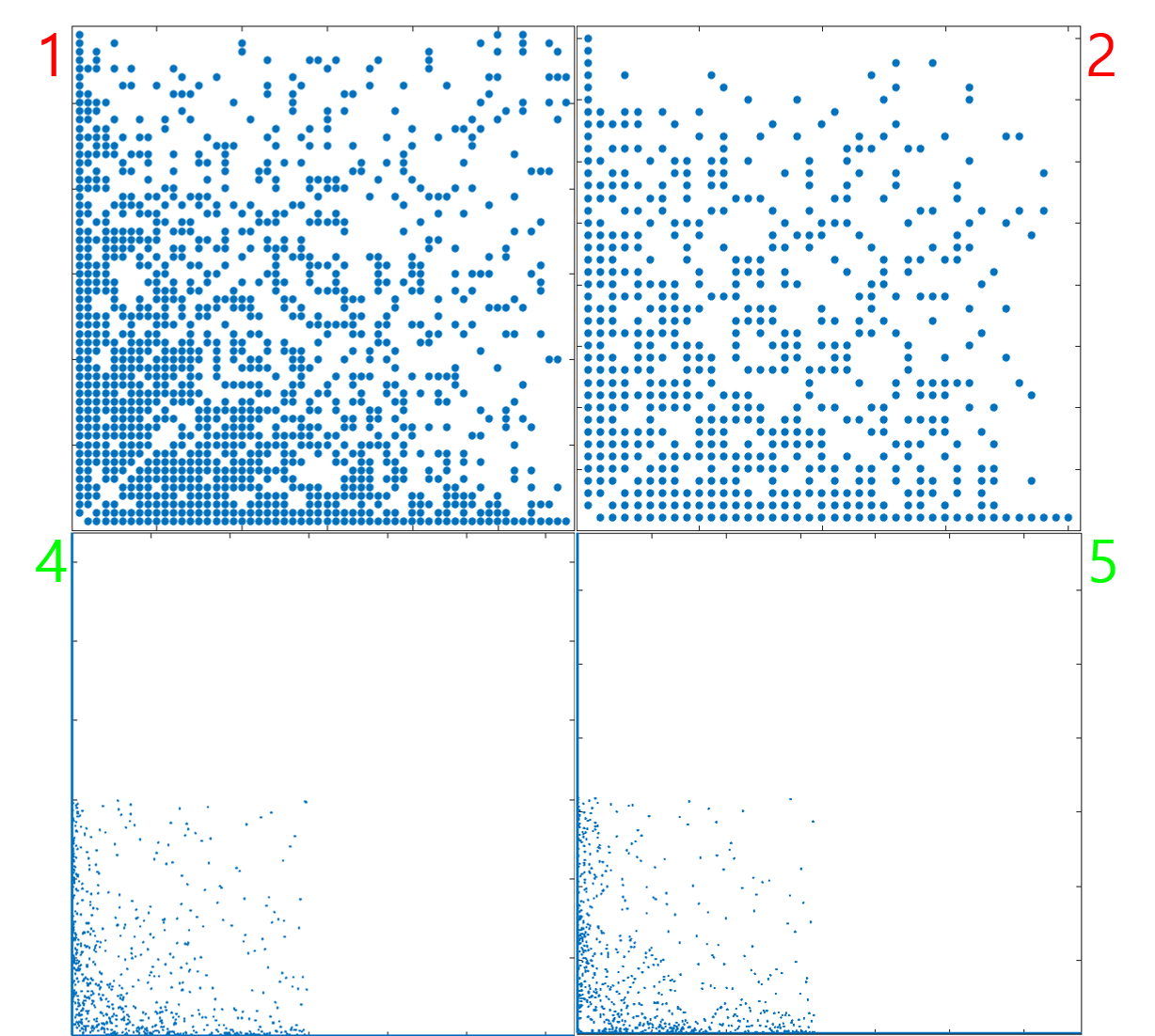}
			\label{fig:yahoomsg_b}}
		 \vspace{-2mm}
		\caption{
			(a) Degree vs.~the number of local triangles estimated by \method in \yahoo.
			The neighbors of each red node contain a very large giant connected component; those of each green node have many tiny connected components.
			(b) Spyplots for the egonetworks of two red nodes (top) and two green nodes (bottom). Each spyplot at the top shows that the egonetworks are tightly connected. Each plot at the bottom represents that the egonetworks are loosely connected.
		}
		\label{fig:yahoomsg}
	\end{center}
	\vspace{-5mm}
\end{figure}

\begin{figure}[t!]
	\begin{center}
		\subfloat[\webgraph]
		 {\includegraphics[width=0.5\columnwidth,natwidth=610,natheight=642]{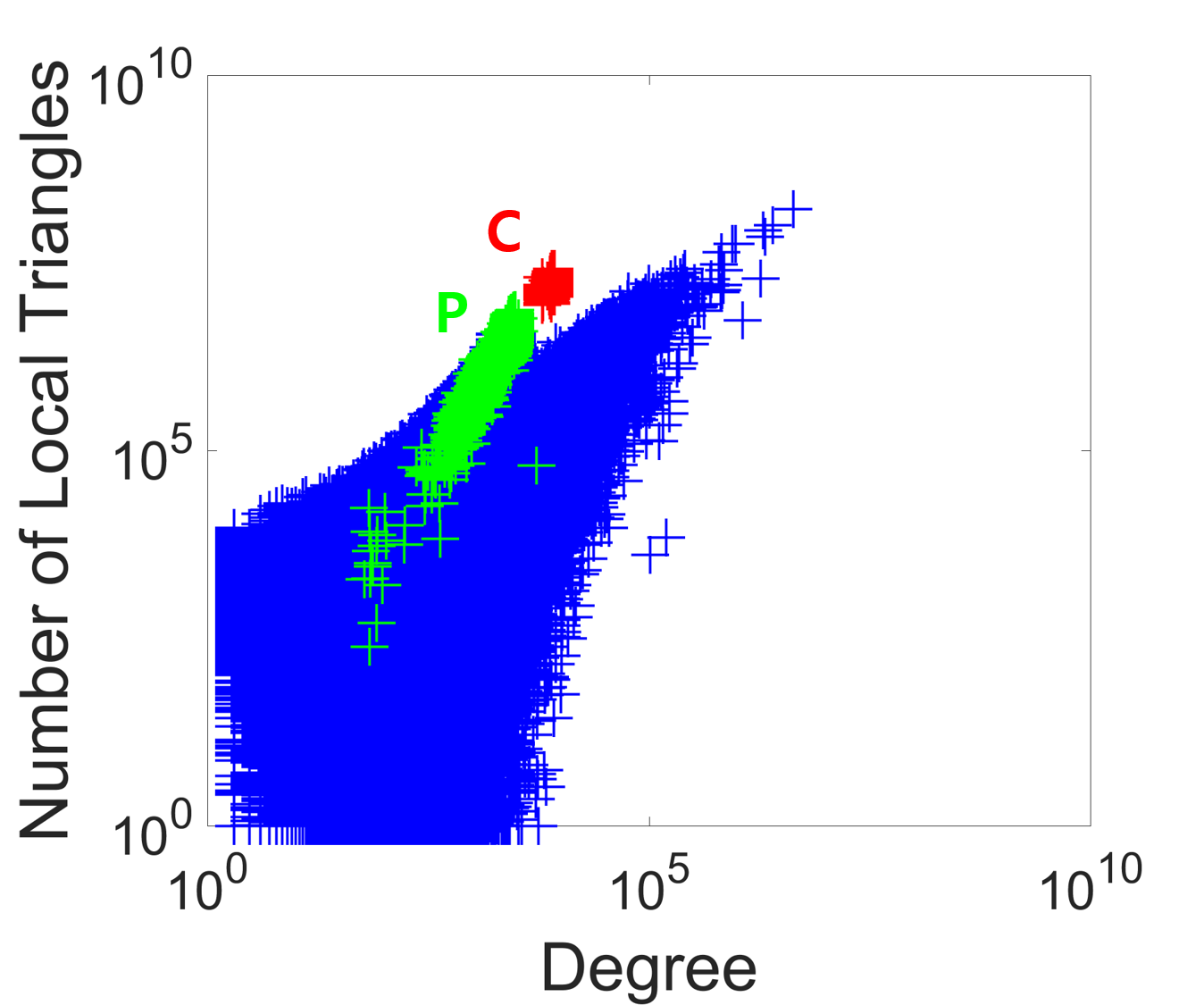}\label{fig:webgraph_a}} ~~~~~
		\subfloat[Spyplot for C, P]
		 {\includegraphics[width=0.43\columnwidth,natwidth=610,natheight=642]{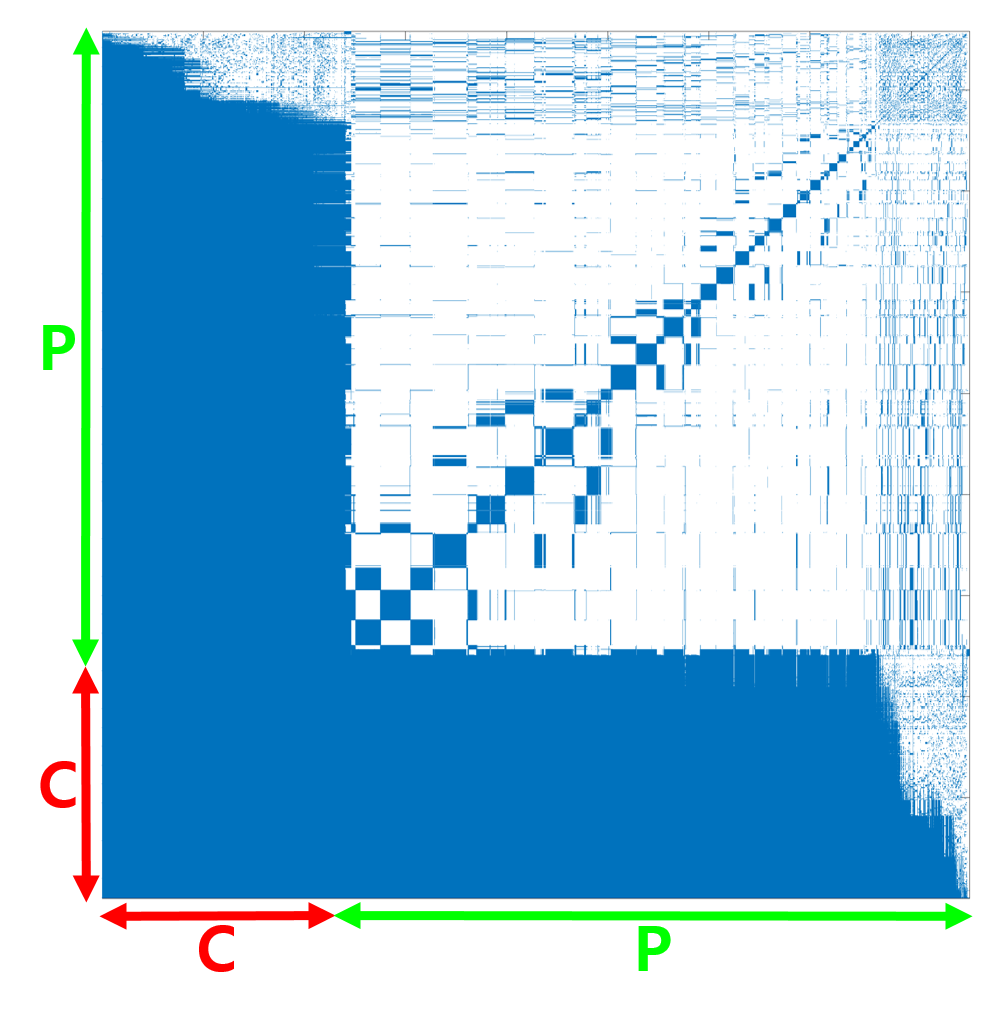}\label{fig:webgraph_b}}
		  \vspace{-2mm}
		\caption{
			(a) Degree vs.~the number of local triangles estimated by \method in \webgraph.
			The red group (marked `C') forms a near-clique.
			The green group (marked `P') corresponds to neighbors of the red nodes. Nodes in the green group have tight connections to the red group but are internally sparse.
			Both groups form a core-periphery structure.
			(b) Spyplot for the subgraph containing all red and green nodes. The core-periphery structure is clearly displayed.
		}
		\label{fig:webgraph}
	\end{center}
	\vspace{-5mm}
\end{figure}

\mchng{
\begin{observation}[Single Large Neighbor Group]
	In \yahoo, there are users whose neighbors form a large connected group. %, excluding the user.
\end{observation}
}

\mchng{
\begin{observation}[Diverse Neighbor Groups]
	In \yahoo, there are users whose neighbors form diverse groups. %, excluding the user. 
The egonetwork of such user is a near-star network.
\end{observation}
}

\mchng{
\figref{fig:yahoomsg_a} shows two anomalous patterns in \yahoo. The first pattern, marked in red, corresponds to users whose neighbors form a large connected component. The second pattern, marked in green, corresponds to users whose neighbors form diverse groups. There are three users in the "red" patterns and two users in the "green" patterns (marked in \figref{fig:yahoomsg_a}). 
In the red patterns, neighbors of each of user1 and user2 make only one connected component, meaning that the users communicate within themselves. \chng{Two} spyplots at the top of \figref{fig:yahoomsg_b} are those for the egonetworks of user1 and user2, respectively.
%and his neighbors and shows the neighbors of user1 make a single large connected component, even except user1.
The neighbors of user3 make two connected components, but about $94.6\%$ of the neighbors belong to one component. %\figref{fig:yahoomsgapd_c} is a spyplot for the egonetwork of user3.
% and shows the neighbors of user3 make a dense connected component, excluding user3.
As a result, all users in the red pattern have tightly connected egonetworks.
}

\mchng{
In the green pattern, neighbors of the users make a number of groups. \chng{Two} spyplots at the bottom of \figref{fig:yahoomsg_b} are spyplots for the egonetworks of user4 and user5, respectively. They demonstrate that the neighbors of each of user4 and user5 are hardly connected to each other.
As a result, user4 has a near-star local structure.
%Also, it shows the structure of the user and the neighbors is a near-star network.
The neighbors of user4 form 357 connected components and their sizes are 1 or 2, except for 6 out of 357. The average size of the connected components is about 1.78. user5 shows a similar pattern. Its neighbors make 378 connected components and the average size of them is about 1.79.
}

\mchng{
\begin{observation}[Core-periphery]
	In \webgraph, there is a core-periphery structure which consists of two groups of nodes. One is a near-clique and the other is a sparse graph. The two subgraphs are tightly connected.
\end{observation}
}

\mchng{
\figref{fig:webgraph_a} shows two groups ``C" and ``P" which are colored in red and green, respectively.
%which are subgraphs of an anomalous structure, core-periphery.
Especially, the group C appears to be clearly separated from the rest of the graph in the figure.
%In this figure, we look into the group, painted in red, separated from others and call this group as "C".
We find that the nodes in C are connected tightly. The number of the nodes in C is 2463 and their adjacency matrix has 2913 zeros. It means that there are only 225 missing edges: note that \webgraph has no self-loop.
%all of the nodes have a link between themselves, except 225 edges (no self-loops in \webgraph).
%We observe neighbors of the nodes in C.
We also examine $6116$ neighbors of the nodes in C.
Most of them are tightly connected with the nodes in C and we denote this neighbor group by ``P".
%6116 of them are tightly connected with the nodes in C and we call these neighbors group as "P".
%Also, the nodes in P have another characteristic. They are loosely connected with each other.
\chng{The nodes in P are loosely connected with each other. }
To sum up, the nodes in C and P respectively form a near-clique and a sparse graph, and the nodes in P are tightly connected with the nodes in C, forming a core-periphery structure~\cite{KangMF14}.
\figref{fig:webgraph_b} is a spyplot of the nodes in C and P and shows they form a core-periphery.
%Therefore, the nodes in C and P form a core-periphery structure.
%Furthermore, we discover that most of the nodes in two subgraphs are German websites. These websites have links to a specific website which provides rental services and their looks are cast in the same mold.
}

\section{Conclusion}
\label{sec:conclusion}
In this paper, we propose \method, an accurate algorithm for local triangle estimation in simple and multigraph streams.
\method handles graph streams of any sizes since it guarantees that the memory usage is fixed, and achieves high accuracy by reducing the variance of estimation.
%, and carefully handling multigraph streams.
%In contrast to previous methods, \method provides more accurate estimation and consider multigraph streams.
%The main idea is to 1) combine past estimations with the current one for reducing a variance and 2) use reservoir sampling with random hash for handling duplicate edges.
%The first idea especially has a benefit in most real world stream mining systems where we are allowed to scan a graph stream only once.
Our experimental results demonstrate that \methodcx, the simple graph stream version of \method, provides the best accuracy compared to the state-of-the-art algorithm.
Furthermore, we show that \methodmx, the multigraph stream version of \method, performs well for graph streams with duplicate edges. %and basic algorithm respectively.
Using \method, we discover abnormal egonetworks like a near clique/star in a user communication network, and a core-periphery structure in the Web.

\bibliography{BIB/yongsub, BIB/ukang}

\begin{thebibliography}{10}

\bibitem{AhmedDNK14}
Nesreen~K. Ahmed, Nick~G. Duffield, Jennifer Neville, and Ramana~Rao Kompella.
\newblock Graph sample and hold: a framework for big-graph analytics.
\newblock In {\em The 20th {ACM} {SIGKDD} International Conference on Knowledge
  Discovery and Data Mining, {KDD} '14, New York, NY, {USA} - August 24 - 27,
  2014}, pages 1446--1455, 2014.

\bibitem{AlonYZ97}
Noga Alon, Raphael Yuster, and Uri Zwick.
\newblock Finding and counting given length cycles.
\newblock {\em Algorithmica}, 17(3):209--223, 1997.

\bibitem{ArifuzzamanKM13}
Shaikh Arifuzzaman, Maleq Khan, and Madhav~V. Marathe.
\newblock {PATRIC:} a parallel algorithm for counting triangles in massive
  networks.
\newblock In {\em 22nd {ACM} International Conference on Information and
  Knowledge Management, CIKM'13, San Francisco, CA, USA, October 27 - November
  1, 2013}, pages 529--538, 2013.

\bibitem{Bar-YossefKS02}
Ziv Bar{-}Yossef, Ravi Kumar, and D.~Sivakumar.
\newblock Reductions in streaming algorithms, with an application to counting
  triangles in graphs.
\newblock In {\em Proceedings of the Thirteenth Annual {ACM-SIAM} Symposium on
  Discrete Algorithms, January 6-8, 2002, San Francisco, CA, {USA.}}, pages
  623--632, 2002.

\bibitem{BecchettiBCG10}
Luca Becchetti, Paolo Boldi, Carlos Castillo, and Aristides Gionis.
\newblock Efficient algorithms for large-scale local triangle counting.
\newblock {\em {TKDD}}, 4(3), 2010.

\bibitem{BerryHLP11}
Jonathan~W Berry, Bruce Hendrickson, Randall~A LaViolette, and Cynthia~A
  Phillips.
\newblock Tolerating the community detection resolution limit with edge
  weighting.
\newblock {\em Physical Review E}, 83(5):056119, 2011.

\bibitem{BroderCFM98}
Andrei~Z. Broder, Moses Charikar, Alan~M. Frieze, and Michael Mitzenmacher.
\newblock Min-wise independent permutations (extended abstract).
\newblock In {\em Proceedings of the Thirtieth Annual {ACM} Symposium on the
  Theory of Computing, Dallas, Texas, USA, May 23-26, 1998}, pages 327--336,
  1998.

\bibitem{BuriolFLMS06}
Luciana~S. Buriol, Gereon Frahling, Stefano Leonardi, Alberto
  Marchetti{-}Spaccamela, and Christian Sohler.
\newblock Counting triangles in data streams.
\newblock In {\em Proceedings of the Twenty-Fifth {ACM} {SIGACT-SIGMOD-SIGART}
  Symposium on Principles of Database Systems, June 26-28, 2006, Chicago,
  Illinois, {USA}}, pages 253--262, 2006.

\bibitem{chou2010discovering}
Bin{-}Hui Chou and Einoshin Suzuki.
\newblock Discovering community-oriented roles of nodes in a social network.
\newblock In {\em Data Warehousing and Knowledge Discovery, 12th International
  Conference, {DAWAK} 2010, Bilbao, Spain, August/September 2010. Proceedings},
  pages 52--64, 2010.

\bibitem{ChuC11}
Shumo Chu and James Cheng.
\newblock Triangle listing in massive networks and its applications.
\newblock In {\em Proceedings of the 17th {ACM} {SIGKDD} International
  Conference on Knowledge Discovery and Data Mining, San Diego, CA, USA, August
  21-24, 2011}, pages 672--680, 2011.

\bibitem{Cohen09}
Jonathan Cohen.
\newblock Graph twiddling in a mapreduce world.
\newblock {\em Computing in Science and Engineering}, 11(4):29--41, 2009.

\bibitem{EckmannM02}
Jean-Pierre Eckmann and Elisha Moses.
\newblock Curvature of co-links uncovers hidden thematic layers in the world
  wide web.
\newblock {\em Proceedings of the national academy of sciences},
  99(9):5825--5829, 2002.

\bibitem{epasto2015ego}
Alessandro Epasto, Silvio Lattanzi, Vahab~S. Mirrokni, Ismail Sebe, Ahmed Taei,
  and Sunita Verma.
\newblock Ego-net community mining applied to friend suggestion.
\newblock {\em {PVLDB}}, 9(4):324--335, 2015.

\bibitem{HuTC13}
Xiaocheng Hu, Yufei Tao, and Chin{-}Wan Chung.
\newblock Massive graph triangulation.
\newblock In {\em Proceedings of the {ACM} {SIGMOD} International Conference on
  Management of Data, {SIGMOD} 2013, New York, NY, USA, June 22-27, 2013},
  pages 325--336, 2013.

\bibitem{JhaSP13}
Madhav Jha, C.~Seshadhri, and Ali Pinar.
\newblock A space efficient streaming algorithm for triangle counting using the
  birthday paradox.
\newblock In {\em The 19th {ACM} {SIGKDD} International Conference on Knowledge
  Discovery and Data Mining, {KDD} 2013, Chicago, IL, USA, August 11-14, 2013},
  pages 589--597, 2013.

\bibitem{JowhariG05}
Hossein Jowhari and Mohammad Ghodsi.
\newblock New streaming algorithms for counting triangles in graphs.
\newblock In {\em Computing and Combinatorics, 11th Annual International
  Conference, {COCOON} 2005, Kunming, China, August 16-29, 2005, Proceedings},
  pages 710--716, 2005.

\bibitem{KaneMSS12}
Daniel~M. Kane, Kurt Mehlhorn, Thomas Sauerwald, and He~Sun.
\newblock Counting arbitrary subgraphs in data streams.
\newblock In {\em Automata, Languages, and Programming - 39th International
  Colloquium, {ICALP} 2012, Warwick, UK, July 9-13, 2012, Proceedings, Part
  {II}}, pages 598--609, 2012.

\bibitem{KangMF14}
U.~Kang, Brendan Meeder, Evangelos~E. Papalexakis, and Christos Faloutsos.
\newblock Heigen: Spectral analysis for billion-scale graphs.
\newblock {\em {IEEE} Trans. Knowl. Data Eng.}, 26(2):350--362, 2014.

\bibitem{KimHLPY14}
Jinha Kim, Wook{-}Shin Han, Sangyeon Lee, Kyungyeol Park, and Hwanjo Yu.
\newblock {OPT:} a new framework for overlapped and parallel triangulation in
  large-scale graphs.
\newblock In {\em International Conference on Management of Data, {SIGMOD}
  2014, Snowbird, UT, USA, June 22-27, 2014}, pages 637--648, 2014.

\bibitem{KutzkovP13}
Konstantin Kutzkov and Rasmus Pagh.
\newblock On the streaming complexity of computing local clustering
  coefficients.
\newblock In {\em Sixth {ACM} International Conference on Web Search and Data
  Mining, {WSDM} 2013, Rome, Italy, February 4-8, 2013}, pages 677--686, 2013.

\bibitem{Latapy08}
Matthieu Latapy.
\newblock Main-memory triangle computations for very large (sparse (power-law))
  graphs.
\newblock {\em Theor. Comput. Sci.}, 407(1-3):458--473, 2008.

\bibitem{LimK15}
Yongsub Lim and U.~Kang.
\newblock {MASCOT:} memory-efficient and accurate sampling for counting local
  triangles in graph streams.
\newblock In {\em Proceedings of the 21th {ACM} {SIGKDD} International
  Conference on Knowledge Discovery and Data Mining, Sydney, NSW, Australia,
  August 10-13, 2015}, pages 685--694, 2015.

\bibitem{MilEtAl02}
Ron Milo, Shai Shen-Orr, Shalev Itzkovitz, Nadav Kashtan, Dmitri Chklovskii,
  and Uri Alon.
\newblock Network motifs: simple building blocks of complex networks.
\newblock {\em Science}, 298(5594):824--827, 2002.

\bibitem{PaghS14}
Rasmus Pagh and Francesco Silvestri.
\newblock The input/output complexity of triangle enumeration.
\newblock In {\em Proceedings of the 33rd {ACM} {SIGMOD-SIGACT-SIGART}
  Symposium on Principles of Database Systems, PODS'14, Snowbird, UT, USA, June
  22-27, 2014}, pages 224--233, 2014.

\bibitem{PaghT12}
Rasmus Pagh and Charalampos~E. Tsourakakis.
\newblock Colorful triangle counting and a mapreduce implementation.
\newblock {\em Inf. Process. Lett.}, 112(7):277--281, 2012.

\bibitem{ParkC13}
Ha{-}Myung Park and Chin{-}Wan Chung.
\newblock An efficient mapreduce algorithm for counting triangles in a very
  large graph.
\newblock In {\em 22nd {ACM} International Conference on Information and
  Knowledge Management, CIKM'13, San Francisco, CA, USA, October 27 - November
  1, 2013}, pages 539--548, 2013.

\bibitem{ParkSKP14}
Ha{-}Myung Park, Francesco Silvestri, U.~Kang, and Rasmus Pagh.
\newblock Mapreduce triangle enumeration with guarantees.
\newblock In {\em Proceedings of the 23rd {ACM} International Conference on
  Conference on Information and Knowledge Management, {CIKM} 2014, Shanghai,
  China, November 3-7, 2014}, pages 1739--1748, 2014.

\bibitem{PavanTTSW13}
A.~Pavan, Kanat Tangwongsan, Srikanta Tirthapura, and Kun{-}Lung Wu.
\newblock Counting and sampling triangles from a graph stream.
\newblock {\em {PVLDB}}, 6(14):1870--1881, 2013.

\bibitem{StefaniERU16}
Lorenzo~De Stefani, Alessandro Epasto, Matteo Riondato, and Eli Upfal.
\newblock Tri{\`{e}}st: Counting local and global triangles in fully-dynamic
  streams with fixed memory size.
\newblock In {\em Proceedings of the 22nd {ACM} {SIGKDD} International
  Conference on Knowledge Discovery and Data Mining, San Francisco, CA, USA,
  August 13-17, 2016}, pages 825--834, 2016.

\bibitem{sunter1977list}
AB~Sunter.
\newblock List sequential sampling with equal or unequal probabilities without
  replacement.
\newblock {\em Applied Statistics}, pages 261--268, 1977.

\bibitem{suri2011counting}
Siddharth Suri and Sergei Vassilvitskii.
\newblock Counting triangles and the curse of the last reducer.
\newblock In {\em Proceedings of the 20th International Conference on World
  Wide Web, {WWW} 2011, Hyderabad, India, March 28 - April 1, 2011}, pages
  607--614, 2011.

\bibitem{SuriV11}
Siddharth Suri and Sergei Vassilvitskii.
\newblock Counting triangles and the curse of the last reducer.
\newblock In {\em Proceedings of the 20th International Conference on World
  Wide Web, {WWW} 2011, Hyderabad, India, March 28 - April 1, 2011}, pages
  607--614, 2011.

\bibitem{TsourakakisKMF09}
Charalampos~E. Tsourakakis, U.~Kang, Gary~L. Miller, and Christos Faloutsos.
\newblock {DOULION:} counting triangles in massive graphs with a coin.
\newblock In {\em Proceedings of the 15th {ACM} {SIGKDD} International
  Conference on Knowledge Discovery and Data Mining, Paris, France, June 28 -
  July 1, 2009}, pages 837--846, 2009.

\bibitem{Vitter85}
Jeffrey~Scott Vitter.
\newblock Random sampling with a reservoir.
\newblock {\em {ACM} Trans. Math. Softw.}, 11(1):37--57, 1985.

\bibitem{WelserGFS07}
Howard~T. Welser, Eric Gleave, Danyel Fisher, and Marc~A. Smith.
\newblock Visualizing the signatures of social roles in online discussion
  groups.
\newblock {\em Journal of Social Structure}, 8, 2007.

\bibitem{YangWWGZD11}
Zhi Yang, Christo Wilson, Xiao Wang, Tingting Gao, Ben~Y. Zhao, and Yafei Dai.
\newblock Uncovering social network sybils in the wild.
\newblock In {\em Proceedings of the 11th {ACM} {SIGCOMM} Internet Measurement
  Conference, {IMC} '11, Berlin, Germany, November 2-, 2011}, pages 259--268,
  2011.

\end{thebibliography}
\bibliographystyle{plain}

\appendix
%\section{Variance Plot for \methodcrm and \methodcxrm}
%\label{sec:appendix}
%
%\begin{figure}[t!]
%	\begin{center}
%	\subfloat[]
%		 {\includegraphics[width=0.4\columnwidth]{FIG/appx/youtube_links_concent_intv[299028]_delta[0]_sz[299028]_47.pdf}} ~~
%	\subfloat[]
%		 {\includegraphics[width=0.4\columnwidth]{FIG/appx/youtube_links_concent_intv[299028]_delta[0]_sz[299028]_48.pdf}} \\
%	\subfloat[]
%		 {\includegraphics[width=0.4\columnwidth]{FIG/appx/youtube_links_concent_intv[299028]_delta[0]_sz[299028]_50.pdf}} ~~
%	\subfloat[]
%		 {\includegraphics[width=0.4\columnwidth]{FIG/appx/youtube_links_concent_intv[299028]_delta[0]_sz[299028]_56.pdf}}
%	\caption{Frequency of MRE for \methodc and \methodcx. Each line is calculated from $100$ independent runnings for the correponding algrorithm.}
%	\end{center}
%\end{figure}

\section{Additional Analysis of \methodcxrm}
\subsection{Lemma~\ref{lem:our_expect}}
Remind the following definitions:
\begin{itemize*}
	\item $W(b_\lambda) = (1-\delta)\delta^{B-b_\lambda}$.
	\item $\phi(b_\lambda) = \delta^{B-b_\lambda+1}$.
\end{itemize*}

The equality examined here for $i\geq 1$ is as follows:
\begin{align*}
\expect{Y_\lambda} &= q_\lambda^{-1} \left(q_\lambda\sum_{b_\lambda\leq j\leq B} W(j)\right) \\
&= \sum_{b_\lambda\leq j\leq B} W(j) \\
&= (1-\delta)\delta^{B-b_\lambda} + (1-\delta)\delta^{B-b_\lambda-1} + \cdots + (1-\delta)
\end{align*}

The last expression is the sum of the following geometric sequence from $n=1$ to $n=B-b_\lambda+1$:
$$
a_n = (1-\delta)\delta^{n-1}.
$$

The sum of the sequence is given by
$$
\frac{(1-\delta)\left(1 - \delta^{n}\right)}{1-\delta}.
$$

Putting $n=B-b_\lambda+1$, by definition, we obtain
$$
1 - \delta^{B-b_\lambda+1} = 1 - \phi(b_\lambda).
$$

%If $i=0$,
%$$
%\sum_{b_\lambda\leq j\leq B} W(j) = 1-\delta^{B-1+1} + \delta^{B} = 1.
%$$
%Thus,
%$$
%\expect{Y_\lambda} = 1.
%$$

\subsection{Lemma~\ref{lem:our_concent_prob}}
\label{sec:appx_analysis}

Let $\mu_Y$ and $\sigma_Y$ be the expectation and the variance of \methodcx, respectively.
Also, let $\mu_X$ and $\sigma_X$ be the expectation and the variance of \methodc, respectively.
We know the followings:
\begin{align}
	\mu_Y &= \psi \\
	\sigma_Y^2 &= \psi^2\left(q-1\right) \\
	\mu_X &= 1 \\
	\sigma_X^2
	 &= \left(q-1\right),
\end{align}
where $\psi = 1-\phi(b)$, and $\phi(b) = \delta^{B-i+1}$. Here, we omit the subscript to denote a certain triangle for simplicity. Using the Chebyshev inequality, we obtain
\begin{align}
	\prob{|Y-\mu_Y| \geq \sigma_Y^2} \leq \frac{1}{\sigma_Y^2},
\end{align}
and thus,
\begin{align}
	\prob{|Y-\mu_X| \geq \mu_X-\mu_Y+\sigma_Y^2} &\leq \frac{1}{\sigma_Y^2} \\
	\prob{|Y-\mu_X| < \mu_X-\mu_Y+\sigma_Y^2} &\geq  1 - \frac{1}{\sigma_Y^2}, 
\end{align}
since $\mu_Y \leq \mu_X$. Similarly, we also derive the following inequality for $X$:
\begin{align}
	\prob{|X-\mu_X| \geq \mu_X-\mu_Y+\sigma_Y^2} &\leq \frac{\sigma_X^2}{(\mu_X-\mu_Y+\sigma_Y^2)^2} \\
	\prob{|X-\mu_X| < \mu_X-\mu_Y+\sigma_Y^2} &\geq 1 - \frac{\sigma_X^2}{(\mu_X-\mu_Y+\sigma_Y^2)^2}.
\end{align}

We want the following inequality holds:

\begin{align}
	1 - \frac{1}{\sigma_Y^2} &> 1 - \frac{\sigma_X^2}{(\mu_X-\mu_Y+\sigma_Y^2)^2} \\
	\frac{1}{\sigma_Y^2} &< \frac{\sigma_X^2}{(\mu_X-\mu_Y+\sigma_Y^2)^2} \\
	1 &< \frac{\sigma_Y^2 \sigma_X^2}{(\mu_X-\mu_Y+\sigma_Y^2)^2} \\
	1 &< \frac{\psi^2(q-1)^2}{(1-\psi + \psi^2 (q-1))^2} \\
	1 &< \frac{\psi (q-1) }{1-\psi + \psi^2 (q-1)} \\
	1 &> \frac{1-\psi + \psi^2 (q-1)}{\psi (q-1) } \\
	1 &> \frac{1}{\psi (q-1)} - \frac{\psi}{\psi (q-1)} + \frac{\psi^2 (q-1)}{\psi (q-1)}\\
	1 &> \frac{1}{\psi (q-1)} - \frac{1}{(q-1)} + \psi \\
	1 &> \frac{1}{q-1}\left(\frac{1}{\psi} - 1\right) + \psi \\
	1-\psi &> \frac{1}{q-1}\left(\frac{1}{\psi} - 1\right) \\
	\phi &> \frac{1}{q-1}\left(\frac{1-\psi}{\psi}\right) \\
	\phi &> \frac{1}{q-1}\left(\frac{\phi}{1-\phi}\right) \\
	1-\phi &> \frac{1}{q-1} \label{eq:appx1} \\
	\frac{1}{1-\phi} &< q - 1 \\
	\frac{2-\phi}{1-\phi} &< q \\
	\frac{2-\delta^{B-b_\lambda+1}}{1-\delta^{B-b_\lambda+1}} &< \frac{(T_\lambda-1)(T_\lambda-2)}{M(M-1)}
\end{align}

If we set $\delta=0.5$, 
this lemma states that every triangle appearing after $\approx 1.7M$ results in concentrated estimation.

\end{document}